\documentclass[unsortedaddress,superscriptaddress,pra,notitlepage]{revtex4-1}

\usepackage{amsmath,amssymb}
\usepackage{braket}
\usepackage[newitem,newenum]{paralist}
\usepackage{theorem}
\usepackage{color}
\usepackage[hidelinks]{hyperref}

\newtheorem{definition}{Definition}[section]
\newtheorem{theorem}[definition]{Theorem}
\newtheorem{prop}[definition]{Proposition}
\newtheorem{lemma}[definition]{Lemma}
\newtheorem{remark}[definition]{Remark}
\newtheorem{rem}[definition]{Remark}

\newtheorem{cor}[definition]{Corollary}

\newenvironment{proof}[1][Proof]{\begin{trivlist}
\item[\hskip \labelsep {\bfseries #1}]}{\hfill$\Box$\end{trivlist}}
\newenvironment{proofof}[1][Proof]{\begin{trivlist}
\item[\hskip \labelsep {\bfseries Proof of #1:}]}{\hfill$\Box$\end{trivlist}}
\newcommand{\nn}{\nonumber \\}
\def\dm#1{\(\displaystyle #1 \)}
\newenvironment{textmath}{\(\displaystyle}{\)}

\def\A{\mathcal{A}}
\def\B{\mathcal{L}}
\def\D{\mathcal{D}}
\def\E{\mathcal{E}}

\def\H{\mathcal{H}}
\def\I{\mathcal{I}}

\def\L{\mathcal{L}}
\def\M{\mathcal{M}}

\def\P{\mathcal{P}}

\def\R{\mathbb{R}} 
\def\S{\mathcal{S}}

\def\V{\mathbb{V}}
\def\W{\mathbb{W}}
\def\X{\mathcal{X}}

\def\R{\mathbb{R}}
\def\C{\mathbb{C}}
\def\E{\mathcal{E}}

\def\H{\mathcal{H}}

\def\L{\mathcal{L}}
\def\M{\mathcal{M}}
\def\S{\mathcal{S}}
\def\norm#1{\left\Vert #1 \right\Vert}

\def\vfi{\varphi}
\def\hil{{\mathcal H}}
\def\kil{{\mathcal K}}
\def\A{{\mathcal A}}
\def\B{{\mathcal L}}
\def\C{{\mathcal C}}
\def\I{\mathcal{I}}

\def\M{\mathcal{M}}
\def\P{\mathcal{P}}
\def\S{{\mathcal S}}
\def\X{{\mathcal X}}

\def\half{\frac{1}{2}}
\def\iff{\Longleftrightarrow}
\def\imp{\Longrightarrow}
\def\ep{\varepsilon}
\def\bN{\mathbb{N}}

\def\bR{\mathbb{R}}

\def\bz{\left(}
\def\jz{\right)}
\def\inv{^{-1}}

\def\kiii{}

\def\map{\Phi}
\def\old{}
\def\nw{^{*}}
\def\x{^{(t)}}
\def\xx{(t)}
\def\typ{t}
\def\bog{^{\flat}}
\def\bogg{\flat}
\def\oldd{\{\s\}}
\def\neww{*}

\def\symm{\L_{\mathrm{sym}}}
\def\ssymm{\S_{\mathrm{sym}}}

\def\m{^{(m)}}
\def\rho{\varrho}
\def\hilin{\hil_{\mathrm{in}}}
\def\hilout{\hil_{\mathrm{out}}}
\def\inn{_{\mathrm{in}}}
\def\out{_{\mathrm{out}}}

\def\Sym{\mathfrak{S}} 
\def\SU{\mathrm{SU}} 
\def\End{\L} 
\def\a{\alpha}

\def\norm#1{\Vert #1 \Vert}

\def\ort{^{\perp}}
\def\sa{\mathrm{sa}}

\newcommand{\ki}{\emph}

\newcommand{\s}{\mbox{ }}
\newcommand{\ds}{\mbox{ }\mbox{ }}
\newcommand{\inner}[2]{\langle #1 , #2\rangle}

\newcommand{\diad}[2]{|#1\rangle\langle #2|}
\newcommand{\pr}[1]{\diad{#1}{#1}}

\newcommand{\ext}[1]{\mathbb{#1}}
\newcommand{\sq}[1]{\underline{#1}}

\DeclareMathOperator{\id}{id}
\DeclareMathOperator{\supp}{supp}
\DeclareMathOperator{\ran}{ran}

\DeclareMathOperator{\Tr}{Tr}
\DeclareMathOperator{\logn}{\widehat\log}

\begin{document}

\title{Strong converse exponent for classical-quantum channel coding}

\author{Mil\'an Mosonyi}
\email{milan.mosonyi@gmail.com}

%

\affiliation{
Mathematical Institute, Budapest University of Technology and Economics, \\
Egry J\'ozsef u~1., Budapest, 1111 Hungary.
}

\author{Tomohiro Ogawa}
\email{ogawa@is.uec.ac.jp}
\affiliation{
Graduate School of Information Systems,
University of Electro-Communications,
1-5-1 Chofugaoka, Chofu-shi, Tokyo, 182-8585, Japan.
}

\begin{abstract}
We determine the exact strong converse exponent of classical-quantum channel coding, for every rate above the Holevo capacity.
Our form of the exponent is an
exact analogue of Arimoto's, given as a transform of the R\'enyi capacities with parameters $\alpha>1$. It is important to note that, unlike in the classical case, there are many inequivalent ways to define the R\'enyi divergence of states, and hence the
R\'enyi capacities of channels. Our exponent is in terms of the R\'enyi capacities corresponding to a
version of the R\'enyi divergences that has been introduced recently in
[M\"uller-Lennert, Dupuis, Szehr, Fehr and Tomamichel, \kiii{J.~Math.~Phys.} \textbf{54}, 122203, (2013)], and [Wilde, Winter, Yang, \kiii{Commun.~Math.~Phys.}, \textbf{331}, (2014)]. Our result adds to the growing body of evidence that this new version is the natural definition for the purposes of strong converse problems.
\end{abstract}

\maketitle

\section{Introduction}

Reliable transmission of information through a noisy channel is one of the central problems
in both classical and quantum information theory.
In quantum information theory, a memoryless classical-quantum channel 
is a map that assigns to every input signal from an input alphabet a state of
a quantum system, and repeated use of the channel maps every sequence of input signals into
the tensor product of the output states corresponding to the elements of the sequence. This is a direct analogue of
a memoryless classical channel, where the outputs are probability distributions on some output alphabet; in fact,
classical channels can be seen as a special subclass of classical-quantum channels where all possible output states commute with
each other.

To transmit information through $n$ uses of the channel, the sender and the receiver
have to agree on a code, i.e., an assignment of a sequence of input signals and a
measurement operator on the output system to each possible message, such that the measurement operators form a
valid quantum measurement, normally described by a POVM (positive operator valued measure).
The maximum rate (the logarithm of the number of messages divided by the number of channel uses)
that can be attained by such coding schemes in the asymptotics of large $n$, with an asymptotically vanishing
probability of erroneous decoding, is the capacity of the channel.
The classical-quantum channel coding theorem, due to Holevo \cite{H} and Schumacher and Westmoreland \cite{SW},
identifies this operational notion of capacity with an entropic quantity, called the Holevo capacity, that is
the maximum mutual information in a classical-quantum state between the
input and the output of the channel that can be obtained from some probability distribution on the input through the action of
the channel. This is one of the cornerstones of quantum information theory,
and is a direct analogue of Shannon's classic channel coding theorem, which in turn can be considered as the starting point
of modern information theory.

Clearly, there is a trade-off between the coding rate and the error probability, and the Holevo-Schumacher-Westmoreland (HSW)
theorem identifies a special point on this trade-off curve, marked by the
Holevo capacity of the channel.
The direct part of the theorem \cite{H,SW} states that for any rate below the Holevo capacity, a sequence of codes with
asymptotically vanishing error probability exists, while the converse part (also known as Holevo bound \cite{H73,H79})
says that for any rate above the Holevo capacity, the error probability cannot go to zero.
In fact, more is true: for any rate above the capacity, the error probability inevitably goes to one asymptotically. This is
known as the strong converse theorem, and it is due to Wolfowitz \cite{Wolfowitz57,Wolfowitz78} in the case of classical channels.
The strong converse theorem for classical-quantum channels has been shown indepently by Winter \cite{W} and Ogawa and Nagaoka
\cite{ON99}. Winter's proof follows Wolfowitz's approach, based on the method of types, while the proof of Ogawa and Nagaoka
follows Arimoto's proof for classical channels \cite{A73}.
A much simplifed approach has been found later by Nagaoka \cite{N}, based on the monotonicity of R\'enyi divergences.

Thus, if one plots the optimal asymptotic error against the coding rate, one sees a sharp jump from zero below the Holevo
capacity to one above the Holevo capacity. However, to understand the trade-off between the error and the coding rate, one has to
plot also the error on the logarithmic scale. Indeed, it is known that in the direct domain, i.e., for any rate below the Holevo
capacity, the optimal error probability vanishes with an exponential speed (see, e.g., \cite{Hayashicq} for the
classical-quantum case), and in the converse domain, i.e., for rates above the capacity, the convergence of the
optimal success probability
to zero is also exponential \cite{W,ON99,N}. The value of these exponents as a function of the coding rate gives a
quantification of the trade-off between the error rate and the coding rate. In the direct domain, it is called the error
exponent, and its
value is only known for classical channels and large enough rates.
In the converse domain, it is called the strong converse exponent, and a lower bound on its value has been given in
Arimoto's work \cite{A73}. This was later complemented by
Dueck and K\"orner \cite{DK}, who obtained an upper bound on
the strong converse exponent in the form of a variational expression in terms of the relative entropy.
Despite their very different forms, the bounds of Arimoto and of Dueck and K\"orner turn out to coincide,
and hence together they give the exact strong converse exponent for classical channels.

In this paper we determine the exact strong converse exponent for classical-quantum channels. Our form of the exponent is an
exact analogue of Arimoto's, given as a transform of the R\'enyi capacities with parameters $\alpha>1$.
If $sc(R,W)$ denotes the strong converse exponent of the classical-quantum channel at coding rate $R$, and $\chi_{\alpha}\nw(W)$ is the R\'enyi $\alpha$-capacity of the channel, then our result tells that
\begin{align}\label{main result}
sc(R,W)=\sup_{\alpha>1}\frac{\alpha-1}{\alpha}\left\{R-\chi_{\alpha}\nw(W)\right\}.
\end{align}

It is important to note that, unlike in the classical case, there are many inequivalent ways to define the R\'enyi divergences of
states, and hence the R\'enyi capacities of channels. Based on results in hypothesis testing \cite{MO,CMW,HT14}, it seems that
there are two different families of R\'enyi divergences that appear naturally in the quantification of trade-off relations for
quantum information theoretic problems: one for the direct domains, and another one,
introduced recently in \cite{Renyi_new} and \cite{WWY}, for the converse domains.
We denote these families by $D_{\alpha}\old$ and $D_{\alpha}\nw$, respectively, and give their definitions in
Section \ref{sec:Renyi divergences}. Our expression for the strong converse exponent is in terms of the $D_{\alpha}\nw$
divergences, and hence it gives an operational interpretation to the capacities derived from these divergences.
This shows that the $D_{\alpha}\nw$ divergences are the natural trade-off quantifiers in the converse domain not only
for hypothesis testing but also for classical-quantum channel coding. This in turn provides further evidence to the expectation that the picture should be the same for other - most probably all - coding problems where a direct and a converse domain can be defined.

In classical information theory, the direct and the strong converse exponents are typically expressed in two very different-looking forms: the Gallager- or Arimoto-type exponents, which are in terms of R\'enyi divergences, or by an optimization formula in terms of the Kullback-Leibler divergence (relative entropy), like the Dueck-K\"orner exponent. 
Despite their very different forms,
the bounds of Arimoto and of Dueck and K\"orner turn out to coincide,
which is stated in the paper of Dueck and K\"orner without proof \cite{DK}; see also \cite{Oo15} for an argument.
For an early discussion about the conversion between the two, see e.g., the work of Blahut \cite{Blahut}. Interestingly, the same optimization formulas with the quantum relative entropy result in suboptimal exponents, as has been observed in \cite{ON}. In fact, they give rise to
the same formulas as the classical R\'enyi divergence expressions, but with a family of quantum R\'enyi divergences (denoted by $D_{\alpha}\bog$ in this paper) that is different from both $D_{\alpha}\old$ and $D_{\alpha}\nw$.
One of the main contributions of our paper is the observation that these suboptimal exponents may be converted into the correct
ones by the method of pinching \cite{HP,H:pinching}. Thus, even though the $D_{\alpha}\bog$ divergences are not expected to have a direct operational interpretation like the $D_{\alpha}\old$ and the $D_{\alpha}\nw$ divergences, they turn out to be an important intermediate quantity when extending classical results into the quantum domain.
We remark that these quantities already appeared in the context of matrix analysis in \cite{HP-GT}, and it has been shown in
\cite{AD} that they arise as a limiting case of a two-parameter family of quantum R\'enyi divergences.
However, their properties and their application in quantum information theory have been largely unexplored so far.
For this reason, and because we need many of their mathematical properties to obtain our main result, we give a detailed
exposition of them in Sections \ref{sec:Renyi divergences} and \ref{sec:Renyi capacities}.

The structure of the paper is as follows. In Section \ref{sec:pre}, we summarize the necessary mathematical preliminaries.
Section \ref{sec:Renyi divergences} is devoted to R\'enyi divergences. The new contribution here is the investigation of the properties of the $D_{\alpha}\bog$ divergences. For completeness and for comparison, we state most results for all the three families of R\'enyi divergences mentioned above, but we give most of the proofs only for the $D_{\alpha}\bog$ quantities, as the
properties of the other two families have been investigated in detail elsewhere \cite{Renyi_new,WWY,MO,FL,Beigi}.
In Section \ref{sec:Renyi capacities} we first give an overview of the different notions of mutual information and capacity formulas derived from
R\'enyi divergences, and then, in Section \ref{sec:pinched}, we give one of the main technical contributions of the paper, the connection of the mutual informations of an $n$-fold product channel and its pinched version. This will be the key tool in converting the suboptimal exponent obtained in Section \ref{sec:DueckKorner} into the correct exponent.
The main result of the paper is given in Section \ref{sec:sc}.
After an overview of the problem of classical-quantum channel coding in Section \ref{sec:cl-q sc}, we give a lower bound
on the strong converse exponent in Section \ref{sec:sc lb}. This follows by the monotonicity of the $D_{\alpha}^*$ divergences by
a standard argument due to Nagaoka \cite{Nagaoka}. In Sections \ref{sec:DueckKorner}--\ref{sec:sc achievability} we show that
this lower bound is also an upper bound. We first extend the result of Dueck and K\"orner to classical-quantum channels in
Section \ref{sec:DueckKorner}, and obtain an upper bound on the strong converse exponent in the form of a relative entropy optimization expression, which is then turned into an Arimoto-type expression involving the $D_{\alpha}\bog$ capacities of the channel. We obtain the correct expression for the strong converse exponent in Section \ref{sec:sc achievability} by applying block pinching to the Arimoto-type expression of the preceding section, for increasing blocklengths.

Finally, in Section \ref{sec:qc} we apply our results to classical-quantum channels to obtain the exact strong converse exponent
for some classes of quantum channels studied in \cite{KW,WWY}. In particular, we show that the strong converse exponent is given by \eqref{main result} also for entanglement breaking channels and group covariant channels with additive minimum output
$\alpha$-entropy.

\section{Preliminaries}
\label{sec:pre}

\subsection{Notations and basic lemmas}
\label{sec:notation}

For a finite-dimensional Hilbert space $\hil$, we use the notation
$\L(\hil)$ for the set of linear operators on $\hil$, and we denote by
$\L(\hil)_+$ and $\L(\hil)_{++}$ the set of non-zero positive semidefinite and
positive definite linear operators on $\hil$, respectively.
The set of density operators on $\hil$ is denoted by
$\S(\H)$, i.e.,
\begin{align*}
\S(\H)=\Set{\tau\in\L(\H)_+ |  \Tr\tau=1 },
\end{align*}
and $\S(\hil)_{++}$ stands for the set of invertible density operators.
For any $\rho\in\L(\hil)_+$, we use the notation
\begin{align*}
\S_{\rho}(\H)=\Set{\tau\in\S(\H) |  \supp\tau\le\supp\rho }.
\end{align*}
We use the notation $\SU(\hil)$ for the special unitary group on $\hil$.

For a self-adjoint operator $A\in\L(\hil)$, let $\{A\ge 0\}$ denote the spectral projection of $A$ corresponding to all non-negative eigenvalues. The notations $\{A>0\}$, $\{A\le 0\}$ and $\{A<0\}$ are defined similarly. The positive part $A_+$ of
$A$ is then defined as
\begin{align*}
A_+:=A\{A>0\}.
\end{align*}
It is easy to see that for any $D\in\L(\hil)_+$ such that $D\le I$, we have
\begin{align}\label{pp trace}
\Tr A_+\ge \Tr AD.
\end{align}

We will use the convention that powers of a positive semidefinite operator are only taken on its support and defined to be $0$ on the orthocomplement of its support.
That is, if $a_1,\ldots,a_r$ are the eigenvalues of $A\in\L(\hil)_+$, with corresponding eigenprojections $P_1,\ldots,P_r$, then $A^{p}:=\sum_{i:\,a_i> 0}a_i^p P_i$ for any
$p\in\bR$. In particular, $A^0=\sum_{i:\,a_i> 0}P_i$ is the projection onto the support of $A$.

We denote the natural logarithm by $\log$, and use the standard conventions of information theory
\begin{align}\label{log ext1}
\log 0:=-\infty,\ds\ds\ds\log+\infty:=+\infty.
\end{align}
We will also use the following extension of the logarithm function:
\begin{align}\label{log ext2}
\logn:\,[0,+\infty)\mapsto\bR,\ds\ds\ds \logn(x):=\log x,\ds x\in(0,+\infty),\ds\ds\text{and}\ds\ds \logn 0:=0.
\end{align}
We use $\logn$ to define the logarithm of a positive semidefinite operator only on its support, and to be $0$ on its 
orthocomplement, analogously to the definition of powers above; that is,
$\logn A=\sum_{i:\,a_i>0}\log(a_i)P_i$. Note, however, that numbers can be considered as one-dimensional operators, and 
\eqref{log ext1} gives a different extension of the logarithm function, that we will
also use frequently in later sections, e.g., when expressing the R\'enyi divergences of states with orthogonal or disjoint 
supports (see \eqref{Renyi div def}, Remark \ref{rem:infty div}).

Measurements with finitely many outcomes on a quantum system with Hilbert space $\hil$ can be identified with (completely) positive trace-preserving maps $M:\,\B(\hil)\to\B(\kil)$
 with some finite-dimensional Hilbert space $\kil$, such that
$M(\B(\hil))$ is commutative. We denote the set of all such maps by $\M(\hil)$.

For an operator $\sigma\in\B(\hil)$, we denote by $v(\sigma)$ the number of different eigenvalues of $\sigma$.
If $\sigma$ is self-adjoint with spectral projections $P_1,\ldots,P_r$, then the \ki{pinching} by $\sigma$
is the map $\E_{\sigma}:\,\B(\hil)\to\B(\hil)$, defined as
\begin{align}\label{pinching}
\E_{\sigma}:\,X\mapsto \sum_{i=1}^r P_i X P_i,\ds\ds\ds X\in\B(\hil).
\end{align}
The \ki{pinching inequality} \cite{H:pinching,Hayashibook} tells that if $X$ is positive semidefinite then
\begin{align}\label{pinching inequality}
X\le v(\sigma)\E_{\sigma}(X).
\end{align}
\smallskip

For self-adjoint operators $A,B\in\L(\hil)$, $A\le B$ is understood in the sense of positive semidefinite
(PSD) ordering, i.e., it means that $B-A$ is positive semidefinite.
The following lemma is standard.

\begin{lemma}\label{lem:trace-mono}
Let $f: J \rightarrow \R$ be a monotone function, where $J$ is some interval in $\bR$,
and let $\L(\hil)_{\sa,J}$ be the set of self-adjoint operators with their spectra in $J$.
\begin{align}
&\text{If $f$ is monotone then}\ds A\mapsto\Tr f(A)\ds\text{is monotone on }\L(\hil)_{\sa,J},\label{trace-mono}\\
&\text{if $f$ is convex then}\ds A\mapsto\Tr f(A)\ds\text{is convex on }\L(\hil)_{\sa,J}.\label{trace-convex}
\end{align}
\end{lemma}

We say that a function $f:\,(0,+\infty)\to\bR$ is \ki{operator monotone increasing} if $A,B\in\L(\hil)_+$, $A\ge B$ implies that
$f(A)\ge f(B)$. We say that $f$ is \ki{operator monotone decreasing} if $-f$ is operator monotone increasing. The following lemma is from \cite[Proposition 1.1]{AH}:
\begin{lemma}\label{lemma:AH}
Let $f$ be a nonnegative operator monotone decreasing function on $(0,+\infty)$,
and $\omega$ be a positive linear functional on $\L(\hil)$. Then the functional
\begin{align*}
 A\mapsto \log \omega(f(A))\ds\ds\ds\text{is convex on}\ds\ds\ds \L(\hil)_{++}.
\end{align*}
\end{lemma}

\subsection{Convexity}

We will use the following lemma, and its equivalent version for concavity, without further notice:

\begin{lemma}\label{lemma:inf sup convexity}
Let $X$ be a convex set in a vector space and $Y$ be an arbitrary set, and let $f:\,X\times Y\to\bR$
be such that for every $y\in Y$, $x\mapsto f(x\,,y)$ is convex. Then
\begin{align}\label{sup convexity}
x\mapsto\sup_{y\in Y}f(x,y)\ds\text{is convex}.
\end{align}
If, moreover, $Y$ is also a convex set in a vector space, and $(x,y)\mapsto f(x,y)$ is convex, then
\begin{align}\label{inf convexity}
x\mapsto\inf_{y\in Y}f(x,y)\ds\text{is convex}.
\end{align}
\end{lemma}
\begin{proof}
The assertion in \eqref{sup convexity} is trivial from the definition of convexity. The proof of
\eqref{inf convexity} is also quite straightforward; see, e.g., \cite[Section 3.2.5]{BV}.
\end{proof}

\begin{definition}
A multi-variable function on the product of convex sets is called jointly convex.
\end{definition}

\begin{lemma}\label{lemma:convex transformation}
Let $f:\,J\to\bR$ be a convex function on some interval $J\subseteq(0,+\infty)$. Then for any affine function 
$\vfi:\,J'\to J$ on an interval $J'\subseteq\bR$, the function $t\mapsto \vfi(t)f\bz\frac{1}{\vfi(t)}\jz$ is convex on any subinterval of $J'$ where it is well-defined and $f$ is continuous on $\ran \vfi$ (this is only important if an endpoint of $J$ is in the range of $\vfi$).
\end{lemma}
\begin{proof}
Since $f$ is convex, it can be written as $f(x)=\sup_{i\in\I}\{a_ix+b_i\}$, where $\I$ is some index set, and 
$a_i,b_i\in\bR$. Thus, 
\begin{align*}
\vfi(t)f\bz\frac{1}{\vfi(t)}\jz=\vfi(t)\sup_i\left\{a_i\frac{1}{\vfi(t)}+b_i\right\}
=
\sup_i\left\{a_i+b_i\vfi(t)\right\}
\end{align*}
which, as the supremum of affine functions, is convex.
\end{proof}

\subsection{Minimax theorems}

Let $X,Y$ be non-empty sets and $f:\,X\times Y\to\bR\cup\{-\infty,+\infty\}$ be a function. Minimax theorems provide sufficient conditions under which
\begin{align}\label{minimax statement}
\inf_{x\in X}\sup_{y\in Y}f(x,y)=
\sup_{y\in Y}\inf_{x\in X}f(x,y).
\end{align}

The following minimax theorem is from \cite[Corollary A.2]{MH}.
\begin{lemma}\label{lemma:minimax2}
Let $X$ be a compact topological space, $Y$ be an ordered set, and let $f:\,X\times Y\to \bR\cup\{-\infty,+\infty\}$ be a function. Assume that
\smallskip

\s(i) $f(.\,,\,y)$ is lower semicontinuous for every $y\in Y$ and
\smallskip

(ii) \begin{minipage}[t]{15cm}
$f(x,.)$ is monotonic increasing for every $x\in X$, or
$f(x,.)$ is monotonic decreasing for every $x\in X$.
\end{minipage}
\smallskip

\noindent Then \eqref{minimax statement} holds,
and the infima in \eqref{minimax statement} can be replaced by minima.
\end{lemma}

The following lemma combines a special version of the minimax theorems due to Kneser \cite{Kneser} and Fan \cite{Fan}
(with conditions (i) and (ii)), and Sion's minimax theorem \cite{Sion,Komiya} (with conditions (i') and (ii')). Recall
that a function $f:\,C\to\bR$ on a convex set $C$ is \ki{quasi-convex} if
\begin{align*}
f(tx_1+(1-t)x_2)\le\max\{f(x_1),f(x_2)\},\ds\ds\ds x_1,x_2\in C,\ds t\in(0,1),
\end{align*}
and it is \ki{quasi-concave} if $-f$ is quasi-convex.

\begin{lemma}\label{lemma:KF minimax}
Let $X$ be a compact convex set in a topological vector space $V$ and $Y$ be a convex
subset of a vector space $W$. Let $f:\,X\times Y\to\bR$ be such that
\smallskip

\s(i) $f(x,.)$ is concave on $Y$ for each $x\in X$, and
\smallskip

(ii) $f(.,y)$ is convex and lower semi-continuous  on $X$ for each $y\in Y$.
\smallskip

\noindent or
\smallskip

\s(i') $f(x,.)$ is quasi-concave and upper semi-continuous on $Y$ for each $x\in X$, and
\smallskip

(ii') $f(.,y)$ is quasi-convex and lower semi-continuous  on $X$ for each $y\in Y$.
\smallskip

\noindent Then \eqref{minimax statement} holds,
and the infima in \eqref{minimax statement} can be replaced by minima.
\end{lemma}

\subsection{Universal symmetric states}

For every $n\in\bN$, let $\Sym_n$ denote the symmetric group, i.e., the group of permutations of $n$ elements.
For every finite-dimensional Hilbert space $\hil$, $\Sym_n$ has a natural unitary representation on
$\hil^{\otimes n}$, defined by
\begin{align*}
\pi_{\hil}: \ket{\psi_{1}} \otimes \ket{\psi_{2}} \otimes \dots \otimes \ket{\psi_{n}}
\longmapsto \ket{\psi_{\pi^{-1}(1)}} \otimes \dots \otimes \ket{\psi_{\pi^{-1}(n)}}
\qquad \ket{\psi_i}\in\H,\,\pi\in\Sym_n.
\end{align*}
Let $\symm(\hil^{\otimes n})$ denote the set of symmetric, or permutation-invariant, operators, i.e.,
\begin{align*}
\symm(\hil^{\otimes n}):=\{A\in\L(\hil^{\otimes n}):\,\pi_{\hil}A=A\pi_{\hil}\ds\forall \pi\in\Sym_n\}
=
\{\pi_{\hil}|\,\pi\in\Sym_n\}',
\end{align*}
where for $\A\subset\L(\kil)$, $\A'$ denotes the commutant of $\A$.
Likewise, we denote by $\ssymm(\hil^{\otimes n})$ the set of symmetric states, i.e.,
$\ssymm(\hil^{\otimes n}):=\symm(\hil^{\otimes n})\bigcap\S(\hil^{\otimes n})$.

\begin{lemma}\label{lemma:universal}
For every finite-dimensional Hilbert space $\hil$ and every $n\in\bN$, there exists a symmetric state
$\sigma_{u,n}\in\ssymm(\hil^{\otimes n})\bigcap \ssymm(\hil^{\otimes n})'$
such that
every symmetric state $\omega\in\ssymm(\hil^{\otimes n})$ is dominated as
\begin{align}\label{univ state}
\omega\le v_{n,d}\,\sigma_{u,n},\ds\ds\ds\ds\ds\ds v_{n,d}\le (n+1)^{\frac{(d+2)(d-1)}{2}},
\end{align}
where $d=\dim\hil$. Moreover, the number of different eigenvalues of $\sigma_{u,n}$ is upper bounded by $v_{n,d}$.
We call every such state $\sigma_{u,n}$ a \emph{universal symmetric state}.
\end{lemma}

A construction for a universal symmetric state has been given in \cite{universalcq}, which we briefly review in Appendix \ref{sec:universal} for readers' convenience. See also
\cite[Lemma 1]{HT14} for a different argument for the existence of a universal symmetric state, with
$(n+1)^{d^2-1}$ in place of $(n+1)^{\frac{(d+2)(d-1)}{2}}$ in \eqref{univ state}.

The crucial property for us is that
\begin{align}\label{v limit}
\lim_{n\to+\infty}\frac{1}{n}\log v_{n,d}=0.
\end{align}

\subsection{Classical-quantum channels}
\label{sec:cq channels}

By a \ki{classical-quantum channel} (or \ki{channel}, for short) we mean a map
\begin{align*}
W:\,\X\to\S(\hil),
\end{align*}
where $\X$ is an arbitrary set (called the \ki{input alphabet}),
and $\hil$ is a finite-dimensional Hilbert space. That is, $W$ maps input signals in
$\X$ into quantum states on $\hil$. We denote the set of classical-quantum channels with input space $\X$ and
output Hilbert space $\hil$ by $C(\hil|\X)$.

For every channel $W\in C(\hil|\X)$, we define the lifted channel
\begin{align*}
\ext{W}:\,\X\to\S(\hil_\X\otimes\hil),\ds\ds\ds
\ext{W}(x):=\pr{x}\otimes W(x).
\end{align*}
Here, $\hil_{\X}$ is an auxiliary Hilbert space, and $\{\ket{x}:\,x\in\X\}$ is an orthonormal basis in it.
As a canonical choice, one can use $\hil_{\X}=l^2(\X)$, the $L^2$-space on $\X$ with respect to the counting measure,
and choose $\ket{x}$ to be the characteristic function (indicator function) of the singleton $\{x\}$.
Note that this is well-defined irrespectively of the cardinality of $\X$.

Let $\P_f(\X)$ denote the set of finitely supported probability measures on $\X$.
We identify every $P\in\P_f(\X)$ with the corresponding probability mass function, and hence
write $P(x)$ instead of $P(\{x\})$ for every $x\in\X$.
We can redefine every channel $W$ with input alphabet $\X$ as a channel on the set of Dirac measures
$\{\delta_x:\,x\in\X\}\subset\P_f(\X)$ by defining $W(\delta_x):=W(x)$. $W$ then admits a natural affine extension
to $\P_f(\X)$, given by
\begin{align*}
W(P):=\sum_{x\in\X}P(x)W(x).
\end{align*}
In particular, the extension of the lifted channel $\ext{W}$ outputs classical-quantum states of the form
\begin{align*}
\ext{W}(P)=\sum_{x\in\X}P(x)\pr{x}\otimes W(x).
\end{align*}
Note that the marginals of $\ext{W}(P)$ are
\begin{align}\label{marginals}
\Tr_{\hil}\ext{W}(P)=\sum_{x\in\X}P(x)\pr{x},\ds\ds\ds
\Tr_{\hil_{\X}}\ext{W}(P)=\sum_{x\in\X}P(x)W(x)=W(P).
\end{align}
With a slight abuse of notation, we will also denote $\sum_{x\in\X}P(x)\pr{x}$ by $P$.

The \ki{$n$-fold i.i.d.~extension} of a channel $W:\,\X\to\S(\hil)$ is defined as
$W^{\otimes n}:\,\X^n\to\S(\hil^{\otimes n})$,
\begin{align*}
W^{\otimes n}(\sq{x}):=W(x_1)\otimes\ldots\otimes W(x_n),\ds\ds\ds\ds\ds
\sq{x}=x_1\ldots x_n\in\X^n.
\end{align*}
Given $\X$, we will always choose the auxiliary Hilbert space $\hil_{\X^n}$ to be $\hil_{\X}^{\otimes n}$ and
$\ket{\sq{x}}:=\ket{x_1}\otimes\ldots\otimes\ket{x_n},\,\sq{x}=x_1\ldots x_n\in\X^n$. With this convention, the lifted channel
of $W^{\otimes n}$ is equal to $\ext{W}^{\otimes n}$. Moreover, for every probability distribution $P\in\P_f(\X)$,
\begin{align*}
W^{\otimes n}(P^{\otimes n})=W(P)^{\otimes n}\ds\ds\ds\text{and}\ds\ds\ds
\ext{W}^{\otimes n}(P^{\otimes n})=\ext{W}(P)^{\otimes n},
\end{align*}
where $P^{\otimes n}\in\P(\X^n),\,P^{\otimes n}(\sq{x}):=P(x_1)\cdot\ldots\cdot P(x_n),\,\sq{x}=x_1\ldots x_n\in\X^n$,
denotes the $n$-th i.i.d.~extension of $P$.

\section{Quantum R\'enyi divergences}
\label{sec:Renyi divergences}

\subsection{Definitions and basic properties}
\label{sec: Renyi def}

For classical probability distributions $p,q$ on a finite set $\X$, their R\'enyi divergence with parameter
$\alpha\in[0,+\infty)\setminus\{1\}$ is defined as
\begin{align*}
D_{\alpha}(p\|q):=\frac{1}{\alpha-1}\log Q_{\alpha}(p\|q),\ds\ds\ds
Q_{\alpha}(p\|q):=\sum_{x\in\X}p(x)^{\alpha}q(x)^{1-\alpha},
\end{align*}
when $\supp p\subseteq\supp q$ or $\alpha\in[0,1)$, and it is defined to be $+\infty$ otherwise.
For non-commuting states, various inequivalent generalizations of the R\'enyi divergences have been proposed.
Here we consider
the following quantities, defined for every pair of positive definite operators $\rho,\sigma\in\L(\hil)_{++}$
and every $\alpha\in(0,+\infty)$:
\begin{align}
Q_{\alpha}(\rho\|\sigma)&:=
\Tr \rho^{\alpha}\sigma^{1-\alpha},\label{quasi}\\
Q_{\alpha}\nw(\rho\|\sigma)&:=
\Tr \bz \rho^{\half}\sigma^{\frac{1-\alpha}{\alpha}}\rho^{\half}\jz^{\alpha},\label{sand}\\
Q_{\alpha}\bog(\rho\|\sigma)&:=
\Tr e^{\alpha\log\rho+(1-\alpha)\log\sigma}.\label{exp}
\end{align}

The expression in \eqref{quasi} is a \ki{quantum $f$-divergence}, or \ki{quasi-entropy}, corresponding to the power function
$x^{\alpha}$ \cite{HMPB,P86}. Its concavity \cite{Lieb-convexity} and convexity \cite{Ando} properties are of central importance to quantum information theory \cite{Lieb-Ruskai,NC},
and the corresponding R\'enyi divergences have an operational significance in the direct part of binary quantum state discrimination as quantifiers of the trade-off between
the two types of error probabilities \cite{ANSzV,Hayashicq,Nagaoka}.
The R\'enyi divergences corresponding to \eqref{sand} have been introduced recently in
\cite{Renyi_new} and \cite{WWY}; in the latter paper, they were named ``\ki{sandwiched R\'enyi relative entropy}''.
They have been shown to have an operational significance in the converse part of various discrimination problems as quantifiers of the trade-off between
the type I success and the type II error probability \cite{CMW,HT14,MO,MO-correlated}.
$Q_{\alpha}\bog$ has been studied in information geometry \cite{AN}, and its logarithm appeared in \cite{HP-GT} in connection to the Golden-Thompson inequality. It is the natural quantity appearing in classical divergence-sphere optimization representations
of various information quantities, as pointed out in \cite[Section VI]{ON}, \cite[Section V]{OH}, \cite{Hayashibook} and \cite[Remark 1]{Nagaoka}.
The corresponding R\'enyi divergence was shown to be a limiting case of a two-parameter family of R\'enyi divergences in \cite{AD}.
A closely related quantity appears as a free energy functional in quantum statistical physics \cite{OP}.
Note that for commuting $\rho$ and $\sigma$, all three definitions \eqref{quasi}--\eqref{exp} coincide, and are equal to the classical expression $\sum_{x}\rho(x)^{\alpha}\sigma(x)^{1-\alpha}$, where $\rho(x)$ and $\sigma(x)$ are the diagonal elements
of $\rho$ and $\sigma$, respectively, in a joint eigen-basis.

We extend the above definitions
for general, not necessarily invertible positive semidefinite operators $\rho,\sigma\in\L(\hil)_+$ as
\begin{align}
Q_{\alpha}\x(\rho\|\sigma)&:=\lim_{\ep\searrow 0}Q_{\alpha}\x(\rho+\ep I\|\sigma+\ep I)\label{ext1}\\
&=
\lim_{\ep\searrow 0}Q_{\alpha}\x(\rho+\ep (I-\rho^0)\|\sigma+\ep (I-\sigma^0)).\label{ext2}
\end{align}
Here and henceforth $\xx$ stands for one of the three possible values
$\xx=\oldd,\,\xx=\neww$ or $\xx=\bogg$, where $\oldd$ denotes the empty string, i.e.,
$Q_{\alpha}\x$ with $\xx=\oldd$ is simply $Q_{\alpha}$.

\begin{lemma}\label{lemma:Q def}
For every $\rho,\sigma\in\L(\hil)_+$, and every $\alpha\in(0,+\infty)\setminus\{1\}$, the limits in
\eqref{ext1} and \eqref{ext2} exist and are equal to each other. Moreover,
if $\alpha\in(0,1)$ or $\rho^0\le\sigma^0$,
\begin{align}
Q_{\alpha}(\rho\|\sigma)&=
\Tr \rho^{\alpha}\sigma^{1-\alpha},\label{quasi2}\\
Q_{\alpha}\nw(\rho\|\sigma)&=
\Tr \bz \rho^{\half}\sigma^{\frac{1-\alpha}{\alpha}}\rho^{\half}\jz^{\alpha},\label{sand2}\\
Q_{\alpha}\bog(\rho\|\sigma)&=
\Tr Pe^{\alpha P(\logn\rho)P+(1-\alpha)P(\logn\sigma)P},\label{exp2}
\end{align}
where $P:=\rho^0\wedge\sigma^0$ is the projection onto the intersection of the supports of $\rho$ and $\sigma$,
and for all three values of $\typ$,
\begin{align*}
Q_1\x(\rho\|\sigma)=\Tr\rho,
\end{align*}
and
\begin{align*}
Q_{\alpha}\x(\rho\|\sigma)=+\infty\ds\ds\ds\text{when}\ds\ds\ds
\alpha>1\ds\text{and}\ds \rho^0\nleq\sigma^0.
\end{align*}
In particular, the extension in \eqref{ext1}--\eqref{ext2} is consistent in the sense that for invertible $\rho$ and $\sigma$ we recover the formulas in \eqref{quasi}--\eqref{exp}.
\end{lemma}
\begin{proof}
We only prove the assertions for $Q_{\alpha}\bog$, as the proofs for the other quantities follow similar lines, and are simpler.
For $\alpha\in(0,1)$, \eqref{exp2} has been proved in \cite[Lemma 4.1]{HP-GT}. Next, assume that
$\alpha>1$ and $\rho^0\le\sigma^0$. Then we can assume without loss of generality that $\sigma$ is invertible. Hence,
\begin{align*}
Q_{\alpha}\bog(\rho+\ep (I-\rho^0)\|\sigma+\ep (I-\sigma^0))&=
Q_{\alpha}\bog(\rho+\ep (I-\rho^0)\|\sigma)
=
\Tr\exp\bz \alpha\log(\rho+\ep(I-\rho^0))+(\alpha-1)\log\sigma\inv\jz,
\end{align*}
and applying again \cite[Lemma 4.1]{HP-GT}, we see that the limit as $\ep\searrow 0$ is equal to
\begin{align*}
\Tr P\exp\bz \alpha P(\logn\rho)P+(\alpha-1)P(\logn\sigma\inv)P\jz=
\Tr P\exp\bz \alpha P(\logn\rho)P+(1-\alpha)P(\logn\sigma)P\jz.
\end{align*}
This shows that the limit in \eqref{ext2} exists and is equal to \eqref{exp2}.
Showing that the limit in \eqref{ext1} also exists, and is equal to \eqref{exp2},
follows by a trivial modification.

Hence, we are left to prove the case where $\alpha>1$ and $\rho^0\nleq\sigma^0$.
By the latter assumption, there exists an eigenvector $\psi$ of $\rho$, with eigenvalue $\lambda>0$
such that $c:=\inner{\psi}{(I-\sigma^0)\psi}>0$.
Then we have
\begin{align}
&\Tr\exp\bz\alpha\log(\rho+\ep I)+(1-\alpha)\log(\sigma+\ep I)\jz\nn
&\ds\ge
\inner{\psi}{\exp\bz\alpha\log(\rho+\ep I)+(1-\alpha)\log(\sigma+\ep I)\jz\psi}\nn
&\ds\ge
\exp\bz\alpha\inner{\psi}{\log(\rho+\ep I)\psi}+(1-\alpha)\inner{\psi}{\log(\sigma+\ep I)\psi}\jz\nn
&\ds=
\exp\bz\alpha\log(\lambda+\ep)+(1-\alpha)\inner{\psi}{\logn(\sigma+\ep \sigma^0)\psi}
+(1-\alpha)(\log\ep)\inner{\psi}{(I-\sigma^0)\psi}\jz\nn
&\ds=
\ep^{(1-\alpha)c}(\lambda+\ep)^{\alpha}\exp((1-\alpha)\inner{\psi}{\logn(\sigma+\ep \sigma^0)\psi}),\label{lb4}
\end{align}
where the first inequality is obvious, and the second one is due to the convexity of the exponential function.
The expression in \eqref{lb4} goes to $+\infty$ as $\ep\searrow 0$, and hence
the limit in \eqref{ext2} is equal to $+\infty$, as required. The proof for
the limit in \eqref{ext1} goes the same way.
\end{proof}

\begin{rem}
When $\rho^0\le\sigma^0$, \eqref{exp2} can be written as
\begin{align}
Q_{\alpha}\bog(\rho\|\sigma)&=
\Tr \rho^0 e^{\alpha\logn\rho+(1-\alpha)\rho^0(\logn\sigma)\rho^0}\label{exp4}\\
&=
\Tr e^{\alpha\logn\rho+(1-\alpha)\rho^0(\logn\sigma)\rho^0}-\Tr(I-\rho^0).\label{exp3}
\end{align}
\end{rem}

\bigskip

The \ki{quantum R\'enyi divergences} corresponding to the $Q$ quantities are defined as
\begin{align}\label{Renyi div def}
D_{\alpha}\x(\rho\|\sigma):=\frac{1}{\alpha-1}\log \frac{Q_{\alpha}\x(\rho\|\sigma)}{\Tr\rho}=\frac{1}{\alpha-1}\log Q_{\alpha}\x(\rho\|\sigma)-\frac{1}{\alpha-1}\log\Tr\rho=\frac{\psi_{\alpha}\x(\rho\|\sigma)-\psi_{1}\x(\rho\|\sigma)}{\alpha-1},
\end{align}
for any $\alpha\in(0,+\infty)\setminus\{1\}$, where $\typ$ is any of the three possible values, and we use the notation
\begin{align}
\psi_{\alpha}\x(\rho\|\sigma):=\log Q_{\alpha}\x(\rho\|\sigma),\ds\ds\ds\alpha\in(0,+\infty).
\label{def:psi}
\end{align}
By definition,
\begin{align}\label{psi ep limit}
\psi_{\alpha}\x(\rho\|\sigma)=\lim_{\ep\searrow 0}\psi_{\alpha}\x(\rho+\ep I\|\sigma+\ep I),\ds\ds\ds\alpha\in(0,+\infty).
\end{align}

With these definitions we have
\begin{lemma}
If $\rho^0\le\sigma^0$ then $\psi_{\alpha}\x(\rho\|\sigma)$ is continuous in $\alpha$ on $(0,+\infty)$. If
$\rho^0\nleq\sigma^0$ then $\psi_{\alpha}\x(\rho\|\sigma)$ is continuous in $\alpha$
on $(0,1)$, it has a jump at $1$ as
\begin{align}\label{discont}
\lim_{\alpha\nearrow 1}\psi_{\alpha}\x(\rho\|\sigma)<\log\Tr\rho=\psi_{1}\x(\rho\|\sigma),
\end{align}
and it is $+\infty$ on $(1,+\infty)$.
\end{lemma}
\begin{proof}
The only non-trivial claim is \eqref{discont} for $\typ=\bogg$, which can be seen the following way.
Let $r_{\min}$ denote the smallest positive eigenvalue of $\rho$, let
$\tilde\rho:=\rho/r_{\min}$, let $P:=\rho^0\wedge\sigma^0$ and $P\ort:=I-P$.
If $P=0$ then $\psi_{\alpha}\bog(\rho\|\sigma)=-\infty$ for all $\alpha\in(0,1)$, from which
\eqref{discont} is immediate. Assume now that $0\ne P\ne\rho^0$. Then
\begin{align*}
\lim_{\alpha\nearrow 1}Q_{\alpha}\bog(\rho\|\sigma)&=
\Tr Pe^{P(\logn\rho)P}
=
\Tr Pe^{P(\logn\tilde\rho)P+(\log r_{\min})P}
=
r_{\min}\Tr Pe^{P(\logn\tilde\rho)P}\\
&=
r_{\min}\left[\Tr e^{P(\logn\tilde\rho)P}-\Tr(I-P)\right]
\le
r_{\min}\left[\Tr e^{P(\logn\tilde\rho)P+P\ort(\logn\tilde\rho)P\ort}-\Tr(I-P)\right]\\
&\le
r_{\min}\left[\Tr e^{\logn\tilde\rho}-\Tr(I-P)\right]
=
r_{\min}\left[\Tr\tilde\rho+\Tr(I-\rho^0)-\Tr(I-P)\right]\\
&=
\Tr\rho-r_{\min}\Tr(\rho^0-P)\\
&<
\Tr\rho,
\end{align*}
where the first inequality is due to \eqref{trace-mono}, and the second inequality is due to
\eqref{trace-convex} and the fact that the pinching by $(P,P\ort)$ can be written as a convex combination of unitaries
\cite[Problem II.5.4]{Bhatia}. Taking the logarithm gives \eqref{discont}.
\end{proof}

\begin{rem}\label{rem:infty div}
Note that 
\begin{align*}
\text{for}\ds \alpha>1,\ds D_{\alpha}\x(\rho\|\sigma) =+\infty\ds\iff\ds Q_{\alpha}\x(\rho\|\sigma) =+\infty\ds\iff\ds
\rho^0\nleq\sigma^0.
\end{align*}
for all three values of $\xx$. On the other hand,
\begin{align*}
\text{for}\ds \alpha\in(0,1),\ds D_{\alpha}\x(\rho\|\sigma) =+\infty\ds\iff\ds Q_{\alpha}\x(\rho\|\sigma) =0\ds\iff\ds
\begin{cases}
\rho^0\sigma^0=0,&\xx=\oldd,\,\xx=\neww,\\
\rho^0\wedge\sigma^0=0,&\xx=\bogg.
\end{cases}
\end{align*}
Here, $\rho^0\sigma^0=0$ is equivalent to the supports of $\rho$ and $\sigma$ being orthogonal to each other, while 
$\rho^0\wedge\sigma^0=0$ is equivalent to the supports being disjoint in the sense that $\supp\rho\cap\supp\sigma=\{0\}$.
In all other cases, $D_{\alpha}\x(\rho\|\sigma)$ is finite.
\end{rem}

\bigskip

The \ki{relative entropy} of a pair of positive semidefinite operators $\rho,\sigma\in\L(\hil)_+$
is defined \cite{Umegaki} as
\begin{align}\label{relentr def}
D(\rho\|\sigma):=\Tr\rho(\logn\rho-\logn\sigma)
\end{align}
when $\rho^0\le\sigma^0$, and $+\infty$ otherwise. It is easy to verify that
\begin{align}\label{relentr limit}
D(\rho\|\sigma)=\lim_{\ep\searrow 0}D(\rho+\ep I\|\sigma+\ep I).
\end{align}
For any $\rho\in\L(\hil)_+$, its von Neumann entropy $H(\rho)$ is defined as
\begin{align*}
H(\rho):=-D(\rho\|I)=-\Tr\rho\log\rho.
\end{align*}
The same way as in \cite[Lemma 2.1]{Streater} (see also \cite[Proposition 3]{Petz94}), we see that
\begin{align}\label{Klein}
D(\rho\|\sigma)-\Tr(\rho-\sigma)\ge\frac{1}{2\max\{\norm{\rho},\norm{\sigma}\}}\Tr(\rho-\sigma)^2\ge 0
\end{align}
for any $\rho,\sigma\in\B(\hil)_+$. As it was pointed out in \cite[Lemma 6]{Tropp}, this implies that for any
$A\in\B(\hil)_+$,
\begin{align}\label{Tropp var}
\Tr A=\max_{\tau\in\B(\hil)_+}\{\Tr\tau-D(\tau\|A)\},
\end{align}
and the maximum is attained uniquely at $\tau=A$.
Strictly speaking, \eqref{Klein} and \eqref{Tropp var} were shown in the above references for invertible operators; they can be obtained in the general case by using
\eqref{relentr limit}.

\begin{lemma}\label{lemma:limit at 1}
For every $\rho,\sigma\in\L(\hil)_+$, and all three possible values of $\typ$,
\begin{align}\label{limit 1}
D_{1}\x(\rho\|\sigma):=\lim_{\alpha\to 1}D_{\alpha}\x(\rho\|\sigma)=D_1(\rho\|\sigma):=\frac{1}{\Tr\rho}D(\rho\|\sigma).
\end{align}
\end{lemma}
\begin{proof}
The case $\xx=\oldd$ follows by a straightforward computation, and the case
$\xx=\neww$ have been proved by various methods in \cite{Renyi_new,WWY,M13}. 
Hence, we only have to prove the case $\xx=\bogg$.
Assume first that $\rho^0\le\sigma^0$. Then \eqref{exp4} holds for every $\alpha\in(0,+\infty)$,
and thus
\begin{align*}
\lim_{\alpha\to 1}D_{\alpha}\bog(\rho\|\sigma)&=
\lim_{\alpha\to 1}\frac{\psi_{\alpha}\bog(\rho\|\sigma)-\psi_{1}\bog(\rho\|\sigma)}{\alpha-1}
=
\frac{d}{d\alpha}\Big\vert_{\alpha=1}\psi_{\alpha}\bog(\rho\|\sigma)\nn
&=
\frac{1}{Q_{\alpha}\bog(\rho\|\sigma)}\Tr\rho^0e^{\alpha\logn\rho+(1-\alpha)\rho^0(\logn\sigma)\rho^0}
\left[\logn\rho-\rho^0(\logn\sigma)\rho^0\right]\Bigg\vert_{\alpha=1}\nn
&=
\frac{1}{\Tr\rho}\Tr\rho(\logn\rho-\logn\sigma),
\end{align*}
as required. Now assume that $\rho^0\nleq\sigma^0$. Then $\lim_{\alpha\nearrow 1}\psi_{\alpha}\bog(\rho\|\sigma)<\psi_{1}\bog(\rho\|\sigma)$ by \eqref{discont},
and thus
\begin{align*}
\lim_{\alpha\nearrow 1}D_{\alpha}\bog(\rho\|\sigma)&=
\lim_{\alpha\nearrow 1}\frac{\psi_{\alpha}\bog(\rho\|\sigma)-\psi_{1}\bog(\rho\|\sigma)}{\alpha-1}
=
+\infty=\frac{1}{\Tr\rho}D(\rho\|\sigma),
\end{align*}
while $D_{\alpha}\bog(\rho\|\sigma)=+\infty,\,\alpha>1$, which implies
\begin{align*}
\lim_{\alpha\searrow 1}D_{\alpha}\bog(\rho\|\sigma)=\lim_{\alpha\searrow 1}+\infty=+\infty=\frac{1}{\Tr\rho}D(\rho\|\sigma).
\end{align*}
\end{proof}

By \eqref{psi ep limit} and \eqref{relentr limit} we have, for every $\rho,\sigma\in\L(\hil)_+$ and all three values of $\typ$,
\begin{align}\label{Renyi ep limit}
D_{\alpha}\x(\rho\|\sigma)=\lim_{\ep\searrow 0}D_{\alpha}(\rho+\ep I\|\sigma+\ep I),\ds\ds\ds
\alpha\in(0,+\infty).
\end{align}

\smallskip

The importance of the $\bogg$ quantities stems from the following variational representations in Theorem \ref{thm:e-geodesic}. Let
\begin{align}\label{s alpha}
s(\alpha):=
\begin{cases}
1,& \alpha\ge 1,\\
-1,& \alpha<1.
\end{cases}
\end{align}

\begin{theorem}\label{thm:e-geodesic}
For every $\rho, \sigma\in\L(\H)_+$ such that $P:=\rho^0\wedge\sigma^0\ne 0$, and for every $\alpha\in(0,+\infty)\setminus\{1\}$,
\begin{align}
Q_{\alpha}\bog(\rho\|\sigma)
&=
\max_{\tau\in\B(\hil)_+,\,\tau^0\le\rho^0}\{\Tr\tau-\alpha D(\tau\|\rho)-(1-\alpha)D(\tau\|\sigma)\},\label{Tropp var2}\\
\psi_{\a}\bog(\rho\|\sigma)&=
-\min_{\tau\in\S_{\rho}(\H)}\left\{\alpha D(\tau\|\rho)+(1-\alpha)D(\tau\|\sigma)\right\},
\label{eq:49}\\
D_{\alpha}\bog(\rho\|\sigma)&=
s(\alpha)\max_{\tau\in\S_{\rho}(\hil)}s(\alpha)\left\{D(\tau\|\sigma)-\frac{\a}{\a-1}D(\tau\|\rho)
\right\}.\label{D bog var}
\end{align}
Moreover, \eqref{eq:49}--\eqref{D bog var} are valid even if $\rho^0\wedge\sigma^0= 0$ and $\alpha\in(0,1)$, and
\eqref{Tropp var2}--\eqref{eq:49} hold also for $\alpha=1$.
When $D_{\alpha}\bog(\rho\|\sigma)$ is finite and $\alpha\in(0,+\infty)\setminus\{1\}$, the optima in
\eqref{eq:49}--\eqref{D bog var} are reached at the unique state
\begin{align}\label{Hellinger arc}
\tau_{\a}:=P e^{\a P(\logn\rho)P+(1-\a)P(\logn\sigma)P}/Q_{\alpha}\bog(\rho\|\sigma),
\end{align}
and at $Q_{\alpha}\bog(\rho\|\sigma)\tau_{\alpha}$ in \eqref{Tropp var2}.
\end{theorem}
\begin{proof}
First, note that \eqref{eq:49} and \eqref{D bog var} only differ in a constant multiplier, and hence we only prove
\eqref{eq:49}.
For $\alpha=1$ we have
\begin{align*}
\max_{\tau\in\B(\hil)_+,\,\tau^0\le\rho^0}\{\Tr\tau-D(\tau\|\rho)\}=\Tr\rho=Q_{1}\bog(\rho\|\sigma),
\end{align*}
due to \eqref{Tropp var}, and
\begin{align*}
\min_{\tau\in\S_{\rho}(\H)}\{D(\tau\|\rho)\}=
\min_{\tau\in\S_{\rho}(\H)}\{D(\tau\|\rho/\Tr\rho)-\log\Tr\rho\}=
-\log\Tr\rho=-\psi_{1}\bog(\rho\|\sigma).
\end{align*}
Next, consider the case where $\alpha>1$ and $\rho^0\nleq\sigma^0$. Then
the choice $\tau=\tilde\rho:=\rho/\Tr\rho\in\S_{\rho}(\hil)$ yields
\begin{align*}
\Tr\tau-\alpha D(\tilde\rho\|\rho)-(1-\alpha)D(\tilde\rho\|\sigma)
=
1+\alpha\log\Tr\rho-(1-\alpha)\cdot(+\infty)=+\infty=Q_{\alpha}\bog(\rho\|\sigma),\\
\alpha D(\tilde\rho\|\rho)+(1-\alpha)D(\tilde\rho\|\sigma)
=
-\alpha\log\Tr\rho+(1-\alpha)\cdot(+\infty)=
-\infty=-\psi_{\a}\bog(\rho\|\sigma),
\end{align*}
proving \eqref{Tropp var2}--\eqref{eq:49}.
If $\alpha\in(0,1)$ and $\rho^0\wedge\sigma^0=0$ then for any state $\tau$,
$D(\tau\|\rho)$ or $D(\tau\|\sigma)$ is equal to $+\infty$, and thus
\begin{align*}
\min_{\tau\in\S_{\rho}(\hil)}\{(1-\alpha)D(\tau\|\sigma)+\alpha D(\tau\|\rho)\}=+\infty
=-\psi_{\a}\bog(\rho\|\sigma).
\end{align*}

Hence, for the rest we assume that $P=\rho^0\wedge\sigma^0\ne 0$, and
$\alpha\in(0,1)$ or $\rho^0\le\sigma^0$, in which case
we can use \eqref{exp2}.
Note that if $\rho^0\le\sigma^0$ then $P=\rho^0$, and if $\alpha\in(0,1)$ and $\tau^0\nleq\sigma^0$ then
$(1-\alpha)D(\tau\|\sigma)+\alpha D(\tau\|\rho)=+\infty$. Hence, in both cases the optimization can be restricted to
$\S_P(\hil)$, i.e., we have to prove that
\begin{align}
Q_{\alpha}\bog(\rho\|\sigma)
&=
\max_{\tau\in\B(\hil)_+,\,\tau^0\le P}\{\Tr\tau-\alpha D(\tau\|\rho)-(1-\alpha)D(\tau\|\sigma)\},\label{Tropp var3}\\
\psi_{\a}\bog(\rho\|\sigma)&=
-\min_{\tau\in\S_{P}(\H)}\left\{\alpha D(\tau\|\rho)+(1-\alpha)D(\tau\|\sigma)\right\}.
\label{var3}
\end{align}
For every $\tau\in\B(\hil)_+,\,\tau^0\le P$, let $c(\tau):=\alpha D(\tau\|\rho)+(1-\alpha)D(\tau\|\sigma)$, and
$\tilde\tau:=\tau/\Tr\tau$. Then
\begin{align}
\max_{\tau\in\B(\hil)_+,\,\tau^0\le P}\{\Tr\tau-c(\tau)\}
&=
\max_{\tau\in\B(\hil)_+,\,\tau^0\le P}\{\Tr\tau-(\Tr\tau)\log(\Tr\tau)-(\Tr\tau)c(\tilde\tau)\}\label{Tropp eq1}\\
&=
\max_{\tau\in\S_P(\hil)}\max_{t>0}\{t-t\log t-tc(\tau)\}
\ds=
\max_{\tau\in\S_P(\hil)}\exp\bz- c(\tau)\jz\label{Tropp eq2}\\
&=
\exp\bz-\min_{\tau\in\S_P(\H)}c(\tau)\jz,\label{Tropp eq3}
\end{align}
i.e., \eqref{Tropp var3} and \eqref{var3} are equivalent to each other. A straightforward computation shows that
\eqref{exp2} and \eqref{Tropp var} yield \eqref{Tropp var3}, proving \eqref{Tropp var2}--\eqref{D bog var}.
The assertion about the unique optimizers then follows from the uniqueness of the optimizer in \eqref{Tropp var}.
\end{proof}

\begin{rem}
The family of states in \eqref{Hellinger arc} is a quantum generalization of the Hellinger arc.
\end{rem}

Next, we give a refinement of \eqref{var3}, which also yields an alternative proof of \eqref{var3}.

\begin{prop}\label{prop:var}
Let $\rho, \sigma\in\L(\H)_+$ be such that $P:=\rho^0\wedge\sigma^0\ne 0$, and
let $\tau_{\alpha}$ be as in \eqref{Hellinger arc}.
For every $\tau\in\S_P(\hil)$, and
every $\alpha\in(0,+\infty)\setminus\{1\}$,
\begin{align}\label{var4}
\alpha D(\tau\|\rho)+(1-\alpha)D(\tau\|\sigma)=D(\tau\|\tau_{\alpha})-\psi_{\a}\bog(\rho\|\sigma),
\end{align}
or equivalently,
\begin{align}\label{var5}
D_{\alpha}\bog(\rho\|\sigma)=\frac{\alpha}{1-\alpha}D(\tau\|\rho)+D(\tau\|\sigma)-\frac{1}{1-\alpha}D(\tau\|\tau_{\alpha}).
\end{align}
In particular, \eqref{var3} holds, with $\tau_{\alpha}$ being the unique minimizer, and
\begin{align*}
D_{\alpha}\bog(\rho\|\sigma)=\frac{\alpha}{1-\alpha}D(\tau_{\alpha}\|\rho)+D(\tau_{\alpha}\|\sigma).
\end{align*}
\end{prop}
\begin{proof}
Let $\alpha\in(0,+\infty)\setminus\{1\}$.
The quantum relative entropy admits the following simple identity,
related to the triangular relation
in information geometry \cite[Theorems 3.7, 7.1]{AN}:
for any $r,s\in\S(\H)$ and any $t\in\L(\H)_+$ such that $r^0\le s^0\le t^0$,
\begin{align}
D(r\|t)=D(r\|s)+D(s\|t)+\Tr(r-s)(\logn s -\logn t).
\label{eq:53}
\end{align}
Hence for any $\tau\in\S_P(\H)$, we have
\begin{align*}
D(\tau\|\rho)
&=D(\tau\|\tau_{\a})+D(\tau_{\a}\|\rho)+\Tr(\tau-\tau_{\a})(\logn\tau_{\a}-\logn\rho) \nn
&=D(\tau\|\tau_{\a})+D(\tau_{\a}\|\rho)+(\a-1)\Tr(\tau-\tau_{\a})(\logn\rho-\logn\sigma),
\\
D(\tau\|\sigma)
&=D(\tau\|\tau_{\a})+D(\tau_{\a}\|\sigma)+\Tr(\tau-\tau_{\a})(\logn\tau_{\a}-\logn\sigma) \nn
&=D(\tau\|\tau_{\a})+D(\tau_{\a}\|\sigma)+\a\Tr(\tau-\tau_{\a})(\logn\rho-\logn\sigma).
\end{align*}
Combining these relations, it holds for any $\tau\in\S_P(\H)$ that
\begin{align}
\alpha D(\tau\|\rho)+(1-\alpha)D(\tau\|\sigma)
&=D(\tau\|\tau_{\a})+\alpha D(\tau_{\a}\|\rho)+(1-\alpha)D(\tau_{\a}\|\sigma) \label{eq:55}
\end{align}
By the definition of $\tau_{\alpha}$,
\begin{align}
\logn\tau_{\a}&=
\a P(\logn\rho)P+(1-\a)P(\logn\sigma)P-\psi_{\a}\bog(\rho\|\sigma),
\end{align}
and thus
\begin{align*}
D(\tau_{\a}\|\rho)&=\Tr\tau_{\alpha}(\logn\tau_{\alpha}-\logn\rho)
=
(\a-1)\Tr\tau_{\a}(\logn\rho-\logn\sigma)-\psi_{\a}\bog(\rho\|\sigma), \\
D(\tau_{\a}\|\sigma)&=\Tr\tau_{\alpha}(\logn\tau_{\alpha}-\logn\sigma)
=\a\Tr\tau_{\a}(\logn\rho-\logn\sigma)-\psi_{\a}\bog(\rho\|\sigma).
\end{align*}
Hence,
\begin{align}
\alpha D(\tau_{\a}\|\rho)+(1-\alpha)D(\tau_{\a}\|\sigma)=-\psi_{\a}\bog(\rho\|\sigma).
\label{eq:54}
\end{align}
Combining this with \eqref{eq:55} yields \eqref{var4}. By the strict positivity of the relative entropy,
\eqref{var4} yields \eqref{var3} and that $\tau_{\alpha}$ is the unique minimizer.
\end{proof}

\begin{remark}
The following variational formulas were shown in \cite{Tropp} and \cite{HP-GT}, respectively:
For any self-adjoint operator $H$, and any positive definite operator $A$,
\begin{align}
\Tr e^{H+\log A}&=\max_{\tau\in\B(\hil)_{++}}\{\Tr\tau+\Tr\tau H-D(\tau\|A)\},\label{Tropp var4}\\
\log\Tr e^{H+\log A}&=\max_{\tau\in\S(\hil)_{++}}\{\Tr\tau H-D(\tau\|A)\}.\label{HP var}
\end{align}
With the substitution $H:=\alpha\log\rho,\,A:=\sigma^{1-\alpha}$, we can recover \eqref{Tropp var2}--\eqref{eq:49}
for invertible $\rho$ and $\sigma$. What is new in Theorem \ref{thm:e-geodesic}, apart from extending the variational representations
for non-invertible $\rho$ and $\sigma$, is making the connection between the variational expressions
\eqref{Tropp var2} and \eqref{eq:49} (equivalently, between \eqref{Tropp var4} and \eqref{HP var}) in \eqref{Tropp eq1}--\eqref{Tropp eq3}. This shows that proving either of \eqref{Tropp var2} or \eqref{eq:49} yields immediately the other
variational expression as well. In the proof of Theorem \ref{thm:e-geodesic} we followed Tropp's argument \cite{Tropp} based on
\eqref{Tropp var} to obtain \eqref{Tropp var2}, and from it \eqref{eq:49}. The proof based on Proposition \ref{prop:var} proceeds the other way around: we first prove \eqref{eq:49}, which then yields \eqref{Tropp var2}.
Note that this alternative proof gives a new
proof of \eqref{Tropp var4} and \eqref{HP var} through the choice $\rho:=e^{H/\alpha},\,\sigma:=A^{\frac{1}{1-\alpha}}$.
\end{remark}

\begin{rem}
For classical random variables (corresponding to commuting density operators), the
expression \eqref{var5} seems to have first appeared in \cite{CsM2003}. This yields
the variational expressions \eqref{eq:49} and \eqref{D bog var} for classical random variables;
an alternative proof for these have appeared in \cite{Shayevitz}.
\end{rem}

Most of the relevant properties of $D_{\alpha}\bog$ can be derived from the variational formula in Theorem \ref{thm:e-geodesic}.
The following Lemma has the same importance for $D_{\alpha}\nw$:
\begin{lemma}\label{lemma:attainability}
For any $\rho,\sigma\in\L(\hil)_+$, we have
\begin{align}
D_{\a}\nw(\rho\|\sigma)
&=\lim_{n\to\infty}\frac{1}{n}D_{\a}(\E_{\sigma^{\otimes n}}\rho^{\otimes n}\|\sigma^{\otimes n}),
&\alpha\in(0,+\infty),
\label{att by pinching}\\
D_{\a}\nw(\rho\|\sigma)&=\lim_{n\to\infty}\frac{1}{n}\max_{M_n\in\M(\hil^{\otimes n})}D_{\a}\bz M_n(\rho^{\otimes n})\|M_n(\sigma^{\otimes n})\jz,
&\alpha\in[1/2,+\infty),
\label{Renyi attainability}
\end{align}
where $\E_{\sigma^{\otimes n}}$ is the pinching \eqref{pinching} by $\sigma^{\otimes n}$, and
the maximization in the second line is over finite-outcome measurements on $\H^{\otimes n}$ (see section \ref{sec:notation}.)
\end{lemma}
Both \eqref{att by pinching} and \eqref{Renyi attainability} tells that $D_{\alpha}\nw$ can be recovered as the limit of the
R\'enyi divergences of commuting operators. The first identity \eqref{att by pinching} was proved in \cite[Corollary III.8]{MO}
for $\alpha\in(1,+\infty)$, and later extended to $\alpha\in(0,+\infty)$ in \cite[Corollary 3]{HT14}.
The second identity \eqref{Renyi attainability} tells that $D_{\alpha}\nw$ can be recovered as the largest post-measurement
R\'enyi divergence in the asymptotics of many copies, and it follows immediately from \eqref{att by pinching} and the
monotonicity of $D_{\alpha}\nw$ under measurements for $\alpha\in[1/2,+\infty)$ \cite{FL}.

\begin{lemma}\label{lemma: conv mon}
Let $\rho,\sigma\in\L(\hil)_+$, and $\typ$ be any of the three possible values. Then the functions
\begin{align}\label{conv in alpha}
\alpha\mapsto \psi_{\alpha}\x(\rho\|\sigma)\ds\ds\text{and}\ds\ds
\alpha\mapsto Q_{\alpha}\x(\rho\|\sigma)
\ds\ds\text{are convex on }\ds\ds (0,+\infty),
\end{align}
and the function
\begin{align}\label{mon in alpha}
\alpha\mapsto D_{\alpha}\x(\rho\|\sigma)\ds\text{is monotone increasing on }\ds (0,+\infty).
\end{align}
\end{lemma}
\begin{proof}
It is enough to prove \eqref{conv in alpha} for invertible $\rho$ and $\sigma$ due to \eqref{psi ep limit} and the fact that the limit of convex functions is convex.
The second derivative of $\alpha\mapsto \psi_{\alpha}(\rho\|\sigma)$ can be seen to be non-negative
by a straightforward computation, proving
\eqref{conv in alpha} for $\xx=\oldd$.
Combining this with \eqref{att by pinching} shows that $\alpha\mapsto \psi_{\alpha}\nw(\rho\|\sigma)$
is the limit of convex functions, and hence is itself convex.
For $\xx=\bogg$, \eqref{eq:49} yields
\begin{align}
\psi_{\alpha}\bog(\rho\|\sigma)=
\sup_{\tau\in\S(\hil)}\left\{-\alpha D(\tau\|\rho)-(1-\alpha)D(\tau\|\sigma)\right\}.
\end{align}
Thus, $\alpha\mapsto \psi_{\alpha}\bog(\rho\|\sigma)$ is the supremum of convex functions in $\alpha$, and hence is itself convex.
Since $Q_{\alpha}\x(\rho\|\sigma)=\exp(\psi_{\alpha}\x(\rho\|\sigma))$, and the exponential function is monotone increasing and convex,
convexity of $\alpha\mapsto Q_{\alpha}\x(\rho\|\sigma)$ follows from the above.
Since
\begin{align*}
D_{\alpha}\x(\rho\|\sigma)=\frac{\psi_{\alpha}\x(\rho\|\sigma)-\psi_{1}\x(\rho\|\sigma)}{\alpha-1},
\end{align*}
\eqref{mon in alpha} follows immediately from the convexity of $\alpha\mapsto \psi_{\alpha}\x(\rho\|\sigma)$.
\end{proof}
\begin{rem}
Monotonicity of $\alpha\mapsto D_{\alpha}\nw(\rho\|\sigma)$ has been shown in \cite[Theorem 7]{Renyi_new} by a different method.
\end{rem}

\begin{rem}
Convexity of $\alpha\mapsto \psi_{\alpha}\x$ easily yields the concavity of the so-called auxiliary function. We comment on it in 
more detail in Appendix \ref{sec:aux}.
\end{rem}

Monotonicity in $\alpha$ ensures that the limits
\begin{align}
D_{0}\x(\rho\|\sigma)&:=\lim_{\alpha\searrow 0}D_{\alpha}\x(\rho\|\sigma)=\inf_{\alpha\in(0,+\infty)}D_{\alpha}\x(\rho\|\sigma),
\label{0 infty limits1}\\
D_{\infty}\x(\rho\|\sigma)&:=\lim_{\alpha\to+\infty}D_{\alpha}\x(\rho\|\sigma)=\sup_{\alpha\in(0,+\infty)}D_{\alpha}\x(\rho\|\sigma)
\label{0 infty limits2}
\end{align}
exist.
For $\alpha=0$, a straightforward computation verifies that
\begin{align*}
D_0(\rho\|\sigma)=\log\Tr\rho-\log\Tr\rho^0\sigma\ds\ds\ds\text{and}\ds\ds\ds
D_0\bog(\rho\|\sigma)=\log\Tr\rho-\log\Tr Pe^{P(\logn\sigma)P},
\end{align*}
where $P=\rho^0\wedge\sigma^0$. For $\xx=\neww$, a procedure to compute $D_0\nw(\rho\|\sigma)$ was given in \cite[Section 5]{AD} for
the case $\rho^0\le\sigma^0$.

For $\alpha=+\infty$, we get
\begin{align}
D_{\infty}(\rho\|\sigma)&=\log\max\left\{\frac{r}{s}:\,\Tr P_r Q_s>0\right\},\label{old infty}\\
D_{\infty}\nw(\rho\|\sigma)&=D_{\max}(\rho\|\sigma):=\log\inf\left\{\lambda:\,\rho\le \lambda\sigma\right\},\label{new infty}\\
D_{\infty}\bog(\rho\|\sigma)&=\log\inf\{\lambda:\,\logn\rho\le \rho^0(\logn (\lambda\sigma))\rho^0\}\label{bog infty}
\end{align}
when $\rho^0\le\sigma^0$, and $D_{\infty}\x(\rho\|\sigma)=+\infty$ otherwise.
In \eqref{old infty}, $P_r$ and $Q_s$ denote the spectral projections of $\rho$ and $\sigma$, corresponding to
the eigenvalues $r$ and $s$, respectively,
and the equality follows by a straightforward computation. In \eqref{new infty}, $D_{\max}$ is the max-relative entropy \cite{RennerPhD,Datta}, and the equality has been shown in
\cite[Theorem 5]{Renyi_new}.
The case $\xx=\bogg$ follows from Theorem \ref{thm:e-geodesic}, as
when $\rho^0\le\sigma^0$,
\begin{align}
D_{\infty}\bog(\rho\|\sigma)&=
\sup_{\alpha>1}\sup_{\tau\in\S_{\rho}(\H)}\left\{D(\tau\|\sigma)-\frac{\a}{\a-1}D(\tau\|\rho)\right\}\nn
&=
\sup_{\tau\in\S_{\rho}(\hil)}\sup_{\alpha>1}\left\{D(\tau\|\sigma)-\frac{\a}{\a-1}D(\tau\|\rho)\right\}\nn
&=
\sup_{\tau\in\S_{\rho}(\hil)}\left\{D(\tau\|\sigma)-D(\tau\|\rho)\right\}\label{infty variational}\\
&=
\sup_{\tau\in\S_{\rho}(\hil)}\left\{\Tr\tau\bz\logn\rho-\logn\sigma\jz\right\}\nn
&=
\inf\{\kappa:\,\logn\rho-\rho^0(\logn\sigma)\rho^0\le\kappa\rho^0\}\nn
&=
\inf\{\kappa:\,\logn\rho\le \rho^0(\logn (e^{\kappa}\sigma))\rho^0\}.\nonumber
\end{align}
Note that \eqref{infty variational} is an extension of \eqref{D bog var} to $\alpha=+\infty$.
\medskip

Lemmas \ref{lemma: conv mon} and \ref{lemma:limit at 1} with the definitions \eqref{0 infty limits1}--\eqref{0 infty limits2} yield the following
\begin{cor}
Let $\rho,\sigma\in\L(\hil)_+$, and $\typ$ be any of the three possible values. Then the function
\begin{align*}
\alpha\mapsto D_{\alpha}\x(\rho\|\sigma)\in[0,+\infty]\ds\text{is continuous on }\ds [0,+\infty].
\end{align*}
\end{cor}

\subsection{Convexity and monotonicity}

It is easy to see that when $\rho$ and $\sigma$ commute, all the
quantum R\'enyi divergences $D_{\alpha}\x$ with $\xx=\oldd,\,\xx=\neww$ and $\xx=\bogg$ coincide, and are equal to the classical R\'enyi divergence of the eigenvalues of $\rho$ and
$\sigma$. The properties of the classical R\'enyi divergences (monotonicity under stochastic maps, joint convexity, etc.) are very well understood, and are fairly easy to prove.
Corresponding properties of the various quantum generalizations are typically much harder to verify and need not hold for every value of $\typ$ and $\alpha$.
The cases $\xx=\oldd$ and $\xx=\neww$ are by now quite well understood, too; see, e.g., \cite{HMPB,MH,P86} for $\xx=\oldd$ and the recent papers
\cite{Renyi_new,WWY,Beigi,FL,MO,MO-correlated,HT14} for $\xx=\neww$. Hence, we will focus on the so far less studied $D_{\alpha}\bog$ below, and prove most of the claims only for this version, but state the various properties for all three values of $\typ$ for completeness and for comparison.

We say that $D_{\alpha}\x$ is monotone for a fixed $\alpha$, if for all finite-dimensional Hilbert spaces $\hil,\kil$,
every $\rho,\sigma\in\B(\hil)_+$ and every linear completely positive trace-preserving map
$\map:\,\B(\hil)\to\B(\kil)$, we have
$D_{\alpha}\x(\map(\rho)\|\map(\sigma))\le D_{\alpha}\x(\rho\|\sigma)$.
Monotonicity of $s(\alpha)Q_{\alpha}\x$ is defined
in the same way, and it is clear that
for a fixed pair $(\xx,\alpha)$, $D_{\alpha}\x$ is monotone if and only if $s(\alpha)Q_{\alpha}\x$ is monotone.
Similarly, we say that $s(\alpha)Q_{\alpha}\x$ is
jointly convex, if for every finite-dimensional Hilbert space $\hil$, the map
$(\rho,\sigma)\mapsto s(\alpha)Q_{\alpha}\x(\rho\|\sigma)$ is convex on $\B(\hil)_+\times\B(\hil)_+$, where
$s(\alpha)$ is given by \eqref{s alpha}.
By a standard argument,
for any $\alpha\in(0,+\infty)\setminus\{1\}$ and any value of $\xx$,
monotonicity of $D_{\alpha}\x$ is equivalent to the
joint convexity of $s(\alpha)Q_{\alpha}\x$.
We have the following:

\begin{theorem}\label{thm:monotonicity}
The maximal interval of $\alpha$ for which $s(\alpha)Q_{\alpha}\x$ is jointly convex 
is
\begin{align*}
[0,2]\ds\text{for}\ds \xx=\oldd,\ds\ds\ds\ds
[1/2,+\infty)\ds\text{for}\ds \xx=\neww,\ds\ds\text{and}\ds\ds
[0,1]\ds\text{for}\ds \xx=\bogg.
\end{align*}
These are also the maximal intervals in $\bR_+$ for which $D_{\alpha}\x$ is monotone,
except for $D_{\alpha}\nw$, which is monotone also for $\alpha=+\infty$.
\end{theorem}

The case $\xx=\oldd$ was proved in \cite{Lieb-convexity,Ando}; see also \cite{P86}. The case
$\xx=\neww$ was proved in \cite{FL}; see also \cite{Renyi_new,WWY} ($\alpha\in(1,2]$) and
\cite{MO,Beigi} ($\alpha>1$). Either of these cases yield the monotonicity of the relative entropy under CPTP maps ($\alpha=1)$,
which is again equivalent to its joint convexity.
Joint convexity for $\xx=\bogg$ and $\alpha\in(0,1)$ follows immediately from \eqref{Tropp var2} and the joint convexity of the relative entropy.
An alternative proof can be obtained from (i) of \cite[Theorem 1.1]{Hiai-convexity} by
taking $A=\rho,\,B=\sigma$, $\map=\Psi=\id$, $p=\alpha/z,\,q=(1-\alpha)/z$, taking the limit $z\to+\infty$, and using
(27) from \cite{AD}.
Failure of joint convexity for $\xx=\oldd$ and $\alpha>2$
was pointed out in \cite{Renyi_new}; see also \cite[Appendix A]{MO}.
Failure of joint convexity for $\xx=\neww$ and $\alpha<1/2$, was also pointed out in
\cite{Renyi_new}, based on numerical counterexamples; an analytic proof was given in
\cite{BFT_variational}.
We are only left to prove the failure of joint convexity of $Q_{\alpha}\bog$ (monotonicity of $D_{\alpha}\bog$) for $\alpha>1$:
\begin{lemma}
$Q_{\alpha}\bog$ is not monotone under CPTP maps for any $\alpha>1$. In fact, it is not even monotone under pinching
by the reference operator; that is, for every $\alpha>1$, there exist $\rho,\sigma\in\L(\hil)_+$ such that
\begin{align*}
Q_{\alpha}\bog(\rho\|\sigma)<Q_{\alpha}\bog(\E_{\sigma}\rho\|\sigma).
\end{align*}
As a consequence, $Q_{\alpha}\bog$ is not jointly convex for $\alpha>1$.
\end{lemma}
\begin{proof}
Let $\rho:=\half\begin{pmatrix}1 & 1 \\ 1 & 1\end{pmatrix}$ and
$\sigma:=\begin{pmatrix} a & 0 \\ 0 & b\end{pmatrix}$, $a,b>0$. A straightforward computation shows that
\begin{align*}
Q_{\alpha}\bog(\rho\|\sigma)=(ab)^{\frac{1-\alpha}{2}}.
\end{align*}
If we take $a\ne b$ then
$\E_{\sigma}\rho=\half I$, and
\begin{align*}
Q_{\alpha}\bog(\E_{\sigma}\rho\|\sigma)=\frac{a^{1-\alpha}+b^{1-\alpha}}{2^{\alpha}}
\end{align*}
for every $\alpha>1$. Thus, our aim is to find $a$ and $b$ such that
\begin{align*}
\sqrt{a^{1-\alpha}b^{1-\alpha}}<\frac{a^{1-\alpha}+b^{1-\alpha}}{2^{\alpha}},
\ds\ds\text{or equivalently},\ds\ds
c\le\frac{1+c^2}{2^{\alpha}},
\ds\ds\text{where}\ds\ds c:=\sqrt{b/a}^{1-\alpha}.
\end{align*}
It is easy to see that this latter inequality has positive solutions, providing examples such that
$Q_{\alpha}\bog(\rho\|\sigma)<Q_{\alpha}\bog(\E_{\sigma}\rho\|\sigma)$.
\end{proof}

Since the $\log$ is concave and increasing,
$s(\alpha)\psi_{\alpha}\x$ is jointly convex for any $\alpha\in[0,1)$
for which $s(\alpha)Q_{\alpha}\x$ is jointly convex. On the other hand, it is well-known and easy to verify that
$s(\alpha)\psi_{\alpha}\x$ is not jointly convex for $\alpha>1$ even for classical probability distributions.
However, we have the following partial convexity properties:

\begin{prop}\label{prop:log convexity}
For every $\rho\in\L(\hil)_+$, 
the functions
\begin{align*}
\sigma\mapsto s(\alpha)Q_{\alpha}\x(\rho\|\sigma),\ds\ds\ds
\sigma\mapsto s(\alpha)\psi_{\alpha}\x(\rho\|\sigma),\ds\ds\ds
\sigma\mapsto D_{\alpha}\x(\rho\|\sigma)
\end{align*}
are convex on $\L(\hil)_+$ for $\xx=\bogg$ and $\alpha\in[0,+\infty)$,
for $\xx=\neww$ and $\alpha\in[1/2,+\infty)$, and for $\xx=\oldd$ and $\alpha\in[0,2]$.
Moreover, $\sigma\mapsto D_{\infty}\bog(\rho\|\sigma)$ and $\sigma\mapsto D_{\infty}\nw(\rho\|\sigma)$ are also convex.
\end{prop}
\begin{proof}
When $\alpha<1$, the assertions follow immediately from Theorem \ref{thm:monotonicity}.
The assertions about $Q_{1}\x$ and $\psi_1\x$ are obvious from their definitions, and
the assertion about $D_1\x$ follows from the joint convexity of the relative entropy.
Hence for the rest we assume that $\alpha>1$.
Note that convexity of $\psi_{\alpha}\x(\rho\|.)$ is equivalent to the log-convexity of
$Q_{\alpha}\x(\rho\|.)$, which is stronger than convexity.
Moreover, since $\psi_{\alpha}\x(\rho\|.)$ and $D_{\alpha}\x(\rho\|.)$ only differ in a positive constant multiplier,
it is enough to show the convexity of $D_{\alpha}\x(\rho\|.)$. Thus, we have to show that
for any $\rho,\sigma_1,\sigma_2\in\L(\hil)_+$ and any $\lambda\in(0,1)$,
\begin{align}\label{log convexity}
D_{\alpha}\x(\rho\|(1-\lambda)\sigma_1+\lambda\sigma_2)\le
(1-\lambda)D_{\alpha}\x(\rho\|\sigma_1)+\lambda D_{\alpha}\x(\rho\|\sigma_2).
\end{align}
By \eqref{Renyi ep limit}, we may and will assume that $\rho,\sigma_1,\sigma_2$ are all invertible.

We first consider the case $\xx=\bogg$.
By Theorem \ref{thm:e-geodesic},
\begin{align}
&D_{\alpha}\bog(\rho\|(1-\lambda)\sigma_1+\lambda\sigma_2)\nn
&\ds=
\sup_{\tau\in\S(\hil)}\left\{D(\tau\|(1-\lambda)\sigma_1+\lambda\sigma_2)-\frac{\a}{\a-1}D(\tau\|\rho)\right\}\nn
&\ds\le
\sup_{\tau\in\S(\hil)}\left\{(1-\lambda)D(\tau\|\sigma_1)+\lambda D(\tau\|\sigma_2)-\frac{\a}{\a-1}D(\tau\|\rho)\right\}\nn
&\ds\le
\sup_{\tau_1,\tau_2\in\S(\H)}
\left\{(1-\lambda)\bz D(\tau_1\|\sigma_1)-\frac{\a}{\a-1}D(\tau_1\|\rho)\jz+\lambda\bz D(\tau_2\|\sigma_2)-\frac{\a}{\a-1}D(\tau_2\|\rho)\jz\right\}\nn
&\ds=
(1-\lambda)D_{\alpha}\bog(\rho\|\sigma_1)+\lambda D_{\alpha}\bog(\rho\|\sigma_2),
\end{align}
where the first inequality is due to the convexity of the relative entropy in its second argument, the second inequality is obvious, and the last line is again due to Theorem \ref{thm:e-geodesic}.
Convexity of $\sigma\mapsto D_{\infty}\bog(\rho\|\sigma)$ follows by taking the limit $\alpha\to+\infty$.

Next, we consider $\xx=\neww$. By \cite[Lemma 12]{Renyi_new}, we have
\begin{align*}
D_{\alpha}\nw(\rho\|\sigma)=\frac{\alpha}{\alpha-1}\sup_{\tau\in\L(\hil)_+,\,\Tr\tau\le 1}\log\omega_{\tau}(f_\alpha(\sigma)),
\end{align*}
where
$f_{\alpha}(x)=x^{(1-\alpha)/\alpha},\,x\in(0,+\infty)$, and $\omega_{\tau}(X):=\Tr(X\rho^{1/2}\tau^{(\alpha-1)/\alpha}\rho^{1/2}),\,X\in\L(\hil)$. Obviously, $\omega_{\tau}$ is a positive linear functional, and for $\alpha>1$,
$f_{\alpha}$ is operator monotone decreasing \cite{Bhatia}. Hence, by Lemma \ref{lemma:AH},
$\sigma\mapsto D_{\alpha}\nw(\rho\|\sigma)$ is the supremum of convex functions, and hence is itself convex.

The case $\xx=\oldd$ follows by a similar argument; see \cite[Theorem II.1]{MH} for details.
\end{proof}

\subsection{Further properties and relations}

Here we establish some further properties of the R\'enyi divergences that are going to be useful later in the paper.
The following Lemma is easy to verify:

\begin{lemma}\label{lemma:basic properties}
Let $\rho,\sigma\in\L(\hil)_+$, $\alpha\in(0,+\infty]\setminus\{1\}$, and $\typ$ be any of the three possible values.
\smallskip

\noindent 1. The $Q$ quantities are multiplicative, and hence the corresponding R\'enyi divergences are additive in the
sense that for every $n\in\bN$,
\begin{align}\label{additivity}
Q_{\alpha}\x(\rho^{\otimes n}\|\sigma^{\otimes n})=Q_{\alpha}\x(\rho\|\sigma)^n,\ds\ds\ds\ds
D_{\alpha}\x(\rho^{\otimes n}\|\sigma^{\otimes n})=nD_{\alpha}(\rho\|\sigma).
\end{align}

\noindent 2. For every $\lambda>0$,
\begin{align*}
Q_{\alpha}\x(\lambda\rho\|\sigma)&=\lambda^{\alpha}Q_{\alpha}\x(\rho\|\sigma),
& &D_{\alpha}\x(\lambda\rho\|\sigma)=D_{\alpha}\x(\rho\|\sigma)+\log\lambda,\\
Q_{\alpha}\x(\rho\|\lambda\sigma)&=\lambda^{1-\alpha}Q_{\alpha}\x(\rho\|\sigma),
& &D_{\alpha}\x(\rho\|\lambda\sigma)=D_{\alpha}\x(\rho\|\sigma)-\log\lambda.
\end{align*}
\end{lemma}
\medskip

We have the following ordering of the R\'enyi divergences:
\begin{prop}\label{lemma:Q ordering}
For any $\rho,\sigma\in\L(\hil)_+$, we have
\begin{align}
&D_{\alpha}\nw(\rho\|\sigma)\le D_{\alpha}(\rho\|\sigma)\le D_{\alpha}\bog(\rho\|\sigma),&\alpha\in[0,1),\label{D ordering 1}\\
&D_{\alpha}\bog(\rho\|\sigma)\le D_{\alpha}\nw(\rho\|\sigma)\le D_{\alpha}(\rho\|\sigma),&\alpha\in(1,+\infty].\label{D ordering}
\end{align}
\end{prop}
\begin{proof}
It is enough to prove the inequalities for positive definite $\rho$ and $\sigma$, as the general case then follows by \eqref{ext1}.
The inequality $D_{\alpha}\nw(\rho\|\sigma)\le D_{\alpha}(\rho\|\sigma)$ is equivalent to the  Araki-Lieb-Thirring inequality \cite{Araki,LT}.
By the Golden-Thompson inequality \cite{Golden,Symanzik,Thompson}, $\Tr e^{A+B}\le\Tr e^A e^B$ for any self-adjoint $A,B$. This yields that
$D_{\alpha}(\rho\|\sigma)\le D_{\alpha}\bog(\rho\|\sigma)$ for $\alpha\in(0,1)$, and
$D_{\alpha}(\rho\|\sigma)\ge D_{\alpha}\bog(\rho\|\sigma)$ for $\alpha>1$.
(Vice versa, the inequality $s(\alpha)D_{\alpha}(\rho\|\sigma)\ge s(\alpha)D_{\alpha}\bog(\rho\|\sigma)$
for a fixed $\alpha\in(0,+\infty)\setminus\{1\}$ and every $\rho,\sigma\in\B(\hil)_{++}$ implies the Golden-Thompson
inequality.)
Hence, we are left to prove the first inequality in \eqref{D ordering}.

Let us fix $\rho,\sigma\in\L(\hil)_{++}$ and $\alpha\in(1,+\infty)$, and for every $n\in\bN$, let
$\rho_n:=\rho^{\otimes n},\,\sigma_n:=\sigma^{\otimes n}$. Then
\begin{align*}
Q_{\alpha}\bog(\rho\|\sigma)^n&=Q_{\alpha}\bog(\rho^{\otimes n}\|\sigma^{\otimes n})
=
\Tr e^{\alpha\log\rho_n+(1-\alpha)\log\sigma_n}\\
&\le
\Tr e^{\alpha\log\E_{\sigma_n}(\rho_n)+\alpha\log v(\sigma_n)+(1-\alpha)\log\sigma_n}\\
&=
v(\sigma_n)^{\alpha}\Tr\bz\E_{\sigma_n}(\rho_n)\jz^{\alpha}\sigma_n^{1-\alpha}
=
v(\sigma_n)^{\alpha}Q_{\alpha}\nw\bz \E_{\sigma_n}(\rho_n)\|\sigma_n\jz\\
&\le
v(\sigma_n)^{\alpha}Q_{\alpha}\nw\bz \rho_n\|\sigma_n\jz
=
v(\sigma_n)^{\alpha}Q_{\alpha}\nw\bz \rho\|\sigma\jz^n,
\end{align*}
where the first and the last identities are due to \eqref{additivity}, and the first inequality is due to the pinching inequality
\eqref{pinching inequality}, $\rho_n\le v(\sigma_n)\E_{\sigma_n}(\rho_n)$, the operator monotonicity of the logarithm, and Lemma \ref{lem:trace-mono}.
The equalities in the third line are due to the fact that $\sigma_n$ and $\E_{\sigma_n}(\rho_n)$ commute, and the last inequality
is due to the monotonicity of $D_{\alpha}\nw$ under pinching \cite[Proposition 14]{Renyi_new}.
Taking now the $n$-th root and then the limit $n\to+\infty$, and using that
$v(\sigma_n)\le (n+1)^{d-1}$, we get the desired inequality.
\end{proof}

\begin{rem}\label{rem:strict ordering}
It is known that equality in the inequality $D_{\alpha}\nw(\rho\|\sigma)\le D_{\alpha}(\rho\|\sigma)$ holds if and only if $\alpha=1$ or $\rho$ and $\sigma$ commute with each other
\cite{Hiai-ALT}. It is also easy to see that the other inequalities don't hold with equality in general, either. Indeed,
choosing $\rho:=\half\begin{pmatrix}1 & 1 \\ 1 & 1\end{pmatrix}$ and
$\sigma:=\begin{pmatrix} a & 0 \\ 0 & b\end{pmatrix}$, $a,b>0$, a straightforward computation shows that
\begin{align*}
Q_{\alpha}\bog(\rho\|\sigma)=(ab)^{\frac{1-\alpha}{2}},\ds\ds\ds
Q_{\alpha}\nw(\rho\|\sigma)=\bz\frac{a^{\frac{1-\alpha}{\alpha}}+b^{\frac{1-\alpha}{\alpha} }}{2}\jz^{\alpha},\ds\ds\ds
Q_{\alpha}(\rho\|\sigma)=\frac{a^{1-\alpha}+b^{1-\alpha}}{2},
\end{align*}
which are not equal to each other for general $a$ and $b$.
\end{rem}

All the R\'enyi divergences are strictly positive on pairs of states:
\begin{prop}
For every $\alpha>0$, and all three values of $\typ$,
\begin{align}\label{positivity1}
D_{\alpha}\x(\rho\|\sigma)\ge \log\Tr\rho-\log\Tr\sigma,\ds\ds\ds\ds\ds\ds
\rho,\sigma\in\L(\hil)_+,
\end{align}
with equality if and only if $\rho$ is a constant multiple of $\sigma$,
or equivalently,
\begin{align}\label{positivity3}
D_{\alpha}\x(\rho\|\sigma)\ge 0,\ds\ds\ds\ds\ds\ds
\rho,\sigma\in\S(\hil),
\end{align}
with equality if and only if $\rho=\sigma$.
\end{prop}
\begin{proof}
The equivalence of \eqref{positivity1} and \eqref{positivity3} is immediate from the scaling properties in Lemma \ref{lemma:basic properties}, and hence we only prove \eqref{positivity3}. Moreover, by the ordering in \eqref{D ordering 1} and
the monotonicity in \eqref{mon in alpha}, it is enough to consider $D_{\alpha}\nw$ and $\alpha\in(0,1)$.
By the monotonicity under pinching \cite[Proposition 14]{Renyi_new} and the classical H\"older inequality, we have
$Q_{\alpha}\nw(\rho\|\sigma)\le Q_{\alpha}\nw(\E_{\sigma}\rho\|\sigma)\le\bz\Tr\E_{\sigma}(\rho)\jz^{\alpha}\bz\Tr\sigma\jz^{1-
\alpha}\le 1$, which yields \eqref{positivity3}.
As a consequence, $\psi_{\alpha}\nw(\rho\|\sigma)\le 0$ for every $\alpha\in (0,1)$. Now if
$D_{\alpha}\nw(\rho\|\sigma)=0$ for some $\alpha\in(0,1)$ then
$\psi_{\alpha}\nw(\rho\|\sigma)=0$. By the convexity of $\psi_{\alpha}\nw(\rho\|\sigma)$ in $\alpha$,
(Lemma \ref{lemma: conv mon}), this is only possible if $\psi_{\alpha}\nw(\rho\|\sigma)=0$ for every
$\alpha\in(0,1)$. Hence, $D_{\alpha}\nw(\rho\|\sigma)=0,\,\alpha\in(0,1)$, and taking the limit
$\alpha\nearrow 1$ yields $D(\rho\|\sigma)=0$. By \eqref{Klein} this implies $\rho=\sigma$.
\end{proof}

\begin{remark}
An alternative proof for the case $\xx=\neww$ has been given in \cite[Theorem 5]{Beigi}.
\end{remark}

\begin{lemma}\label{cor:monotonicity}
For every $\rho\in\L(\hil)_+$, we have
\begin{align*}
\sigma'\ge\sigma\ds\imp\ds D_{\alpha}\x(\rho\|\sigma')\le D_{\alpha}\x(\rho\|\sigma)
\end{align*}
for $\xx=\oldd$ and $\alpha\in[0,1]$, for $\xx=\neww$ and $\alpha\in[1/2,+\infty]$, and for
$\xx=\bogg$ and $\alpha\in[0,+\infty]$.
\end{lemma}
\begin{proof}
By \eqref{Renyi ep limit}, we can assume without loss of generality that $\rho,\sigma,\sigma'$ are invertible.
The assertions then follow immeditely from the following:
For $\xx=\oldd$ from the fact that $x\mapsto x^{\alpha},\,x\in(0,+\infty)$, is operator monotone increasing for
$\alpha\in(0,1)$; for $\xx=\neww$ from the fact that $x\mapsto x^{\frac{1-\alpha}{\alpha}}$ is operator monotone
increasing for $\alpha\in[1/2,1)$, and operator monotone decreasing for $\alpha\in(1,+\infty)$, and from \eqref{trace-mono}; and for
$\xx=\bogg$ from the fact that the logarithm is operator monotone increasing.
When necessary, the cases $\alpha=0,1,+\infty$ can be obtained by taking the appropriate limit in $\alpha$.
\end{proof}
\begin{remark}
The case $\xx=\neww$ has been shown in \cite[Proposition 4]{Renyi_new}, in a slightly different way.
\end{remark}
\smallskip

The following technical lemma will play an important role in the minimax arguments used to establish the equivalence of the various definitions of the R\'enyi capacities in Proposition \ref{prop:cap equal}.

Recall the definition of $s(\alpha)$ from \eqref{s alpha}.

\begin{lemma}\label{lemma:ep limit}
For any of the three possible values of $\typ$, define
\begin{align}\label{alpha int}
\A:=\A\nw:=(0,+\infty)\setminus\{1\},\ds\ds\ds
\A\bog:=(1,+\infty).
\end{align}
Let $\rho,\sigma\in\L(\hil)_+$, and let $\alpha\in \A\x$.

\noindent 1. The function
$\ep\mapsto s(\alpha)Q_{\alpha}\x(\rho\|\sigma+\ep I)$ is monotone decreasing in $\ep\in(0,+\infty)$, and
\begin{align}\label{Q limit}
s(\alpha)Q_{\alpha}\x(\rho\|\sigma)
=
\lim_{\ep\searrow 0}s(\alpha)Q_{\alpha}\x(\rho\|\sigma+\ep I)
=
\sup_{\ep>0}s(\alpha)Q_{\alpha}\x(\rho\|\sigma+\ep I).
\end{align}
The same hold for $s(\alpha)\psi_{\alpha}\x$ and $D_{\alpha}\x$ in place of $s(\alpha)Q_{\alpha}\x$, and also for
$D_{\infty}\x$.

\noindent 2. If $\sigma$ is invertible then the function
$\ep\mapsto Q_{\alpha}\x(\rho+\ep I\|\sigma)$ is monotone increasing in $\ep\in(0,+\infty)$, and
\begin{align}\label{Q limit2}
Q_{\alpha}\x(\rho\|\sigma)
=
\lim_{\ep\searrow 0}Q_{\alpha}\x(\rho+\ep I\|\sigma)
=
\inf_{\ep>0}Q_{\alpha}\x(\rho+\ep I\|\sigma).
\end{align}
The same hold for $\psi_{\alpha}\x$ and $s(\alpha)D_{\alpha}\x$ in place of $Q_{\alpha}\x$.
Moreover, these relations are valid also when $\xx=\bogg$ and $\alpha\in(0,1)$.

\noindent 3. We have
\begin{align}\label{Q limit3}
Q_{\alpha}\x(\rho\|\sigma)
=
\lim_{\ep\searrow 0}\lim_{\delta\searrow 0}Q_{\alpha}\x(\rho+\delta I\|\sigma+\ep I)=
\begin{cases}
\inf_{\ep> 0}\inf_{\delta> 0}Q_{\alpha}\x(\rho+\delta I\|\sigma+\ep I),&\alpha\in(0,1),\\
\sup_{\ep> 0}\inf_{\delta> 0}Q_{\alpha}\x(\rho+\delta I\|\sigma+\ep I),&\alpha>1.
\end{cases}
\end{align}
and the same hold for $\psi_{\alpha}\x$ and $s(\alpha)D_{\alpha}\x$ in place of $Q_{\alpha}\x$.
\end{lemma}
\begin{proof}
Note that \eqref{Q limit3} is immediate from \eqref{Q limit} and \eqref{Q limit2}, and
the claims about the monotonicity are trivial to verify, and hence the second identities in \eqref{Q limit} and \eqref{Q limit2} follow if we can prove the first identities.
It is easy to see that for an invertible $\sigma$,
\begin{align*}
\lim_{\ep\searrow 0}Q_{\alpha}\x(\rho+\ep I\|\sigma)
=
\lim_{\ep\searrow 0}Q_{\alpha}\x(\rho+\ep (I-\rho^0)\|\sigma),
\end{align*}
and thus \eqref{Q limit2} follows from \eqref{ext2} and Lemma \ref{lemma:Q def}.

We only prove \eqref{Q limit} for $\xx=\bogg$, as the other cases follow by very similar, and slightly simpler arguments.
Let $P_s$ denote the spectral projection of $\sigma$ corresponding to $s\in\bR$; if $s$ is not an eigenvalue of $\sigma$ then
$P_s=0$. Then
\begin{align*}
\rho^0(\logn(\sigma+\ep I))\rho^0=
\sum_{s>0}\rho^0 P_s\rho^0\log(s+\ep)+\rho^0(I-\sigma^0)\rho^0\log\ep.
\end{align*}
If $\rho^0\le\sigma^0$ then $\rho^0(I-\sigma^0)\rho^0=0$, and \eqref{Q limit} follows trivially.
Assume next that $\rho^0\nleq\sigma^0$. Then there exists
some $c>0$ such that $\rho^0(I-\sigma^0)\rho^0\ge cQ$, where $Q:=(\rho^0(I-\sigma^0)\rho^0)^0\ne 0$,
and $Q\le\rho^0$.
Hence, for every $\ep\in(0,1)$,
\begin{align*}
\rho^0(\logn(\sigma+\ep I))\rho^0
\le
\kappa_{\ep}\rho^0+(\log\ep)cQ,
\end{align*}
where $\kappa_{\ep}:=\max\{0,\log(\norm{\sigma}+\ep)\}$.
Let $\rho_{\min}$ denote the smallest non-zero eigenvalue of $\rho$. By the above,
\begin{align*}
\Tr \rho^0 e^{\alpha\logn\rho+(1-\alpha)\rho^0(\logn(\sigma+\ep I))\rho^0}
&\ge
\Tr \rho^0 e^{\alpha(\log\rho_{\min})\rho^0+(1-\alpha) \kappa_{\ep}\rho^0+(1-\alpha)(\log\ep)cQ}\\
&=
\rho_{\min}^{\alpha}e^{\kappa_{\ep}(1-\alpha)}\Tr\left[\ep^{c(1-\alpha)}Q+\rho^0-Q\right],
\end{align*}
and the last quantity goes to $+\infty=Q_{\alpha}\bog(\rho\|\sigma)$ as $\ep\searrow 0$.

For $\alpha\in\A\x$, the assertions about $\psi_{\alpha}\x$ and $D_{\alpha}\x$ follow trivially from \eqref{Q limit} and \eqref{Q limit2}. Finally, $\ep\mapsto D_{\infty}\x(\rho\|\sigma+\ep I)$ is the pointwise limit of monotone functions, and hence is itself monotone, and
\begin{align}
\lim_{\ep\searrow 0}D_{\infty}\x(\rho\|\sigma+\ep I)
&=
\sup_{\ep> 0}D_{\infty}\x(\rho\|\sigma+\ep I)
=
\sup_{\ep> 0}\sup_{\alpha\in(1,+\infty)}D_{\alpha}\x(\rho\|\sigma+\ep I)\nn
&=
\sup_{\alpha\in(1,+\infty)}\sup_{\ep> 0}D_{\alpha}\x(\rho\|\sigma+\ep I)
=
\sup_{\alpha\in(1,+\infty)}D_{\alpha}\x(\rho\|\sigma)
=
D_{\infty}\x(\rho\|\sigma).
\nonumber
\end{align}
\end{proof}

\begin{cor}\label{cor:lsc}
Let $\typ$ be any of the three possible values, and let
$\alpha\in\A\x$, where $\A\x$ is given in \eqref{alpha int}.
For every $\rho\in\L(\hil)_+$, the function
\begin{align*}
\sigma\mapsto s(\alpha)Q_{\alpha}\x(\rho\|\sigma)\ds\ds\ds\text{is lower semicontinuous on} \ds\L(\hil)_+,
\end{align*}
and the same hold for $s(\alpha)\psi_{\alpha}\x$ and $D_{\alpha}\x$ in place of $s(\alpha)Q_{\alpha}\x$.

For every $\sigma\in\L(\hil)_{++}$, the function
\begin{align*}
\rho\mapsto Q_{\alpha}\x(\rho\|\sigma)\ds\ds\ds\text{is upper semicontinuous on} \ds\L(\hil)_+,
\end{align*}
and the same hold for $\psi_{\alpha}\x$ and $s(\alpha)D_{\alpha}\x$ in place of $Q_{\alpha}\x$.
Moreover, $\sigma\mapsto D_{\infty}\x(\rho\|\sigma)$ is also lower semicontinuous on $\L(\hil)_+$.
The assertions about upper semicontinuity are also valid for $\xx=\bogg$ and $\alpha\in(0,1)$.
\end{cor}
\begin{proof}
Let $\alpha\in\A\x$ be fixed.
For every $\ep>0$, $\sigma\mapsto s(\alpha)Q_{\alpha}\x(\rho\|\sigma+\ep)$ is continuous. Hence, by Lemma \ref{lemma:ep limit},
the function $\sigma\mapsto s(\alpha)Q_{\alpha}\x(\rho\|\sigma)$ is the supremum of continuous functions, and thus is itself lower
semicontinuous. Similarly, if $\sigma\in\L(\hil)_{++}$ then $\rho\mapsto Q_{\alpha}\x(\rho+\ep I\|\sigma)$ is continuous for
every $\ep>0$, and hence, by Lemma \ref{lemma:ep limit}, the function $\rho\mapsto Q_{\alpha}\x(\rho\|\sigma)$ is the infimum
of continuous functions and thus upper semicontinuous. The assertions about $\psi_{\alpha}\x$ and
$D_{\alpha}\x$ follow immediately.
In particular, 
$\sigma\mapsto D_{\infty}\x(\rho\|\sigma)$
is the supremum
of lower semicontinuous functions in $\sigma$, and hence is itself lower semicontinuous.
\end{proof}

\section{R\'enyi capacities}
\label{sec:Renyi capacities}

The celebrated Holevo-Schumacher-Westmoreland theorem \cite{H,SW} states that the
asymptotic classical information transmission capacity of a quantum channel under the constraint of asymptotically vanishing 
error probability coincides with its Holevo capacity; see Section \ref{sec:cl-q sc} for details. Based on results in classical 
information theory, it is natural to expect that a more refined description of the trade-off between the coding rate and the 
decoding error would involve R\'enyi generalizations of the Holevo capacity. In the main result of our paper, Theorem \ref{thm:sc exponent}, we will show that this is indeed the case for the strong converse exponent of classical-quantum channels.
In this section we collect the necessary definitions and technicalities that we will need later in the proof of our main result.

\subsection{R\'enyi mutual informations and equivalent definitions of the R\'enyi capacities}

For a quantum channel $W:\,\X\to\S(\hil)$ and a finitely supported probability distribution $P\in\P_f(\X)$, 
the corresponding 
\ki{Holevo quantity} $\chi(W,P)$ is defined as the mutual information in the classical-quantum state $\ext{W}(P)$, expressable in the following equivalent ways:
\begin{align}
\chi(W,P)
&:=
D(\ext{W}(P)\|P\otimes W(P))
=
\inf_{\sigma\in\S(\hil)}D(\ext{W}(P)\|P\otimes \sigma)\label{Holevo1}\\
&=
\sum_{x\in\X}P(x)D(W(x)\|W(P))=
\inf_{\sigma\in\S(\hil)}\sum_{x\in\X}P(x)D(W(x)\|\sigma).\label{Holevo2}
\end{align}
(Recall that in $P\otimes\sigma$ in \eqref{Holevo1}, $P$ stands for $\sum_{x\in\X}P(x)\pr{x}$, the first marginal of $\ext{W}(P)$ \eqref{marginals}.)
It is also customary to call $\chi(W,P)$ the Holevo quantity of the ensemble of states $\{\rho,P(\rho)\}_{\rho\in\supp P}$.

The \ki{Holevo capacity} $\chi(W)$ of the channel $W$ is then defined as the supremum of all such mutual informations over all finitely supported probability distributions at the input of the channel:
\begin{align}
\chi(W):=\sup_{P\in\P_f(\X)}\chi(W,P)&=
\sup_{P\in\P_f(\X)}\inf_{\sigma\in\S(\hil)}D(\ext{W}(P)\|P\otimes \sigma)\label{Holevo3}\\
&=\sup_{P\in\P_f(\X)}\inf_{\sigma\in\S(\hil)}\sum_{x\in\X}P(x)D(W(x)\|\sigma).\label{Holevo4}
\end{align}

When the relative entropy is replaced with some R\'enyi divergence, the expressions in \eqref{Holevo1}-\eqref{Holevo2} may not coincide anymore, and therefore we have various options to formally define the R\'enyi analogues of the Holevo quantity. We will consider the following two options that will both turn out to be useful for later applications:
\begin{align}
\chi_{\alpha,1}\x(W,P)&:=\inf_{\sigma\in\S(\hil)}D_{\alpha}\x(\ext{W}(P)\|P\otimes \sigma),\ds\ds\ds\ds\ds\ds
&\alpha\in[0,+\infty].\label{chi def}\\
\chi_{\alpha,2}\x(W,P)&:=\inf_{\sigma\in\S(\hil)}\sum_{x\in\X}P(x)D_{\alpha}\x(W(x)\|\sigma),\ds\ds\ds\ds\ds\ds
&\alpha\in[0,+\infty].\label{chi def2}
\end{align}
These are exact analogues of the second formulas in \eqref{Holevo1} and \eqref{Holevo2}, respectively.
Note that \eqref{chi def} is a notion of R\'enyi mutual information in the classical-quantum state $\ext{W}(P)$; see, e.g.~\cite{HT14} for an operational interpretation of this quantity in the context of hypothesis testing (for $\xx=\oldd$ and $\xx=\neww$).
A straightforward computation verifies that for all $\alpha$,
\begin{align*}
Q_{\alpha}\x(\ext{W}(P)\|P\otimes \sigma)
=
\sum_{x\in\X}P(x)Q_{\alpha}\x(W(x)\|\sigma),
\end{align*}
and hence
\begin{align}\label{chi def3}
\chi_{\alpha,1}\x(W,P)=\inf_{\sigma\in\S(\hil)}\frac{1}{\alpha-1}\log \sum_{x\in\X}P(x)Q_{\alpha}\x(W(x)\|\sigma).
\end{align}

We define the \ki{R\'enyi capacities} of a channel $W$, corresponding to the R\'enyi $\alpha$-divergence $D_{\alpha}\x$, as
\begin{align}\label{alpha cap def}
\chi_{\alpha,i}\x(W):=
\sup_{P\in\P_f(\X)}\chi_{\alpha,i}\x(W,P),\ds\ds\ds i=1,2.
\end{align}
Although in general the idenitites in \eqref{Holevo1}--\eqref{Holevo2} do not extend to the R\'enyi quantities with 
$\alpha\ne 1$,
we will show in Proposition \ref{prop:cap equal} below that after the optimization over all finitely supported probability distributions, the possible differences disappear, and we have
\begin{align}\label{cap equality}
\chi_{\alpha}\x(W):=\chi_{\alpha,1}\x(W)=\chi_{\alpha,2}\x(W),
\end{align}
at least for certain pairs of $((t),\alpha)$, including those important for us later. 
Thus, in these cases we can uniquely define the R\'enyi capacity of a channel $W$.
We will show the equality in \eqref{cap equality} by showing that both $\chi_{\alpha,1}\x(W)$ and $\chi_{\alpha,2}\x(W)$
are equal to 
the \ki{R\'enyi divergence radius} of the image of the channel, defined as
\begin{align}\label{radius def}
R_{\alpha}\x(W):=\inf_{\sigma\in\S(\hil)}\sup_{x\in\X}D_{\alpha}\x(W(x)\|\sigma).
\end{align}

The equality \eqref{cap equality} will play an important role in proving our main theorem.
Indeed, following the Dueck-K\"orner argument for classical channels, we naturally get a bound on the strong converse exponent in terms of the second type R\'enyi-Holevo quantities, more precisely,  the
$\chi_{\alpha,2}\bog$ quantities; see Theorems \ref{thm:sc upper}, \ref{thm:achive}, and \ref{thm:conversion}. To convert that to the correct form involving the 
$\chi_{\alpha}\nw$ capacities, we use additivity properties that are only known for the first type 
R\'enyi-Holevo quantities, more precisely for 
$\chi_{\alpha,1}$ and $\chi_{\alpha,1}\nw$; see Lemma \ref{lemma:chi additivity}, Corollary \ref{cor:pinched chi limit},
and the last part of the proof of Theorem \ref{thm:sc achievability}.
The identity in \eqref{cap equality} tells that we can freely switch between the two types of R\'enyi-Holevo quantities after 
optimization over the input probability distributions.

\begin{rem}
Note that for $\alpha\in\A\x$ (defined in \eqref{alpha int}), the infima in \eqref{chi def}, \eqref{chi def2} and \eqref{radius def} can be replaced with minima, due to Corollary \ref{cor:lsc}.
\end{rem}

\begin{prop}\label{prop:cap equal}
We have
\begin{align}
R_{\alpha}\x(W)&=
\sup_{\ep>0}\inf_{\sigma\in\S(\hil)}\sup_{x\in\X}D_{\alpha}\x(W(x)\|\sigma+\ep I)\label{cap eq1}\\
&=\sup_{P\in\P_f(\X)}\chi_{\alpha,2}\x(W,P),\label{cap eq3}\\
&=
\sup_{P\in\P_f(\X)}\chi_{\alpha,1}\x(W,P)\label{cap eq2}\\
&=
\chi_{\alpha}\x(W).\label{cap eq4}
\end{align}
for $\xx=\oldd$ and $\alpha\in(0,2]$, for $\xx=\neww$ and $\alpha\in[1/2,+\infty)$, and for
$\xx=\bogg$ and $\alpha\in(1,+\infty)$. Moreover, the expressions in \eqref{cap eq3} and \eqref{cap eq2}
are also equal to each other for $\xx=\bogg$ and $\alpha\in(0,1)$.
\end{prop}
\begin{proof}
Let us fix a matching pair $\typ$ and $\alpha$ as in the statement of the Theorem. We assume that $\alpha\ne 1$, since that case is already known \cite{OPW,SW}.
By definition,
\begin{align}
R_{\alpha}\x(W)&=
\inf_{\sigma\in\S(\hil)}\sup_{x\in\X}D_{\alpha}\x(W(x)\|\sigma)
=
\inf_{\sigma\in\S(\hil)}\sup_{x\in\X}\sup_{\ep>0}D_{\alpha}\x(W(x)\|\sigma+\ep I)\label{cap eq ep}\\
&=\inf_{\sigma\in\S(\hil)}\sup_{\ep>0}\sup_{x\in\X}D_{\alpha}\x(W(x)\|\sigma+\ep I),\nonumber
\end{align}
where the last expression in \eqref{cap eq ep} follows from Lemma \ref{lemma:ep limit}. Note that
$\ep\mapsto\sup_{x\in\X}D_{\alpha}\x(W(x)\|\sigma+\ep I)$ is monotone decreasing for every
$\sigma\in\S(\hil)$, due to Lemma \ref{lemma:ep limit}. On the other hand, for every $\ep>0$ and
$x\in\X$, $\sigma\mapsto D_{\alpha}\x(W(x)\|\sigma+\ep I)$ is continuous, and hence
$\sigma\mapsto \sup_{x\in\X}D_{\alpha}\x(W(x)\|\sigma+\ep I)$ is lower semi-continuous on the compact set $\S(\hil)$.
Hence, by Lemma \ref{lemma:minimax2},
\begin{align*}
R_{\alpha}\x(W)&=
\sup_{\ep>0}\inf_{\sigma\in\S(\hil)}\sup_{x\in\X}D_{\alpha}\x(W(x)\|\sigma+\ep I),
\end{align*}
proving \eqref{cap eq1}.

By Lemma \ref{lemma:ep limit},
\begin{align*}
\sup_{P\in\P_f(\X)}\chi_{\alpha,1}\x(W,P)=
\frac{1}{\alpha-1}\log s(\alpha)\sup_{P\in\P_f(\X)}\inf_{\sigma\in\S(\hil)}\sup_{\ep>0}
s(\alpha)\sum_{x\in\X}P(x)Q_{\alpha}\x(W(x)\|\sigma+\ep I),
\end{align*}
where $s(\alpha)$ is given in \eqref{s alpha}.
Note that $s(\alpha)\sum_{x\in\X}P(x)Q_{\alpha}\x(W(x)\|\sigma+\ep I)$ is monotone decreasing in $\ep$ and continuous in $\sigma$, and hence, by Lemma \ref{lemma:minimax2},
\begin{align*}
&\sup_{P\in\P_f(\X)}\inf_{\sigma\in\S(\hil)}\sup_{\ep>0}s(\alpha)\sum_{x\in\X}P(x)Q_{\alpha}\x(W(x)\|\sigma+\ep I)\\
&\ds\ds=
\sup_{P\in\P_f(\X)}\sup_{\ep>0}\inf_{\sigma\in\S(\hil)}s(\alpha)\sum_{x\in\X}P(x)Q_{\alpha}\x(W(x)\|\sigma+\ep I)\nn
&\ds\ds=
\sup_{\ep>0}\sup_{P\in\P_f(\X)}\inf_{\sigma\in\S(\hil)}s(\alpha)\sum_{x\in\X}P(x)Q_{\alpha}\x(W(x)\|\sigma+\ep I),
\end{align*}
where the second equality is trivial.
For every $\ep>0$, $s(\alpha)\sum_{x\in\X}P(x)Q_{\alpha}\x(W(x)\|\sigma+\ep I)$ is convex and continuous in $\sigma$
due to Proposition \ref{prop:log convexity}, and it is affine (and thus concave) in $P$. Hence, by Lemma \ref{lemma:KF minimax},
\begin{align*}
\sup_{P\in\P_f(\X)}\inf_{\sigma\in\S(\hil)}s(\alpha)\sum_{x\in\X}P(x)Q_{\alpha}\x(W(x)\|\sigma+\ep I)
&=
\inf_{\sigma\in\S(\hil)}\sup_{P\in\P_f(\X)}s(\alpha)\sum_{x\in\X}P(x)Q_{\alpha}\x(W(x)\|\sigma+\ep I)\\
&=
\inf_{\sigma\in\S(\hil)}\sup_{x\in\X}s(\alpha)Q_{\alpha}\x(W(x)\|\sigma+\ep I),
\end{align*}
where the second equality is trivial. This proves the equality of \eqref{cap eq1} and \eqref{cap eq2},
and the equality of \eqref{cap eq2} and \eqref{cap eq4} is by definition.

Finally, the expression in \eqref{cap eq3} can be written as
\begin{align*}
\sup_{P\in\P_f(\X)}\inf_{\sigma\in\S(\hil)}\sup_{\ep>0}\sum_{x\in\X}P(x)D_{\alpha}\x(W(x)\|\sigma+\ep I)
&=
\sup_{P\in\P_f(\X)}\sup_{\ep>0}\inf_{\sigma\in\S(\hil)}\sum_{x\in\X}P(x)D_{\alpha}\x(W(x)\|\sigma+\ep I)\nn
&=
\sup_{\ep>0}\sup_{P\in\P_f(\X)}\inf_{\sigma\in\S(\hil)}\sum_{x\in\X}P(x)D_{\alpha}\x(W(x)\|\sigma+\ep I).
\end{align*}
where the first expression is due to Lemma \ref{lemma:ep limit}. The second expression follows from
the continuity of $D_{\alpha}\x(W(x)\|\sigma+\ep I)$ in $\sigma$ and its monotonicity in $\ep$,
due to Lemma \ref{lemma:minimax2}. The third expression follows trivially from the second.
Using now the convexity of $D_{\alpha}\x(W(x)\|\sigma+\ep I)$ in $\sigma$, due to Proposition
\ref{prop:log convexity}, and
following the same argument as in the previous paragraph, we see that the last expression above
is equal to \eqref{cap eq1}.
\end{proof}

\begin{rem}
Proofs of some the above identities for various values of $\xx$ and $\alpha$ can be found scattered in the literature; see, e.g.~\cite{OPW,SW,WWY,MH,KW}. The above Proposition unifies and extends all such previous results; the only exception is the equality
$R_{\alpha}\old(W)=\sup_{P\in\P_f(\X)}\chi_{\alpha,1}\old(W,P)$ for $\alpha>2$, that was shown in \cite{KW}, but is not covered
by the above Proposition.
\end{rem}
\begin{rem}
As in the classical case (see, e.g.~\cite{Csiszar}), our proof of Proposition \ref{prop:cap equal} is based on minimax arguments. However, unlike in the proof in 
\cite{Csiszar}, where $\X$ is assumed to be finite, we cannot assume the compactness of $\P_f(\X)$, and hence we need to use the 
compactness of $\S(\hil)$ instead. This poses a technical difficulty, as the R\'enyi divergences can take infinite values when
the reference state is not invertible, in which case the usual minimax theorems may not be applicable. We introduced the 
intermediate step in \eqref{cap eq1} to circumvent this difficulty.
One might also extend the definitions of the R\'enyi-Holevo quantities using arbitrary (not necessarily finitely supported) 
probability measures on the image of the channel and work with the compact (w.r.t.~the weak topology) set of such probability measures, as was done in \cite{MH}. 
However, the problem of infinite values persists in this case, and some step similar to the one in \eqref{cap eq1} is still necessary. Moreover, such a treatment 
requires mathematically more involved arguments, which we could avoid in the above proof.
\end{rem}

\begin{cor}\label{cor:infty cap}
For $\xx=\neww$ and $\xx=\bogg$,
\begin{align}
\chi_{\infty}\x(W)&=
\sup_{P\in\P_f(\X)}\chi_{\infty,1}\x(W,P)
=
\sup_{P\in\P_f(\X)}\chi_{\infty,2}\x(W,P)=
\sup_{\alpha\in(1,+\infty)}\chi_{\alpha}\x(W)
.\label{infty cap eq}
\end{align}
\end{cor}
\begin{proof}
The first equality in \eqref{infty cap eq} is by definition, and
\begin{align*}
\sup_{P\in\P_f(\X)}\chi_{\infty,1}\x(W,P)
&=
\sup_{P\in\P_f(\X)}\sup_{\alpha\in(1,+\infty)}\chi_{\alpha,1}\x(W,P)
=
\sup_{\alpha\in(1,+\infty)}\sup_{P\in\P_f(\X)}\chi_{\alpha,1}\x(W,P)\\
&=
\sup_{\alpha\in(1,+\infty)}\chi_{\alpha}\x(W)\\
&=
\sup_{\alpha\in(1,+\infty)}\sup_{P\in\P_f(\X)}\chi_{\alpha,2}\x(W,P)
=
\sup_{P\in\P_f(\X)}\sup_{\alpha\in(1,+\infty)}\chi_{\alpha,2}\x(W,P)\\
&=
\sup_{P\in\P_f(\X)}\chi_{\infty,2}\x(W,P),
\end{align*}
where we used Proposition \ref{prop:cap equal}.
\end{proof}
\medskip

\begin{lemma}
For all three values of $\xx$, any $P\in\P_f(\X)$, and $i=1,2$, $\alpha\mapsto\chi_{\alpha,i}\x(W,P)$ is monotone increasing on $(0,+\infty]$, and
\begin{align}
\chi(W,P)&=\lim_{\alpha\searrow 1}\chi_{\alpha,i}\x(W,P)=\inf_{\alpha>1}\chi_{\alpha,i}\x(W,P),\label{chi limit1}\\
\chi_{\infty,i}\x(W,P)&=\lim_{\alpha\nearrow +\infty}\chi_{\alpha,i}\x(W,P)=\sup_{1<\alpha<+\infty}\chi_{\alpha,i}\x(W,P).\label{chi limit2}
\end{align}
Similarly, $\alpha\mapsto\chi_{\alpha}\x(W)$ is monotone increasing on $(0,+\infty]$, and
\begin{align}
\chi_{\infty}\x(W)&=\lim_{\alpha\nearrow +\infty}\chi_{\alpha}\x(W)=\sup_{1<\alpha<+\infty}\chi_{\alpha}\x(W),
\label{chi limit4}\\
\chi(W)&=\lim_{\alpha\searrow 1}\chi_{\alpha}\x(W)=\inf_{1<\alpha<+\infty}\chi_{\alpha}\x(W).\label{chi limit5}
\end{align}
\end{lemma}
\begin{proof}
The assertions about the monotonicity are obvious from Lemma \ref{lemma: conv mon}.
Let $f_{\alpha,1}\x(W,P,\sigma):=D_{\alpha}\x(\ext{W}(P)\|P\otimes\sigma)$, and
$f_{\alpha,2}\x(W,P,\sigma):=\sum_{x\in\X}P(x)D_{\alpha}\x(W(x)\|\sigma)$.
By Lemma \ref{lemma:limit at 1} and \eqref{Holevo1}--\eqref{Holevo2},
\begin{align*}
\lim_{\alpha\searrow 1}\chi_{\alpha,i}\x(W,P)&=\inf_{\alpha>1}\chi_{\alpha,i}\x(W,P)
=
\inf_{\alpha>1}\inf_{\sigma\in\S(\hil)}f_{\alpha,i}\x(W,P,\sigma)
=
\inf_{\sigma\in\S(\hil)}\inf_{\alpha>1}f_{\alpha,i}\x(W,P,\sigma)\\
&=
\inf_{\sigma\in\S(\hil)}D(\ext{W}(P)\|P\otimes\sigma)
=
\chi(W,P).
\end{align*}

By Corollary \ref{cor:lsc}, both $f_{\alpha,1}$ and $f_{\alpha,2}$ are lower semicontinuous in $\sigma$ on $\S(\hil)$, and hence, by
Lemma \ref{lemma:minimax2},
\begin{align*}
\lim_{\alpha\nearrow +\infty}\chi_{\alpha,i}\x(W,P)&=\sup_{\alpha>1}\chi_{\alpha,i}\x(W,P)
=
\sup_{\alpha>1}\inf_{\sigma\in\S(\hil)}f_{\alpha,i}\x(W,P,\sigma)
=
\inf_{\sigma\in\S(\hil)}\sup_{\alpha>1}f_{\alpha,i}\x(W,P,\sigma)\\
&=
\inf_{\sigma\in\S(\hil)}
\left\{
\begin{array}{ll}
D_{\infty}\x(\ext{W}(P)\|P\otimes\sigma),&i=1,\\
\sum_{x\in\X}P(x)D_{\infty}\x(W(x)\|\sigma),&i=2,
\end{array}
\right\}
=
\chi_{\infty,i}\x(W,P).
\end{align*}

The identities in \eqref{chi limit4} are immediate from Corollary \ref{cor:infty cap}. Finally,
\eqref{chi limit5} has been proved for $\xx=\oldd$ in \cite[Proposition B.5]{MH} (see also \cite{ON99} for finite $\X$).
Thus, by Proposition \ref{lemma:Q ordering} and the monotonicity of $\alpha\mapsto \chi_{\alpha}\x(W)$, we get
\begin{align*}
\chi(W)&\le\liminf_{\alpha\searrow 1}\chi_{\alpha}\x(W)
\le
\limsup_{\alpha\searrow 1}\chi_{\alpha}\x(W)
\le
\limsup_{\alpha\searrow 1}\chi_{\alpha}\old(W)
=
\chi(W).
\end{align*}
proving \eqref{chi limit5} for $\xx=\neww$ and $\xx=\bogg$.
\end{proof}

\begin{rem}
The limit relation in \eqref{chi limit5} for $\xx=\neww$ has been proved in \cite[Section 8]{WWY} by a different method.
\end{rem}

\subsection{R\'enyi capacities of pinched channels}
\label{sec:pinched}

For the rest of the section, we fix a channel $W:\,\X\to\S(\hil)$. For every $n\in\bN$, let
$\sigma_{u,n}$ be a universal symmetric state on $\hil^{\otimes n}$ as in Lemma \ref{lemma:universal}.
We denote by $\E_n$ the pinching by $\sigma_{u,n}$. If we use the construction from Appendix \ref{sec:universal}
then $\E_n$ can be explicitly written as
\begin{align*}
\E_n(X)=\sum_{\lambda\in Y_{n,d}}(I_{U_{\lambda}}\otimes I_{V_{\lambda}})X(I_{U_{\lambda}}\otimes I_{V_{\lambda}}),
\ds\ds\ds\ds\ds\ds  X\in\L(\hil^{\otimes n}),
\end{align*}
where $d:=\dim\hil$. By the pinching inequality \eqref{pinching inequality} and Lemma \ref{lemma:universal},
\begin{align*}
X\le v(\sigma_{u,n})\,\E_n(X)\le v_{n,d}\,\E_n(X),
\ds\ds\ds\ds\ds\ds  X\in\L(\hil^{\otimes n})_+,
\end{align*}
where $v_{n,d}\le (n+1)^{\frac{(d+2)(d-1)}{2}}$.

For every $n\in\bN$, we define the pinched channel $\E_n W^{\otimes n}:\,\X^n\to\S(\hil^{\otimes n})$ as
\begin{align*}
(\E_n W^{\otimes n})(\sq{x}):=\E_n(W^{\otimes n}(\sq{x})),\ds\ds\ds \sq{x}\in\X^n.
\end{align*}
We use the shorthand notation $\E_n\ext{W}^{\otimes n}$ for its lifted channel, i.e.,
\begin{align*}
(\E_n \ext{W}^{\otimes n})(\sq{x}):=
((\id\otimes\E_n)\ext{W}^{\otimes n})(\sq{x})=
\pr{\sq{x}}\otimes\E_n(W^{\otimes n}(\sq{x})),\ds\ds\ds \sq{x}\in\X^n.
\end{align*}
Our aim in the rest of the section is to relate the $\chi_{\alpha,1}\bog$-quantity for the pinched channel
$\E_n W^{\otimes n}$ to the $\chi_{\alpha,1}\nw$-quantity of the original channel $W^{\otimes n}$. We obtain
such a relation in Corollary \ref{cor:pinched chi limit}, which will be a key technical tool to determine the strong converse exponent
of $W$ in Section \ref{sec:sc}.

We will benefit from the following additivity properties:
\begin{lemma}\label{lemma:chi additivity}
For every $P\in\P_f(\X)$ and every $\alpha>1$,
\begin{align*}
\chi_{\alpha,1}(W^{\otimes n},P^{\otimes n})=n\chi_{\alpha,1}(W,P),\ds\ds\ds\ds\ds
\chi_{\alpha,1}\nw(W^{\otimes n},P^{\otimes n})=n\chi_{\alpha,1}\nw(W,P),\ds\ds n\in\bN.
\end{align*}
\end{lemma}
\begin{proof}
In the case of $\chi_{\alpha,1}$, the unique minimizer state in \eqref{chi def} can be determined explicitly due to the quantum Sibson's
identity, and one can observe that the minimizer for a general $n$ is the $n$-th tensor power of the minimizer for $n=1$;
see, e.g., \cite[Section 4.4]{M13} for details.
The addditivity of $\chi_{\alpha,1}\nw$ is a special case of \cite[Theorem 11]{Beigi}.
\end{proof}

For every $\pi\in\Sym_n$, we denote its natural action on $\X^n$ by the same symbol $\pi$, i.e.,
\begin{align*}
\pi(x_1,\ldots,x_n):=(x_{\pi\inv(1)},\ldots,x_{\pi\inv(n)}),\ds\ds\ds\ds\ds
x_1,\ldots,x_n\in\X.
\end{align*}
We say that a probability density $P_n\in\P(\X^n)$ is symmetric if
$P_n\circ\pi=P_n$ for every $\pi\in\Sym_n$.

\begin{lemma}\label{lem:min-attain}
Let $P_n\in\P(\X^n)$ be a symmetric probability density on $\X^n$.
Then for any $\a>1$ and $\xx=\bogg$ or $\xx=\neww$,
\begin{align}
\chi_{\alpha,1}\x(\E_n W^{\otimes n},P_n)=
\min_{\sigma_n\in\ssymm(\hil^{\otimes n})}D_{\alpha}\x\bz\E_n\ext{W}^{\otimes n}(P_n)\|P_n\otimes\sigma_n\jz,
\label{eq:21}\\
\chi_{\alpha,2}\x(\E_n W^{\otimes n},P_n)=
\min_{\sigma_n\in\ssymm(\hil^{\otimes n})}
\sum_{\sq{x}\in\X^n}P_n(\sq{x})D_{\alpha}\x\bz\E_n W^{\otimes n}(\sq{x})\|\sigma_n\jz,
\label{eq:21-1}
\end{align}
i.e., the minimizations in \eqref{chi def} and \eqref{chi def2} can be restricted to symmetric states.
Moreover, the minima in \eqref{eq:21}--\eqref{eq:21-1} can be replaced with infima over $\ssymm(\hil^{\otimes n})_{++}$, the set of
invertible symmetric states.

The same hold for $\xx=\oldd$ and $\alpha\in(1,2]$.
\end{lemma}
\begin{proof}
Let us fix $\alpha>1$.
We only prove \eqref{eq:21} and for $\xx=\bogg$, as the proofs for the other cases follow completely similar lines.
Thus, with the shorthand notation
\begin{align*}
f(\sigma_n):=Q_{\alpha}\bog(\E_n\ext{W}^{\otimes n}(P_n)\|P_{n}\otimes \sigma_n)
=
\sum_{\sq{x}\in\X^n}P_n(\sq{x})Q_{\alpha}\bog(\E_n(W^{\otimes n}(\sq{x}))\|\sigma_n),
\end{align*}
our aim is to show that
\begin{align}\label{perm inv}
\min_{\sigma_n\in\S(\hil^{\otimes n})}f(\sigma_n)=\min_{\sigma_n\in\ssymm(\hil^{\otimes n})}f(\sigma_n),
\end{align}
which follows immediately if we can show that for any $\sigma_n\in\S(\hil_n^{\otimes n})$,
\begin{align}
f\left(\frac{1}{n!}\sum_{\pi\in\Sym_n}\pi_{\hil}^*\sigma_n\pi_{\hil}\right)\le
\max_{\pi\in\Sym_n}f(\pi_{\hil}^*\sigma_n\pi_{\hil})=
f(\sigma_n).
\label{eq:23}
\end{align}
Note that the minima in \eqref{perm inv} exist because of the lower semicontinuity established in Corollary \ref{cor:lsc},
and the inequality in \eqref{eq:23} follows from the convexity of
$Q_{\alpha}\bog$ in its second argument, Proposition \ref{prop:log convexity}.
Hence, the assertion follows if we can prove the permutation invariance of $f$, i.e., that
$f(\pi_{\hil}^*\sigma_n\pi_{\hil})=f(\sigma_n)$ for any $\sigma_n\in\S(\hil^{\otimes n})$ and any $\pi\in\Sym_n$.

Let us introduce
the shorthand notation $\rho_{\sq{x}}:=\E_n(W^{\otimes n}(\sq{x})),\,\sq{x}\in\X^n$. Then
\begin{align*}
f(\pi_{\hil}^*\sigma_n\pi_{\hil}+\ep I)&=
\sum_{\sq{x}\in\X^n}P_n(\sq{x})
\Tr \rho_{\sq{x}}^0 \,\exp\bz \a\logn \rho_{\sq{x}}+(1-\a)\rho_{\sq{x}}^0(\log(\pi_{\hil}^*(\sigma_n+\ep I)\pi_{\hil})) \rho_{\sq{x}}^0\jz.
\end{align*}
Note that the spectral projections of $\sigma_{u,n}$ commute with all $\pi_{\hil},\,\pi\in\Sym_n$, thus
\begin{align*}
\pi_{\hil}\,\rho_{\sq{x}}\,\pi_{\hil}^*=
\pi_{\hil}\,\E_n(W^{\otimes n}(\sq{x}))\,\pi_{\hil}^*
=\E_n\bz\pi_{\hil}\,W^{\otimes n}(\sq{x})\,\pi_{\hil}^*\jz
=\E_n\left(W^{\otimes n}(\pi(\sq{x}))\right)
=\rho_{\pi(\sq{x})},\ds \pi\in\Sym_n.
\end{align*}
As a consequence,
\begin{align*}
\rho_{\sq{x}}^0(\log(\pi_{\hil}^*(\sigma_n+\ep I)\pi_{\hil})) \rho_{\sq{x}}^0
=
\rho_{\sq{x}}^0\pi_{\hil}^*(\log(\sigma_n+\ep I))\pi_{\hil}\rho_{\sq{x}}^0
=
\pi_{\hil}^*\rho_{\pi(\sq{x})}^0(\log(\sigma_n+\ep I))\rho_{\pi(\sq{x})}^0\pi_{\hil},
\end{align*}
and
\begin{align*}
\logn\rho_{\sq{x}}=\logn(\pi_{\hil}^*\rho_{\pi(\sq{x})}\pi_{\hil})
=
\pi_{\hil}^*(\logn \rho_{\pi(\sq{x})})\pi_{\hil}.
\end{align*}
Putting it together, we get
\begin{align*}
f(\pi_{\hil}^*\sigma_n\pi_{\hil}+\ep I)&=
\sum_{\sq{x}\in\X^n}P_n(\sq{x})
\Tr\rho_{\sq{x}}^0\, \exp\bz\a\pi_{\hil}^*(\logn \rho_{\pi(\sq{x})})\pi_{\hil}+(1-\a) \pi_{\hil}^*\rho_{\pi(\sq{x})}^0(\log(\sigma_n+\ep I))\rho_{\pi(\sq{x})}^0\pi_{\hil}\jz \nn
&=
\sum_{\sq{x}\in\X^n}P_n(\sq{x})
\Tr\pi_{\hil}\rho_{\sq{x}}^0\pi_{\hil}^*\, \exp\bz\a \logn \rho_{\pi(\sq{x})}+(1-\a)\rho_{\pi(\sq{x})}^0(\log(\sigma_n+\ep I))\rho_{\pi(\sq{x})}^0\jz\nn
&=
\sum_{\sq{x}\in\X^n}P_n(\pi(\sq{x}))
\Tr\rho_{\pi(\sq{x})}^0\,
\exp\bz\a \logn \rho_{\pi(\sq{x})}+(1-\a)\rho_{\pi(\sq{x})}^0(\log(\sigma_n+\ep I))\rho_{\pi(\sq{x})}^0\jz\nn
&=
f(\sigma_n+\ep I),
\end{align*}
where in the third line we used that $P_n$ is symmetric.
Taking the limit $\ep\to 0$ proves the desired permutation invariance, due to Lemma \ref{lemma:ep limit}.

To show the assertion that the minimization can be restricted to invertible states, let $\sigma_n$ be a state where the minimum in \eqref{eq:21} is attained, and for
every $\ep\in(0,1)$, let $\sigma_{n,\ep}:=(1-\ep)\sigma_n+\ep(I-\sigma_n^0)/\Tr(I-\sigma_n^0)$.
Note that $\sigma_{n,\ep}$ is symmetric and invertible for every $\ep\in(0,1)$.
Since $\sigma_n$ is a minimizer, we have $\rho_{\sq{x}}^0=\bz\E_n W^{\otimes n}(\sq{x})\jz^0\le \sigma_n^0$ for every
$\sq{x}\in\X^n$ such that $P_n(\sq{x})>0$, and thus
\begin{align*}
\rho_{\sq{x}}^0\bz\logn \sigma_{n,\ep}\jz\rho_{\sq{x}}^0=
\rho_{\sq{x}}^0\bz\logn (1-\ep)\sigma_n\jz\rho_{\sq{x}}^0=
\rho_{\sq{x}}^0\bz\logn \sigma_n\jz\rho_{\sq{x}}^0+\rho_{\sq{x}}^0\log(1-\ep),
\end{align*}
which in turn yields
\begin{align*}
f(\sigma_{n,\ep})=(1-\ep)^{1-\alpha}f(\sigma_n)\xrightarrow[\ep\searrow 0]{}f(\sigma_n).
\end{align*}
Hence,
\begin{align*}
f(\sigma_n)=\inf_{\ep\in(0,1)}f(\sigma_{n,\ep})\ge\inf_{\sigma_n\in\ssymm(\hil^{\otimes n})_{++}}f(\sigma_n),
\end{align*}
and the converse inequality is obvious.
\end{proof}

\begin{lemma}\label{lemma:chi bounds}
For every $P\in\P_f(\X)$, every $\alpha>1$, and $i=1,2$,
\begin{align}
\chi_{\alpha,i}\bog(\E_n W^{\otimes n},P^{\otimes n})
&\ge\begin{cases}
\chi_{\alpha,i}(\E_n W^{\otimes n},P^{\otimes n})-\log v_{n,d},\\
\s\\
\chi_{\alpha,i}\nw(W^{\otimes n},P^{\otimes n})-3\log v_{n,d}.
\end{cases}\label{lb2}
\end{align}
\end{lemma}
\begin{proof}
We prove the assertion only for $i=2$, since the other case is completely similar.
Let $\sigma_n\in\ssymm(\hil^{\otimes n})_{++}$ be an invertible symmetric state. Since $\sigma_{u,n}$ is a universal symmetric state, we have $\sigma_n\le v_{n,d}\sigma_{u,n}$, and thus for every $\sq{x}\in\X^n$,
\begin{align}\label{pinched channel1}
D_{\alpha}\bog(\E_nW^{\otimes n}(\sq{x})\|\sigma_n)\ge
D_{\alpha}\bog(\E_nW^{\otimes n}(\sq{x})\|\sigma_{u,n})-\log v_{n,d}
=
D_{\alpha}(\E_nW^{\otimes n}(\sq{x})\|\sigma_{u,n})-\log v_{n,d},
\end{align}
where the inequality is due to Lemma \ref{cor:monotonicity} and Lemma \ref{lemma:basic properties}, and the
equality is due to the fact that $\E_nW^{\otimes n}(\sq{x})$ and $\sigma_{u,n}$ commute with each other.
Hence,
\begin{align}
\chi_{\alpha,2}\bog(\E_n W^{\otimes n},P^{\otimes n})
&=
\inf_{\sigma_n\in\ssymm(\hil^{\otimes n})_{++}}\sum_{\sq{x}\in\X^n}P^{\otimes n}(\sq{x})D_{\alpha}\bog(\E_nW^{\otimes n}(\sq{x})\|\sigma_n)\nn
&\ge
\sum_{\sq{x}\in\X^n}P^{\otimes n}(\sq{x})D_{\alpha}(\E_nW^{\otimes n}(\sq{x})\|\sigma_{u,n})-\log v_{n,d}
\label{pinched channel2}\\
&\ge
\chi_{\alpha,2}(\E_n W^{\otimes n},P^{\otimes n})-\log v_{n,d},\nonumber
\end{align}
where the first equality is due to Lemma \ref{lem:min-attain}, the first inequality is due to \eqref{pinched channel1},
and the last inequality is due to the definition \eqref{chi def2}. This proves the first bound in \eqref{lb2}.

By \cite[Lemma 2]{HT14}, we have
\begin{align*}
D_{\alpha}(\E_nW^{\otimes n}(\sq{x})\|\sigma_{u,n})\ge D_{\alpha}\nw(W^{\otimes n}(\sq{x})\|\sigma_{u,n})-2\log v_{n,d}.
\end{align*}
Plugging it into \eqref{pinched channel2}, we get
\begin{align*}
\chi_{\alpha,2}\bog(\E_n W^{\otimes n},P^{\otimes n})
&\ge
\sum_{\sq{x}\in\X^n}P^{\otimes n}(\sq{x})D_{\alpha}\nw(W^{\otimes n}(\sq{x})\|\sigma_{u,n})-3\log v_{n,d}
\ge
\chi_{\alpha,2}\nw(W^{\otimes n},P^{\otimes n})-3\log v_{n,d},
\end{align*}
proving the second bound in \eqref{lb2}.
\end{proof}

\begin{cor}\label{cor:pinched chi limit}
For every $P\in\P_f(\X)$ and every $\alpha>1$,
\begin{align}
\lim_{n\to+\infty}\frac{1}{n}\chi_{\alpha,1}\bog(\E_n W^{\otimes n},P^{\otimes n})=
\lim_{n\to+\infty}\frac{1}{n}\chi_{\alpha,1}(\E_n W^{\otimes n},P^{\otimes n})
=\chi_{\alpha,1}\nw(W,P).
\label{eq:26}
\end{align}
\end{cor}
\begin{proof}
We have
\begin{align*}
\chi_{\alpha,1}\nw(W^{\otimes n},P^{\otimes n})-3\log v_{n,d}
\le
\chi_{\alpha,1}\bog(\E_n W^{\otimes n},P^{\otimes n})
\le
\chi_{\alpha,1}\nw(\E_n W^{\otimes n},P^{\otimes n})
\le
\chi_{\alpha,1}\nw(W^{\otimes n},P^{\otimes n}),
\end{align*}
where the first inequality is due to Lemma \ref{lemma:chi bounds}, the second one is due to Proposition \ref{lemma:Q ordering},
and the last one follows from the monotonicity of $D_{\alpha}\nw$ under pinching \cite[Proposition 14]{Renyi_new}.
By Lemma \ref{lemma:chi additivity}, $\chi_{\alpha,1}\nw(W^{\otimes n},P^{\otimes n})=n\chi_{\alpha,1}\nw(W,P)$.
Thus, dividing the above chain of inequalities by $n$, taking the limit $n\to+\infty$, and using \eqref{v limit},
we obtain
\begin{align}\label{pinched chi limit}
\lim_{n\to+\infty}\frac{1}{n}\chi_{\alpha,1}\bog(\E_n W^{\otimes n},P^{\otimes n})=\chi_{\alpha,1}\nw(W,P).
\end{align}
Next, we use
\begin{align*}
\chi_{\alpha,1}\bog(\E_n W^{\otimes n},P^{\otimes n})
\le
\chi_{\alpha,1}(\E_n W^{\otimes n},P^{\otimes n})
\le
\chi_{\alpha,1}\bog(\E_n W^{\otimes n},P^{\otimes n})
+\log v_{n,d},
\end{align*}
where the first inequality is due to Proposition \ref{lemma:Q ordering}, and the second one is due to
Lemma \ref{lemma:chi bounds}. Combining with \eqref{pinched chi limit}, we get
\begin{align*}
\lim_{n\to+\infty}\frac{1}{n}\chi_{\alpha,1}(\E_n W^{\otimes n},P^{\otimes n})=\chi_{\alpha,1}\nw(W,P).
\end{align*}
\end{proof}

\section{The strong converse exponent for classical-quantum channels}
\label{sec:sc}

\subsection{Classical-quantum channel coding and the strong converse exponent}
\label{sec:cl-q sc}

Let $W:\,\X\to\S(\hil)$ be a classical-quantum channel, as described in Section \ref{sec:cq channels}.
The encoding and decoding process of message transmission over the $n$-fold extension of the channel is described as follows.
Each message $k\in\{1,2,\dots,M_n\}$ is encoded to a codeword by an encoder $\phi_n$:
\begin{align*}
\phi_n: k\in\{1,2,\dots,M_n\} \longmapsto \phi_n(k)=x_{k,1},x_{k,2},\dots,x_{k,n}\in\X^n
\end{align*}
and is mapped by $n$ uses of the channel to
\begin{align*}
W^{\otimes n}(\phi_n(k))=W(x_{k,1})\otimes W(x_{k,2})\otimes\dots\otimes W(x_{k,n}) \in \S(\H^{\otimes n}).
\end{align*}
The set $\{\phi_n(k)\}_{k=1}^{M_n}\subset\X^n$ is called a codebook, which is agreed upon by the sender and the receiver in advance.
The decoding process, called the decoder, is described by a POVM $D_n=\{D_n(k)\}_{k=1}^{M_n}$ on $\H^{\otimes n}$,
where the outcomes $1,2,\dots,M_n$ indicate decoded messages.
The pair $\C_n=(\phi_n,D_n)$ is called a code
with cardinality $|\C_n|:=M_n$.

When the message $k$ was sent, the probability of obtaining the outcome $l$ is given by
\begin{align*}
P(l|k) = \Tr W^{\otimes n}(\phi_n(k)) D_n(l).
\end{align*}
The average error probability of the code $\C_n$ is then given by
\begin{align*}
P_e(W^{\otimes n},\C_n) = 1-\frac{1}{M_n}\sum_{k=1}^{M_n}\Tr W^{\otimes n}(\phi_n(k)) D_n(k),
\end{align*}
which is required to vanish asymptotically for reliable communication.
At the same time, the aim of classical-quantum channel coding is to make the transmission rate
$\liminf_{n\to+\infty}\frac{1}{n}\log|\C_n|$ as large as possible.
The channel capacity $C(W)$ is defined as the supremum of achievabile rates
with asymptotically vanishing error probabilities, i.e.,
\begin{align*}
C(W)=\sup\Bigl\{R \Bigm| \exists\{\C_n\}_{n=1}^{\infty} \; \text{such that} \;
\liminf_{n\to\infty}\frac{1}{n}\log|\C_n| \ge R \;\;\text{and}\; \lim_{n\to\infty}P_e(\C_n,W^{\otimes n}) = 0 \Bigr\}.
\end{align*}
According to the Holevo-Schumacher-Westmoreland theorem \cite{H,SW},
\begin{align}
C(W)=\chi(W),
\label{HSW}
\end{align}
where $\chi(W)$ is the Holevo capacity from \eqref{Holevo3}.

By the definition of $C(W)$, \eqref{HSW} means that for any rate $R$ below the Holevo capacity, there exists a sequence of codes
with rate $R$ and asymptotically vanishing error probability. Moreover, it is known that the error probability
can be made to vanish with an exponential speed \cite{Hayashicq}. On the other hand, the strong converse
theorem of classical-quantum channel coding \cite{ON99,W} tells that for any sequence of codes with a rate above
the Holevo capacity, the error probability
inevitably goes to $1$, with an exponential speed, or equivalently, the success probability
\begin{align*}
P_s(W^{\otimes n},\C_n): = 1-P_e(W^{\otimes n},\C_n) 
\end{align*}
decays to zero exponentially fast. The optimal achievable exponent of this decay for a given rate $R$ is called the strong converse exponent $sc(R,W)$:

\begin{definition}[strong converse exponent]
\label{def:sc-exponent}
The success rate $r$ is said to be $R$-achievable, if there exists a sequence of codes $\{\C_n\}_{n=1}^{\infty}$ such that
\begin{align}
\liminf_{n\to\infty}\frac{1}{n}\log|\C_n| \ge R
\quad \text{and} \quad
\liminf_{n\to\infty} \frac{1}{n}\log P_s(\C_n,W^{\otimes n}) \ge -r.
\label{def:achievable-rate}
\end{align}
The strong converse exponent corresponding to the rate $R$ is the infimum of all $R$-achievable rates:
\begin{align*}
sc(R,W)=\inf\Set{r | \text{$r$ is $R$-achievable}}.
\end{align*}
\end{definition}

Alternatively, the strong converse exponent can be expressed as
\begin{align*}
sc(R,W)=\inf\left\{-\liminf_{n\to\infty} \frac{1}{n}\log P_s(\C_n,W^{\otimes n})\,\Big\vert\,\liminf_{n\to\infty}\frac{1}{n}\log|\C_n| \ge R\right\},
\end{align*}
where the infimum is taken over all sequences of codes $\{\C_n\}_{n\in\bN}$. Note that we take the infimum here, as our aim is to make the success probability vanish as slow as possible.
\smallskip

Our main result is the following expression for the strong converse exponent, which is an exact analogue of the
Arimoto-Dueck-K\"orner exponent for classical channels.
\begin{theorem}\label{thm:sc exponent}
Let $W:\,\X\to\S(\hil)$ be a classical-quantum channel. For any rate $R\ge 0$,
\begin{align}\label{main theorem}
sc(R,W)=\sup_{\alpha>1}\frac{\alpha-1}{\alpha}\left\{R-\chi_{\alpha}\nw(W)\right\}=:H_{R,c}\nw(W).
\end{align}
\end{theorem}

The proof follows from Lemma \ref{lemma:Nagaoka bound} and Theorem \ref{thm:sc achievability} below.

\begin{rem}
We remark that the inequality 
$sc(R,W)\ge H_{R,c}\nw(W)$ (the optimality part of the theorem) follows by standard techniques; we explain it in the next section for readers' convenience. Hence, the novelty of Theorem \ref{thm:sc exponent} is the converse inequality $sc(R,W)\le H_{R,c}\nw(W)$
(the so-called achievability part).
\end{rem}

\begin{rem}
In binary state discrimination, the strong converse exponent is given by a similar transform of the R\'enyi divergences,
known as the Hoeffding anti-divergence \cite{Hayashibook,MO}; for two states $\rho$ and $\sigma$ and a rate $R$
it is defined as
\begin{align*}
H_{R,c}\nw(\rho\|\sigma):=
\sup_{\a>1}\frac{\a-1}{\a}\left\{R-D_{\alpha}\nw(\rho\|\sigma)\right\}.
\end{align*}
The expression in \eqref{main theorem} is a direct analogue of this for capacities, which we can also extend to other R\'enyi capacities as
\begin{align}\label{H cap def}
H_{R,c}\x(W):=
\sup_{\a>1}\frac{\a-1}{\a}\left\{R-\chi_{\alpha}\x(W)\right\}.
\end{align}
We call these quantities \ki{converse Hoeffding capacities}.
Theorem \ref{thm:sc exponent} shows that it is the converse Hoeffding capacity corresponding to $D_{\alpha}\nw$
that is operationally relevant for the strong converse of c-q channel coding, (just as in state discrimination),
but in our proof in Section \ref{sec:sc}, the quantity corresponding to
$D_{\alpha}\bog$ also plays an important role.
\end{rem}

\subsection{Lower bound for the strong converse exponent}
\label{sec:sc lb}

Applying the method developed in \cite{N,PV,SW12}
to the new R\'enyi relative entropies, we have the following lemma.

\begin{lemma}\label{lemma:Nagaoka bound}
For any classical-quantum channel $W\in C(\H|\X)$ and $R\ge 0$, we have
\begin{align}\label{Nagaoka bound}
sc(R,W)\ge\sup_{\alpha> 1} \frac{\alpha-1}{\alpha} \left\{R-\chi_{\a}^*(W)\right\}.
\end{align}
\end{lemma}

\begin{proof}
Suppose that $r$ is $R$-achievable.
Then there exists a sequence of codes
$\C_n=(\phi_n,D_n),\,n\in\bN$,
such that \eqref{def:achievable-rate} holds.
Let $\sigma\in\S(\H)$, and define density operators on $\kil_n:=\bigoplus_{k=1}^{M_n}\H^{\otimes n}$ by
\begin{align*}
R_n:=\frac{1}{M_n}\bigoplus_{k=1}^{M_n}W^{\otimes n}(\phi_n(k)),
\quad\quad
S_n:=\frac{1}{M_n}\bigoplus_{k=1}^{M_n}\sigma^{\otimes n},
\end{align*}
where $M_n:=|\C_n|$.
Note that $T_n:=\bigoplus_{k=1}^{M_n}D_n(k)$ and $T_n^{\perp}:=I_{\kil_n}-T_n$ form
a two-valued POVM $\{T_n,T_n^{\perp}\}$ on $\bigoplus_{k=1}^{M_n}\H^{\otimes n}$,
and we have $\Tr R_n T_n=P_s(W^{\otimes n},\C_n)$ and $\Tr S_n T_n=\frac{1}{M_n}$.
With the notation $\phi_n(k)=x_{k,1},x_{k,2},\dots,x_{k,n}$,
the monotonicity of $Q_{\a}^*$ yields  that for $\a\ge 1$,
\begin{align}
P_s(W^{\otimes n},\C_n)^{\a} \frac{1}{M_n^{1-\a}}
&=(\Tr R_n T_n)^{\a}(\Tr S_n T_n)^{1-\a} \nn
&\le(\Tr R_n T_n)^{\a}(\Tr S_n T_n)^{1-\a} + (\Tr R_n T_n^{\perp})^{\a}(\Tr S_n T_n^{\perp})^{1-\a} \nn
&\le Q_{\a}^*(R_n\|S_n) \nn
&=\frac{1}{M_n}\sum_{k=1}^{M_n}\prod_{i=1}^nQ_{\a}^*(W(x_{k,i})\|\sigma) \nn
&\le \left\{\sup_{x\in\X}Q_{\a}^*(W(x)\|\sigma)\right\}^n,
\label{lower1}
\end{align}
where the last equality follows from the multiplicativity of $Q_{\a}^*$.
Since this holds for every $\sigma\in\S(\hil)$, we get
\begin{align*}
\frac{\alpha}{n}\log P_s(W^{\otimes n},\C_n)+\frac{\alpha-1}{n}\log M_n
\le
\inf_{\sigma\in\S(\hil)}\sup_{x\in\X}\log Q_{\a}^*(W(x)\|\sigma),
\end{align*}
or equivalently,
\begin{align*}
\frac{1}{n}\log P_s(W^{\otimes n},\C_n)
\le
-\frac{\alpha-1}{\alpha}\left\{\frac{1}{n}\log M_n
-\inf_{\sigma\in\S(\hil)}\sup_{x\in\X}D_{\a}^*(W(x)\|\sigma)
\right\}.
\end{align*}
Using Proposition \ref{prop:cap equal}, and taking the limsup in $n$, we get
\begin{align*}
-\frac{\alpha-1}{\alpha}\left\{R
-\chi_{\alpha}^*(W)\right\}
\ge
\limsup_{n\to+\infty}\frac{1}{n}\log P_s(W^{\otimes n},\C_n)
\ge
\liminf_{n\to+\infty}\frac{1}{n}\log P_s(W^{\otimes n},\C_n)
\ge -r.
\end{align*}
Since this is true for every $\alpha> 1$, the assertion follows.
\end{proof}

\subsection{Dueck-K\"orner exponent}
\label{sec:DueckKorner}

In this section we show the following weak converse to Lemma \ref{lemma:Nagaoka bound}:
\begin{theorem}\label{thm:sc upper}
For every $R>0$,
\begin{align}\label{DK bound}
sc(R,W)&\le\sup_{\a>1}\frac{\a-1}{\a}\left\{R-\chi_{\alpha}\bog(W)\right\}.
\end{align}
\end{theorem}
This will follow immediately from Theorems \ref{thm:achive} and \ref{thm:conversion} and Lemma \ref{lemma:P optimization}. We give the proof at the end of the section.
\smallskip

Note that for classical channels, the right-hand sides of the bounds in \eqref{Nagaoka bound} and
\eqref{DK bound} coincide, and the two bounds together give Theorem \ref{thm:sc exponent}.
For quantum channels, however, they need not be the same.
In Section \ref{sec:sc achievability} we will combine the bound in
\eqref{DK bound} with block pinching to obtain an upper bound on $sc(R,W)$ that matches \eqref{Nagaoka bound}.
\medskip

Given classical-quantum channels $V,\,W\in C(\H|\X)$ and a finitely supported probability distribution $P\in\P_f(\X)$,
the quantum relative entropy between the outputs of the extended channels, $\V(P)$ and $\W(P)$, is written as
\begin{align*}
D(\V(P)\|\W(P))=\sum_{x\in\X} P(x) D(V(x)\|W(x))=:D(V\|W|P),
\end{align*}
which is called the conditional quantum relative entropy.

For every $P\in\P_f(\X)$ and $R\ge 0$, let
\begin{align}
F(P,R,W):=\inf_{V\in C(\H|\X)} \{ D(V\|W|P)+|R-\chi(V,P)|^+ \},
\label{def:F-quantity}
\end{align}
where
\begin{align*}
|x|^+=\max\{0,x\},\ds\ds\ds x\in\bR.
\end{align*}
Note that for $D(V\|W|P)$ and $\chi(V,P)$, only the values of $V$ and $W$ on the support of $P$ are relevant, and
therefore we can replace $\X$ with $\supp P$ without loss of generality.
Moreover, $D(V\|W|P)=+\infty$ if there exists an $x\in\supp P$ such that $V(x)^0\nleq W(x)^0$.
Hence, we can restrict the infimum to
\begin{align*}
C_{W,P}:=\left\{V\in C(\hil|\supp P):\, V(x)^0\le W(x)^0,\,x\in\supp P\,\right\}.
\end{align*}
Let us introduce the norm
$\norm{F}:=\sum_{x\in\supp P}\norm{F(x)}_1$ on $\{F:\,\supp P\to\L(\hil)\}$. Then $C_{W,P}$ is compact
w.r.t.~this norm, and it is easy to see that
$V\mapsto D(V\|W|P)+|R-\chi(V,P)|^+$ is continuous on $C_{W,P}(\hil|\X)$. Hence, we have
\begin{align*}
F(P,R,W)&=\min_{V\in C_{W,P}} \{ D(V\|W|P)+|R-\chi(V,P)|^+ \}\\
&=\min_{V\in C(\H|\X)} \{ D(V\|W|P)+|R-\chi(V,P)|^+ \}.
\end{align*}
\smallskip

The following is a direct analogue of the Dueck-K\"orner upper bound \cite{DK}:
\begin{theorem}\label{thm:achive}
For any rate $R> 0$, and any $P\in\P_f(\X)$,
\begin{align*}
sc(R,W)\le F(P,R,W).
\end{align*}
\end{theorem}
\begin{proof}
Let
\begin{align}
F_1(P,R,W)&:=\inf_{V:\,\chi(V,P)> R} D(V\|W|P),
\label{eq:29} \\
F_2(P,R,W)&:=\inf_{V:\,\chi(V,P)\le R} \{ D(V\|W|P)+R-\chi(V,P) \}.
\label{eq:30}
\end{align}
Then it is easy to see that
\begin{align*}
F(P,R,W)=\min\{F_1(P,R,W),F_2(P,R,W)\},
\end{align*}
and hence the assertion follows from Lemmas \ref{lem:achive1} and \ref{lem:achive2} below.
\end{proof}
\smallskip

We will need the following two lemmas to prove Lemma \ref{lem:achive1}.
The first one is a key tool in the information spectrum method \cite{BD,NH}, and
is a consequence of the quantum Stein's lemma \cite{HP,ON}.
\begin{lemma}\label{lem:spectrum}
For any states $\rho,\,\sigma\in\S(\H)$, we have
\begin{align*}
\lim_{n\to\infty}\Tr(\rho^{\otimes n}-e^{na}\sigma^{\otimes n})_+ =
\begin{cases}
1, & a<D(\rho\|\sigma), \\
0, & a>D(\rho\|\sigma).
\end{cases}
\end{align*}
\end{lemma}
\smallskip

The second one is the key observation behind the dummy channel technique:
\begin{lemma}\label{lem:code-NP}
Let $W\in C(\H|\X)$ be a classical-quantum channel.
For any code $\C_n=(\phi_n,D_n)$,
any $a\in\R$, and any classical-quantum channel $V\in C(\H|\X)$ (the dummy channel),
we have
\begin{align}
P_s(W^{\otimes n},\C_n) \ge
e^{-na}\left\{ P_s(V^{\otimes n},\C_n)
-\frac{1}{M_n}\sum_{k=1}^{M_n}\Tr\left(V^{\otimes n}(\phi_n(k)) - e^{na}\,W^{\otimes n}(\phi_n(k))\right)_+ \right\}.
\label{eq:22}
\end{align}
\end{lemma}
\begin{proof}
Let $\C_n=(\phi_n,D_n)$ be a code with $|\C_n|=M_n$. According to \eqref{pp trace}, we have
\begin{align*}
\Tr\left(V^{\otimes n}(\phi_n(k)) - e^{na}\,W^{\otimes n}(\phi_n(k))\right)_+
\ge
\Tr\left(V^{\otimes n}(\phi_n(k)) - e^{na}\,W^{\otimes n}(\phi_n(k))\right)D_n(k),
\end{align*}
and hence
\begin{align*}
\Tr W^{\otimes n}(\phi_n(k))D_n(k)
\ge
e^{-na}\Tr V^{\otimes n}(\phi_n(k)) D_n(k)
-e^{-na}
\Tr\left(V^{\otimes n}(\phi_n(k)) - e^{na}\,W^{\otimes n}(\phi_n(k))\right)_+
\end{align*}
for every $k\in\{1,\ldots,M_n\}$.
Summing over $k$ and dividing by $M_n$ yields \eqref{eq:22}.
\end{proof}

\begin{lemma}\label{lem:achive1}
In the setting of Theorem \ref{thm:achive}, we have
\begin{textmath}
sc(R,W)\le F_1(P,R,W).
\end{textmath}
\end{lemma}
\begin{proof}
We show that any rate $r$ satisfying $r>F_1(P,R,W)$ is $R$-achievable.
By the definition \eqref{eq:29}, there exists a classical-quantum channel $V$ such that
\begin{align}
r&>D(V\|W|P),
\label{eq:31}
\\
R&<\chi(V,P).
\label{eq:32}
\end{align}
By Lemma \ref{lem:code-NP}, we have
\begin{align}
P_s(W^{\otimes n},\C_n) \ge
e^{-nr}\left\{ P_s(V^{\otimes n},\C_n)
-\frac{1}{M_n}\sum_{k=1}^{M_n}\Tr\left(V^{\otimes n}(\phi_n(k)) - e^{nr}\,W^{\otimes n}(\phi_n(k))\right)_+ \right\}
\label{eq:33}
\end{align}
for 
any code $\C_n=(\phi_n,D_n)$.
Now we apply the random coding argument and choose codewords
$\phi_n(k)\in\X^n$, $k=1,2,\dots,M_n=\lceil e^{nR}\rceil$,
independently and identically according to $P^{\otimes n}$.
For the decoder, we choose the Hayashi-Nagaoka decoder \cite{HN}.

Let $E\left[ \,\cdot\, \right]$ denote the expectation w.r.t.~the random coding ensemble.
Taking the expectation of both sides of \eqref{eq:33} w.r.t.~$E$, we get
\begin{align*}
E\left[ P_s(W^{\otimes n},\C_n) \right]
&\ge
e^{-nr}\left\{ E\left[ P_s(V^{\otimes n},\C_n) \right]
-\sum_{\sq{x}\in\X^n}P^n(\sq{x})\Tr\left(V^{\otimes n}(\sq{x}) - e^{nr}\,W^{\otimes n}(\sq{x})\right)_+ \right\}\\
&=
e^{-nr}\left\{ E\left[ P_s(V^{\otimes n},\C_n) \right]
-\Tr\bz\ext{V}(P)^{\otimes n}-e^{nr}\ext{W}(P)^{\otimes n}\jz_+\right\}.
\end{align*}
Since $R<\chi(V,P)$, the results of \cite{HN} yield
\begin{align}
\lim_{n\to\infty} E\left[ P_s(V^{\otimes n},\C_n) \right] = 1,
\label{eq:34}
\end{align}
and, since $r>D(V\|W|P)=D(\ext{V}(P)\|\ext{W}(P)$, Lemma \ref{lem:spectrum} yields
\begin{align*}
\lim_{n\to+\infty}\Tr\bz\ext{V}(P)^{\otimes n}-e^{nr}\ext{W}(P)^{\otimes n}\jz_+=0.
\end{align*}
Hence,
\begin{align*}
&\liminf_{n\to\infty}\frac{1}{n}\log \max_{\C_n} P_s(W^{\otimes n},\C_n) \nn
&\ds\ge \liminf_{n\to\infty}\frac{1}{n}\log E\left[ P_s(W^{\otimes n},\C_n) \right] \nn
&\ds\ge -r + \lim_{n\to\infty}\frac{1}{n}\log\left\{ E\left[ P_s(V^{\otimes n},\C_n) \right]
-\sum_{\sq{x}\in\X^n}P^n(\sq{x}) \Tr\left(V^{\otimes n}(\sq{x}) - e^{nr}\,W^{\otimes n}(\sq{x})\right)_+ \right\} \nn
&\ds= -r,
\end{align*}
where the maximum in the first line is taken over all codes with cardinality $|\C_n|=\lceil e^{nR}\rceil$.
Thus, $r$ is $R$-achievable.
\end{proof}

\begin{lemma}\label{lem:achive2}
In the setting of Theorem \ref{thm:achive}, we have
\begin{textmath}
sc(R,W)\le F_2(P,R,W).
\end{textmath}
\end{lemma}

\begin{proof}
First note that if $P$ is supported on one single point $x_0\in\X$ then $\chi(V,P)=0$ for every channel $V$, and hence
$F_2(P,R,W)=R$. It is easy to see that the trivial encoding $\phi_n$ that assigns
$\phi_n(k)_i:=x_0,\,i\in\{1,\ldots,n\}$, to every message $k\in\{1,\ldots,\lceil e^{nR}\rceil\}$, together with any decoding yields a code $\C_n$ with transmission rate $R$ and success rate $R$, showing that
$sc(R,W)\le R=F_2(P,R,W)$. Hence for the rest we will assume that $|\supp P|>2$.

To prove the assertion,
it is enough to show that any rate $r$ satisfying $r>F_2(P,R,W)$ is $R$-achievable.
By the definition \eqref{eq:30}, there exists a classical-quantum channel $V$ such that
\begin{align}
r&>D(V\|W|P)+R-\chi(V,P),
\label{eq:37}
\\
R&\ge \chi(V,P).
\label{eq:38}
\end{align}
Using the assumption that $|\supp P|>2$, and the continuity of $V\mapsto \chi(V,P)$ and $V\mapsto D(V\|W|P)$, we can assume that
$\chi(V,P)>0$, while \eqref{eq:37} and \eqref{eq:38} still hold.
Let $\delta>0$ be such that $r>D(V\|W|P)+R-\chi(V,P)+\delta$ and
$R_1:=\chi(V,P)-\delta>0$. Then we have
\begin{align*}
r_1:=r-R+\chi(V,P)-\delta &> D(V\|W|P),
\\
R_1 &< \chi(V,P).
\end{align*}
Since \eqref{eq:31} and \eqref{eq:32} are satisfied for $r_1$ and $R_1$,
we can see that $r_1$ is $R_1$-achievable from Lemma \ref{lem:achive1}, i.e.,
there exists a sequence of codes $\Psi_{n}=(\psi_n,Y_n)$ such that
\begin{align*}
\liminf_{n\to\infty}\frac{1}{n}\log|\Psi_n| \ge -R_1,
\quad \liminf_{n\to\infty}\frac{1}{n}\log P_s(W^{\otimes n},\Psi_n) \ge -r_1.
\end{align*}
Let $N_n=\lceil e^{nR_1} \rceil$ and $M_n=\lceil e^{nR}\rceil$. Since $N_n\le M_n$ holds,
we can expand the code $\Psi_n=(\psi_n,Y_n)$ to construct a code $\C_n=(\phi_n,D_n)$ with the rate $R$ by
\begin{align*}
\phi_n(k):=
\begin{cases}
\psi_n(k) & (1\le k\le N_n) \\
\psi_n(1) & (N_n< k\le M_n)
\end{cases}, \quad\quad\quad
D_n(k:)=
\begin{cases}
Y_n(k) & (1\le k\le N_n) \\
0 & (N_n< k\le M_n).
\end{cases}
\end{align*}
Then we have
\begin{align*}
P_s(W^{\otimes n},\C_n)
=\frac{1}{M_n} \sum_{k=1}^{M_n} \Tr W^{\otimes n}(\phi_n(k))D_n(k)
=\frac{N_n}{M_n} \frac{1}{N_n} \sum_{k=1}^{N_n} \Tr W^{\otimes n}(\psi_n(k))Y_n(k)
=\frac{N_n}{M_n} P_s(W^{\otimes n},\Psi_n),
\end{align*}
and hence,
\begin{align*}
\liminf_{n\to\infty}\frac{1}{n}\log P_s(W^{\otimes n},\C_n)
&\ge \liminf_{n\to\infty}\frac{1}{n}\log P_s(W^{\otimes n},\Psi_n)+\lim_{n\to\infty}\frac{1}{n}\log \frac{N_n}{M_n} \nn
&\ge -r_1 + R_1 - R \nn
&= -(r - R+\chi(V,P)-\delta) + \chi(V,P)-\delta -R \nn
&= -r,
\end{align*}
proving that $r$ is $R$-achievable.
\end{proof}
\medskip

Our next step is deriving
another representation for $F(P,R,W)$ defined in \eqref{def:F-quantity}.

\begin{theorem}\label{thm:conversion}
Given $W\in C(\H|\X)$ and $R\ge 0$, for any $P\in\P_f(\X)$, we have
\begin{align}\label{div sphere rep1}
F(P,R,W)&=
\sup_{\a>1}\frac{\a-1}{\a}\left\{R-\chi_{\alpha,2}\bog(W,P)\right\}.
\end{align}
\end{theorem}

\begin{proof}
Since
\begin{textmath}
\sup_{0<\delta< 1}\delta(a-b)=|a-b|^+
\end{textmath}
holds for any $a,b\in\R$, we have
\begin{align}
F(P,R,W)
&=\min_{V\in C_{W,P}} \{ D(V\|W|P)+|R-\chi(V,P)|^+ \} \nn
&=\min_{V\in C_{W,P}}\sup_{0<\delta< 1}\{ D(V\|W|P)+\delta\{R-\chi(V,P)\} \}
\label{eq:40}.
\end{align}
Note that $\chi(V,P)=\inf_{\sigma\in\S(\H)_{++}}D(V\|\sigma|P)$, where $D(V\|\sigma|P)$ denotes, with a slight abuse of notation, the
conditional relative entropy of $V$ with respect to the constant classical-quantum channel $x\in\X\mapsto\sigma\in\S(\H)$.
Thus, we can rewrite \eqref{eq:40} as
\begin{align}
F(P,R,W)
&=\min_{V\in C_{W,P}}\sup_{0<\delta< 1}\sup_{\sigma\in\S(\H)_{++}} G(P,\delta,\sigma,V),\label{eq:43}
\end{align}
where for every $\delta\in[0,1]$, $V\in C_{W,P}$ and $\sigma\in\S(\H)_{++}$,
\begin{align}
G(P,\delta,\sigma,V):=&D(V\|W|P)+\delta\{R-D(V\|\sigma|P)\}\label{eq:41}\\
=&\delta R - (1-\delta)\sum_{x\in\X}P(x)H(V(x))\nn
&-\sum_{x\in\X}P(x)\Tr V(x)\logn W(x)+\delta\sum_{x\in\X}P(x)\Tr V(x)\log\sigma,\label{G expanded}
\end{align}
In the above, $H(\rho):=-\Tr\rho\log\rho$ stands for the von Neumann entropy of a state
$\rho\in\S(\hil)$.  Moreover,
\begin{align}
\sup_{\sigma\in\S(\H)_{++}}G(P,\delta,\sigma,V)=&D(V\|W|P)+\delta\{R-\chi(V,P)\}
\label{eq:42}\\
=&
\sum_{x\in\X}P(x)D(V(x)\|W(x))+\delta R
-\delta H\bz\sum_{x\in\X}P(x)V(x)\jz+\delta\sum_{x\in\X}P(x)H(V(x)).\label{eq:42-1}
\end{align}

The following properties of $G$ are easy to verify:
\begin{enumerate}[(i)]
\item\label{i} $G(P,\delta,\sigma,V)$ is concave and continuous with respect to $\sigma\in\S(\H)_{++}$,
\item\label{ii} $G(P,\delta,\sigma,V)$ is convex and continuous with respect to $V\in C_{W,P}$,
\item\label{iii} \dm{\sup_{\sigma\in\S(\H)_{++}} G(P,\delta,\sigma,V)} is convex and continuous with respect to $V\in C_{W,P}$,
\item\label{iv} \dm{\sup_{\sigma\in\S(\H)_{++}} G(P,\delta,\sigma,V)} is affine with respect to $0< \delta< 1$.
\end{enumerate}
Indeed, the claim (\ref{iv}) is obvious from \eqref{eq:42}, and (\ref{i}) is immediate
from \eqref{G expanded} due to the operator concavity of the logarithm.
Also by \eqref{G expanded} and the concavity of the von Neumann entropy,
we get (\ref{ii}). The convexity property in (\ref{iii}) is obvious from (\ref{ii}), and the continuity is clear from
\eqref{eq:42-1}.

Applying now the minimax theorem in Lemma \ref{lemma:KF minimax} first to $\sup_{\sigma\in\S(\H)_{++}} G(P,\delta,\sigma,V)$
and then to $G(P,\delta,\sigma,V)$, and benefiting both times from the compactness of $C_{W,P}$, \eqref{eq:43} can be rewritten as
\begin{align}
F(P,R,W)
&=\min_{V\in C_{W,P}}\sup_{0<\delta< 1}\sup_{\sigma\in\S(\H)_{++}} G(P,\delta,\sigma,V)
\nn
&=\sup_{0<\delta< 1}\min_{V\in C_{W,P}}\sup_{\sigma\in\S(\H)_{++}} G(P,\delta,\sigma,V)
\label{eq:45} \\
&=\sup_{0<\delta< 1}\sup_{\sigma\in\S(\H)_{++}}\min_{V\in C_{W,P}} G(P,\delta,\sigma,V).
\label{eq:46}
\end{align}
Moreover, in each line the supremum over $\delta\in(0,1)$ can be replaced with supremum over $\delta\in[0,1]$.

Note that by \eqref{D bog var} and \eqref{infty variational} we have
\begin{align*}
\min_{\tau\in S_{\rho}(\hil)}\left\{D(\tau\|\rho)-\delta D(\tau\|\sigma)\right\}
=
\begin{cases}
0,& \delta = 0,\\
-\delta D_{\frac{1}{1-\delta}}\bog(\rho\|\sigma),&\delta\in(0,1),\\
-D_{\infty}\bog(\rho\|\sigma),&\delta=1
\end{cases}
\end{align*}
for any $\rho,\sigma$ such that $\rho^0\wedge\sigma^0\ne 0$. Hence,
\begin{align}
\min_{V\in C_{W,P}} G(P,\delta,\sigma,V)&=
\delta R+\sum_{x\in\X}P(x)\min_{V(x)\in S_{W(x)}(\hil)}\left\{D(V(x)\|W(x))-\delta D(V(x)\|\sigma)\right\}\nonumber\\
&=
\begin{cases}
\delta R,& \delta = 0,\\
\delta R-\delta \sum_{x\in\X}P(x)D_{\frac{1}{1-\delta}}\bog(W(x)\|\sigma),&\delta\in(0,1),\\
\delta R-\sum_{x\in\X}P(x)D_{\infty}\bog(W(x)\|\sigma),&\delta=1.
\end{cases}\label{G2 repr}
\end{align}

Now let us introduce a new parameter $\a:=\frac{1}{1-\delta}$,
for which the interval $0<\delta< 1$ corresponds to $\a> 1$. Then
\eqref{eq:46} can be  rewritten as
\begin{align}
F(P,R,W)&=\sup_{\a>1}\sup_{\sigma\in\S(\H)_{++}}\min_{V\in C_{W,P}} G(P,\tfrac{\a-1}{\a},\sigma,V)\label{F1}\\
&=\sup_{\a>1}\frac{\a-1}{\a}\left\{R-\inf_{\sigma\in\S(\H)_{++}}\sum_{x\in\X}P(x)D_{\alpha}\bog(W(x)\|\sigma)\right\}\nonumber\\
&=\sup_{\a>1}\frac{\a-1}{\a}\left\{R-\chi_{\alpha,2}\bog(W,P)\right\},\nonumber
\end{align}
as required.
\end{proof}
\smallskip

By Theorems \ref{thm:achive} and \ref{thm:conversion}, we have
\begin{align}\label{DK infimum}
sc(R,W)\le\inf_{P\in\P_f(\X)}F(P,R,W)=\inf_{P\in\P_f(\X)}\sup_{\a>1}\frac{\a-1}{\a}\left\{R-\chi_{\alpha,2}\bog(W,P)\right\}.
\end{align}
To arrive at the expression in Theorem \ref{thm:sc exponent}, we need the following minimax-type lemma (for $i=2$):

\begin{lemma}\label{lemma:P optimization}
For every $R\ge 0$,
\begin{align}
H_{R,c}\bog(W)=
\sup_{\a>1}\inf_{P\in\P_f(\X)}\frac{\a-1}{\a}\left\{R-\chi_{\alpha,i}\bog(W,P)\right\}
=
\inf_{P\in\P_f(\X)}\sup_{\a>1}\frac{\a-1}{\a}\left\{R-\chi_{\alpha,i}\bog(W,P)\right\},
\label{DK opt2}
\end{align}
where the equalities hold for both $i=1$ and $i=2$.
\end{lemma}
\begin{proof}
The first identity is due to the definition \eqref{H cap def}, so our aim is to show that the infimum and the supremum
in \eqref{DK opt2} can be interchanged.
Let us introduce the notation $\delta=(\alpha-1)/\alpha,\,\alpha\in(1,+\infty)$, and
for every $P\in\P_f(\X), \sigma\in\S(\hil)_{++}$ and $\delta\in[0,1)$, let
\begin{align}
G_{1}(P,\delta,\sigma)&:=\delta R-\delta D_{\frac{1}{1-\delta}}\bog(\ext{W}(P)\|P\otimes\sigma),\label{G1 def}\\
G_{2}(P,\delta,\sigma)&:=\delta R-\delta\sum_{x\in\X}P(x)D_{\frac{1}{1-\delta}}\bog(W(x)\|\sigma).\label{G2 def}
\end{align}
We define
\begin{align*}
G_i(P,1,\sigma):=\lim_{\delta\nearrow 1}G_i(P,\delta,\sigma)
=
\begin{cases}
R-D_{\infty}\bog(\ext{W}(P)\|P\otimes\sigma),&i=1,\\
R-\sum_{x\in\X}P(x)D_{\infty}\bog(W(x)\|\sigma),&i=2.
\end{cases}
\end{align*}
Let
\begin{align}\label{Gi def}
G_i(P,\delta):=
\sup_{\sigma\in\S(\hil)_{++}}G_{i}(P,\delta,\sigma)
=
\begin{cases}
\delta R-\delta\chi_{\frac{1}{1-\delta},i}\bog(W,P),&\delta\in[0,1),\\
R-\chi_{\infty,i}\bog(W,P),&\delta=1.
\end{cases}
\end{align}
Then \eqref{DK opt2} is equivalent to
\begin{align}\label{DK opt3}
\sup_{\delta\in(0,1)}\inf_{P\in\P_f(\X)}G_i(P,\delta)
=
\inf_{P\in\P_f(\X)}\sup_{\delta\in(0,1)}G_i(P,\delta).
\end{align}
Note that $\delta\mapsto G_i(P,\delta)$ is continuous on $[0,1]$,
and
\begin{align*}
\lim_{\delta\searrow 0}\inf_{P\in\P_f(\X)}G_i(P,\delta)
=
\lim_{\delta\searrow 0}\left\{\delta R-\delta\chi_{\frac{1}{1-\delta},i}\bog(W)\right\}=0
=\inf_{P\in\P_f(\X)}G_{i}(P,0),\\
\lim_{\delta\nearrow 1}\inf_{P\in\P_f(\X)}G_i(P,\delta)
=
\lim_{\delta\nearrow 1}\left\{\delta R-\delta\chi_{\frac{1}{1-\delta},i}\bog(W)\right\}
=
R-\chi_{\infty}\bog(W)=
\inf_{P\in\P_f(\X)}G_{i}(P,1).
\end{align*}
Hence, \eqref{DK opt3} can be rewritten as
\begin{align}\label{DK opt4}
\sup_{\delta\in[0,1]}\inf_{P\in\P_f(\X)}G_i(P,\delta)
=
\inf_{P\in\P_f(\X)}\sup_{\delta\in[0,1]}G_i(P,\delta).
\end{align}
This will follow from Lemma \ref{lemma:KF minimax} if we can show that $G_i(P,\delta)$ is convex in $P\in\P_f(\X)$
and concave and upper semi-continuous in $\delta\in[0,1]$.
We will only give a proof for the case $i=2$, since this is what we need for our main result, and because the proof for $i=1$ follows very similar lines.

By \eqref{G2 def} and \eqref{Gi def}, $G_2(P,\delta)$ is the supremum of convex and continuous functions in $P$, and hence is itself convex and lower semi-continuous for every $\delta\in[0,1]$.

By \eqref{G2 repr} and \eqref{G2 def}, we have 
\begin{align*}
G_2(P,\delta,\sigma)=\min_{V\in C_{W,P}} G(P,\delta,\sigma,V).
\end{align*}
Using properties \eqref{i} and \eqref{ii} in the proof of Theorem \ref{thm:conversion}, and Lemma \ref{lemma:KF minimax}, we get
\begin{align*}
G_2(P,\delta)
=
\sup_{\sigma\in\S(\hil)_{++}}\min_{V\in C_{W,P}} G(P,\delta,\sigma,V)
=
\min_{V\in C_{W,P}}\sup_{\sigma\in\S(\hil)_{++}} G(P,\delta,\sigma,V).
%
\end{align*}
Using \eqref{eq:42},
\begin{align*}
G_2(P,\delta)
=\min_{V\in C_{W,P}}\left\{D(V\|W|P)+\delta\{R-\chi(V,P)\}\right\},
\end{align*}
and hence $G_2(P,\delta)$ is the infimum of affine functions in $\delta$, and therefore is itself
concave and upper semi-continuous. This finishes the proof.
\end{proof}

\begin{proofof}[Theorem \ref{thm:sc upper}]
We have
\begin{align}
sc(R,W)\le\inf_{P\in\P_f(\X)}F(P,R,W)
&=
\inf_{P\in\P_f(\X)}\sup_{\a>1}\frac{\a-1}{\a}\left\{R-\chi_{\alpha,2}\bog(W,P)\right\}\\
&=
\sup_{\a>1}\frac{\a-1}{\a}\left\{R-\chi_{\alpha}\bog(W)\right\}
\end{align}
where the first line is due to \eqref{DK infimum}, and
the second line follows by Lemma \ref{lemma:P optimization}.
\end{proofof}

\subsection{Strong converse exponent: achievability}
\label{sec:sc achievability}

Here we prove the following converse to Lemma \ref{lemma:Nagaoka bound}:
\begin{theorem}\label{thm:sc achievability}
Let $W:\,\X\to\S(\hil)$ be a classical-quantum channel.
For any $R> 0$, we have
\begin{align}\label{sc upper bound}
sc(R,W)\le\sup_{\alpha> 1} \frac{\alpha-1}{\alpha} \left\{R-\chi_{\a}^*(W)\right\}.
\end{align}
\end{theorem}
\begin{proof}
For every $m\in\bN$, define the pinched channel
\begin{align*}
W_m(\sq{x}):=(\E_m W^{\otimes m})(\sq{x})=\E_m(W(x_1)\otimes\ldots\otimes W(x_m)),\ds\ds\ds
\sq{x}=(x_1,\ldots,x_m)\in\X^m,
\end{align*}
where $\E_m=\E_{\sigma_{u,m}}$ is the pinching by a universal symmetric state $\sigma_{u,m}$.
By Theorem \ref{thm:sc upper} and \eqref{DK opt2}, for every $R>0$, there exists a sequence of codes
$\C\m_k=(\phi_k\m,\D_k\m),\,k\in\bN$, such that
\begin{align}
\liminf_{k\to+\infty}\frac{1}{k}\log|\C\m_k|&\ge Rm,\label{rate}\\
\liminf_{k\to+\infty}\frac{1}{k}\log P_s\bz W_m^{\otimes k},\C\m_k\jz
&\ge
-\sup_{\alpha>1}\frac{\alpha-1}{\alpha}
\left\{Rm-\chi_{\alpha}\bog(\E_m W^{\otimes m})\right\}.\label{rate2}
\end{align}

Note that we can assume that the elements of the decoding POVM $\D\m_k=\{D\m_k(i):\,i=1,\ldots,|\C\m_k|\}$
are invariant under $\E_m^{\otimes k}$, i.e.,
\begin{align}\label{invariance}
\E_m^{\otimes k} (D\m_k(i))=D\m_k(i)\ds\ds\ds\forall i\s\forall k.
\end{align}
 Indeed,
\begin{align*}
P_s(W_m^{\otimes k},\C\m_k)
&=
\frac{1}{|\C\m_k|}\sum_{i=1}^{|\C\m_k|}\Tr W_m^{\otimes k}(\phi\m_k(i))D\m_k(i)\\
&=
\frac{1}{|\C\m_k|}\sum_{i=1}^{|\C\m_k|}\Tr \E_m^{\otimes k} \bz W^{\otimes mk}(\phi\m_k(i))\jz D\m_k(i)\\
&=
\frac{1}{|\C\m_k|}\sum_{i=1}^{|\C\m_k|}\Tr \E_m^{\otimes k} \bz W^{\otimes mk}(\phi\m_k(i))\jz \E_m^{\otimes k} (D\m_k(i))\\
&=
\frac{1}{|\C\m_k|}\sum_{i=1}^{|\C\m_k|}\Tr W_m^{\otimes k}(\phi\m_k(i)) \E_m^{\otimes k} (D\m_k(i)),
\end{align*}
i.e., the success probability does not change if we replace $\{D\m_k(i):\,i=1,\ldots,|\C\m_k|\}$
with $\{\E_m^{\otimes k}(D\m_k(i)):\,i=1,\ldots,|\C\m_k|\}$.

Now, from the above code $\C\m_k$, we construct a code
$\C_n$ for $W^{\otimes n}$ for every $n\in\bN$. For a given $n\in\bN$, let $k\in\bN$ be such that
$km\le n<(k+1)m$ (we assume that $m>1$).
For every $i\in\{1,\ldots,|\C\m_k|\}$, define $\phi_n(i)$ as any continuation of $\phi\m_k(i)$ in $\X^n$,
and define the decoding POVM elements as
$D_n(i):=D\m_k(i)\otimes I^{\otimes (n-mk)}$.
Then $|\C_n|=|\C\m_k|$, and hence, by \eqref{rate}
\begin{align}\label{error rate bound1}
\liminf_{n\to+\infty}\frac{1}{n}\log|\C_n|=\frac{1}{m}\liminf_{k\to+\infty}\frac{1}{k}\log|\C\m_k|\ge R.
\end{align}
Moreover,
\begin{align*}
P_s(W^{\otimes n},\C_n)&=
\frac{1}{|\C_n|}\sum_{i=1}^{|\C_n|}\Tr W^{\otimes n}(\phi_n(i))\bz D\m_k(i)\otimes I^{\otimes (n-mk)}\jz\\
&=
\frac{1}{|\C\m_k|}\sum_{i=1}^{|\C\m_k|}\Tr W^{\otimes km}(\phi\m_k(i)) D\m_k(i)\\
&=
\frac{1}{|\C\m_k|}\sum_{i=1}^{|\C\m_k|}\Tr W^{\otimes km}(\phi\m_k(i))\E_m^{\otimes k}( D\m_k(i))\\
&=
\frac{1}{|\C\m_k|}\sum_{i=1}^{|\C\m_k}\Tr\E_m^{\otimes k}( W^{\otimes km}(\phi\m_k(i)))D\m_k(i)\\
&=
P_s(W_m^{\otimes k},\C\m_k),
\end{align*}
where in the third line we used \eqref{invariance}. Hence, by \eqref{rate2},
\begin{align}
\liminf_{n\to+\infty}\frac{1}{n}\log P_s(W^{\otimes n},\C_n)&=
\frac{1}{m}\liminf_{k\to+\infty}\frac{1}{k}\log P_s\bz W_m^{\otimes k},\C\m_k\jz\nn
&\ge
-\sup_{\alpha>1}\frac{\alpha-1}{\alpha}
\left\{R-\frac{1}{m}\chi_{\alpha}\bog(W_m)\right\}.\label{sc rate bound2}
\end{align}

Now we have
\begin{align*}
\chi_{\alpha}\bog(W_m)
&=
\sup_{P_m\in\P_f(\X^m)}\chi_{\alpha,1}\bog(\E_m W^{\otimes n},P_m)
\ge
\sup_{P\in\P_f(\X)}\chi_{\alpha,1}\bog(\E_m W^{\otimes n},P^{\otimes m})\\
&\ge
\sup_{P\in\P_f(\X)}\chi_{\alpha,1}\nw(W^{\otimes n},P^{\otimes m})-3\log v_{m,d}
=
m\sup_{P\in\P_f(\X)}\chi_{\alpha,1}\nw(W,P)-3\log v_{m,d}\\
&=
m\chi_{\alpha}\nw(W)-3\log v_{m,d},
\end{align*}
where the first equality is by definition, the first inequality is trivial, the second inequality is due to
Lemma \ref{lemma:chi bounds},
the following identity is due to Lemma \ref{lemma:chi additivity}, and the last identity is again by definition. Thus, by \eqref{sc rate bound2},
\begin{align*}
\liminf_{n\to+\infty}\frac{1}{n}\log P_s(W^{\otimes n},\C_n)
&\ge
-\sup_{\alpha>1}\frac{\alpha-1}{\alpha}
\left\{R-\chi_{\alpha}\nw(W)\right\}-3\frac{1}{m}\log v_{m,d}.
\end{align*}
Taking the limit $m\to+\infty$ and using \eqref{v limit}, we obtain
\begin{align}
\liminf_{n\to+\infty}\frac{1}{n}\log P_s(W^{\otimes n},\C_n)
&\ge
-\sup_{\alpha>1}\frac{\alpha-1}{\alpha}\left\{R-\chi_{\alpha}\nw(W)\right\}.
\end{align}
Combining this with \eqref{error rate bound1}, the assertion follows.
\end{proof}

\section{Application to quantum channels}
\label{sec:qc}

By a quantum channel we mean a completely positive trace-preserving (CPTP) linear map from $\B(\hilin)$ to $\B(\hilout)$,
where $\hilin,\hilout$ are finite-dimensional Hilbert spaces. The $n$-fold product extension of a quantum channel
$\map:\,\B(\hilin)\to\B(\hilout)$ is
the unique linear extension of the map $\map^{\otimes n}:\,X_1\otimes\ldots X_n\mapsto \map(X_1)\otimes \ldots\otimes \map(X_n)$ to $\B(\hilin^{\otimes n})$. Obviously, any quantum channel $\map:\,\B(\hilin)\to\B(\hilout)$
can be viewed as a classical-quantum channel with $\X=\hilin$, $\hil=\hilout$, and $W(x):=\map(x),\,x\in\B(\hilin)$.
However, the $n$-fold product extension of a quantum channel is different from its $n$-fold product extension as a classical-quantum channel, as the latter can only take product inputs, while the former can take entangled inputs.
The coding process and the definition of the strong converse exponent $sc(R,\map)$ remains the same as described in Section \ref{sec:cl-q sc}, with the exception that the codewords $\phi_n(k)$ can be any density operators in $\hilin^{\otimes n}$.
By Theorem \ref{thm:sc achievability}, for any rate $R$, there exists a sequence of codes $\C_n,\,n\in\bN$, with product codewords such that
\begin{align*}
\liminf_{n\to\infty}\frac{1}{n}\log|\C_n| \ge R
\quad \text{and} \quad
\liminf_{n\to\infty} \frac{1}{n}\log P_s(\C_n,\map^{\otimes n})
\ge
-\sup_{\alpha>1}\frac{\alpha-1}{\alpha}\left\{R-\chi_{\a}^*(\map)\right\}.
\end{align*}
Hence,
\begin{align*}
sc(R,\map)\le \sup_{\alpha>1}\frac{\alpha-1}{\alpha}\left\{R-\chi_{\a}^*(\map)\right\}.
\end{align*}

On the other hand, the same argument as in Lemma \ref{lemma:Nagaoka bound} yields that for any sequence of codes with rate
$R$,
\begin{align*}
\limsup_{n\to+\infty}\frac{1}{n}\log P_s(\map^{\otimes n},\C_n)\le-
\sup_{\alpha>1}\frac{\alpha-1}{\alpha}\left\{R-\limsup_{n\to+\infty}\frac{1}{n}\chi_{\a}^*(\map^{\otimes n})\right\}.
\end{align*}
Hence, we arrive at the following
\begin{theorem}\label{thm:qchannel sc bounds}
For any quantum channel $\map$ and any rate $R>0$,
\begin{align*}
\sup_{\alpha>1}\frac{\alpha-1}{\alpha}\left\{R-\limsup_{n\to+\infty}\frac{1}{n}\chi_{\a}^*(\map^{\otimes n})\right\}
\le
sc(R,\map)
\le
\sup_{\alpha>1}\frac{\alpha-1}{\alpha}\left\{R-\chi_{\a}^*(\map)\right\}.
\end{align*}
In particular, if $\limsup_{n\to+\infty}\frac{1}{n}\chi_{\a}^*(\map^{\otimes n})\le \chi_{\a}^*(\map)$ for every
$\alpha\in(1,+\infty)$ then
\begin{align}\label{general sc exp}
sc(R,\map)=\sup_{\alpha>1}\frac{\alpha-1}{\alpha}\left\{R-\chi_{\a}^*(\map)\right\}.
\end{align}
\end{theorem}

\subsection{Entanglement breaking channels}

It has been shown in \cite[Theorem 18]{WWY} that if $\map$ is entanglement breaking then
$\chi_{\a}^*(\map^{\otimes n})\le n\chi_{\a}^*(\map)$ for every $n\in\bN$ and $\alpha\in(1,2]$.
The range of $\alpha$ for which this subadditivity property holds was limited to $(1,2]$ in \cite{WWY}
because the equality $\chi_{\alpha}^*(\map)=R_{\alpha}(\map)$ was only showed for this parameter range
(in \cite[Lemma 14]{WWY}), which in turn was due to the fact that convexity of
$\sigma\mapsto Q_{\alpha}\nw(\rho\|\sigma)$ for a fixed $\rho$ was only known for this parameter range.
Having Propositions \ref{prop:log convexity} and \ref{prop:cap equal} at our disposal, we can conclude that
$\chi_{\a}^*(\map^{\otimes n})\le n\chi_{\a}^*(\map)$ holds for every $n\in\bN$ and $\alpha\in(1,+\infty)$,
and thus the strong converse exponent of an entanglement breaking channel $\map$ is given by \eqref{general sc exp}.

\subsection{Covariant channels}

We say that a quantum channel $\map:\,\B(\hilin)\to\B(\hilout)$ is \ki{group covariant} if there exists a compact group $G$
and continuous unitary representations $U\inn$ and $U\out$ on $\hilin$ and $\hilout$, respectively, such that
\begin{itemize}
\item[(i)]
$U\out$ is irreducible, and
\item[(ii)]
$\map(U\inn(g)X U\inn(g)^*)=U\out(g)\map(X)U\out(g)^*$ for all $g\in G$ and all $X\in\B(\hilin)$.
\end{itemize}

For a density operator $\sigma$, let
\begin{align*}
H_{\alpha}(\sigma):=\frac{1}{1-\alpha}\log\Tr\sigma^{\alpha}
\end{align*}
denote its $\alpha$-entropy. The minimum output $\alpha$-entropy $H_{\alpha}^{\min}(\map)$ of a channel
$\map:\,\B(\hilin)\to\B(\hilout)$ is defined as
\begin{align*}
H_{\alpha}^{\min}(\map):=\min_{\rho\in\S(\hilin)}H_{\alpha}(\map(\rho)).
\end{align*}
We say that a channel $\map$ belongs to the KW-class (after \cite{KW}) if it is group covariant, and
has additive minimum output $\alpha$-entropy for all $\alpha\in(1,+\infty)$, i.e.,
\begin{itemize}
\item[(iii)]
$H_{\alpha}^{\min}(\map^{\otimes n})=nH_{\alpha}^{\min}(\map),\ds\ds\ds n\in\bN,\ds\alpha\in(1,+\infty)$.
\end{itemize}
Typical examples for channels in the KW-class are the depolarizing channels \cite{King} and unital qubit channels \cite{King2}.

It has been shown in \cite{KW} that all channels in the KW-class have the strong converse property for classical information
transmission, i.e., for all sequences of codes with rate above the Holevo capacity, the success probability goes to zero with the
number of channel uses. Note that the properties (i)--(iii) imply that the classical information transmission capacity of a channel in the KW-class is equal to its single-shot Holevo capacity, according to the Holevo-Schumacher-Westmoreland theorem
\cite{H,SW}. Moreover, the results of \cite{KW} yield the following bound on the strong converse exponent of any channel
$\map$ in the KW-class:
\begin{align*}
sc(R,\map)\ge \sup_{\alpha>1}\frac{\alpha-1}{\alpha}\left\{R-\chi_{\a}(\map)\right\}.
\end{align*}
Note that in the above expression we have $\chi_{\alpha}\old(\map)$, which can be strictly larger than $\chi_{\alpha}\nw(\map)$.

It has been shown in \cite[Section 8.1]{WWY} that for any quantum channel $\map:\,\B(\hilin)\to\B(\hilout)$,
\begin{align*}
\chi_{\a}^*(\map^{\otimes n})\le n\log\dim\hilout-H_{\alpha}^{\min}(\map^{\otimes n}).
\end{align*}
In particular, for channels satisfying (iii) above, we have
\begin{align}\label{WWY upper bound}
\chi_{\a}^*(\map^{\otimes n})\le n\left[\log\dim\hilout-H_{\alpha}^{\min}(\map)\right].
\end{align}
The following is an analogue of \cite[Lemma 1.3]{KW}. While we follow the main idea of the proof of \cite[Lemma 1.3]{KW},
some technical details are different; in particular, Proposition \ref{prop:log convexity} plays a crucial role.
\begin{lemma}\label{lemma:KW cap}
For any group covariant channel $\map:\,\B(\hilin)\to\B(\hilout)$, and any $\alpha\in[1/2,+\infty)$, we have
\begin{align*}
\chi_{\alpha}\nw(\map)=\log\dim\hilout-H_{\alpha}^{\min}(\map).
\end{align*}
\end{lemma}
\begin{proof}
Let $\map:\,\B(\hilin)\to\B(\hilout)$ be a group covariant channel (with covariace group $G$), and $\alpha\in[1/2,+\infty)$ be fixed.
By Proposition \ref{prop:cap equal}
\begin{align*}
\chi_{\alpha}^*(\map)&=
\inf_{\sigma\in\S(\hilout)}\sup_{\rho\in S(\hilin)}D_{\alpha}\nw(\map(\rho)\|\sigma)
\le
\sup_{\rho\in S(\hilin)}D_{\alpha}\nw\bz\map(\rho)\Big\|\frac{I_{\hilout}}{\dim \hilout}\jz
=
\log\dim \hilout-H_{\alpha}^{\min}(\map).
\end{align*}

On the other hand, for any $\rho\in\S(\hilin),\,\sigma\in\S(\hilout)$, and any $g\in G$, we have
\begin{align*}
\sup_{\rho\in\S(\hilin)}D_{\alpha}\nw(\map(\rho)\|\sigma)
&\ge
D_{\alpha}\nw(\map(U\inn(g)\rho U\inn(g)^*)\|\sigma)
=
D_{\alpha}\nw(U\out(g)\map(\rho) U\out(g)^*\|\sigma)\\
&=
D_{\alpha}\nw(\map(\rho)\| U\out(g)^*\sigma U\out(g)).
\end{align*}
Integrating both sides with respect to the normalized Haar measure $\mu$ on $G$, we get
\begin{align*}
\sup_{\rho'\in\S(\hilin)}D_{\alpha}\nw(\map(\rho')\|\sigma)
&\ge
\int_G D_{\alpha}\nw(\map(\rho)\| U\out(g)^*\sigma U\out(g))\,d\mu
\ge
D_{\alpha}\nw\bz\map(\rho)\Bigg\| \int_G U\out(g)^*\sigma U\out(g)\,d\mu\jz\\
&=
D_{\alpha}\nw\bz\map(\rho)\Big\|\frac{I_{\hilout}}{\dim \hilout} \jz=\log\dim\hilout-H_{\alpha}(\map(\rho)),
\end{align*}
where the second inequality is due to Proposition \ref{prop:log convexity}, and in the second line we used the irreducibility of
$U\out$. Taking first the infimum over $\sigma\in\S(\hilout)$ and then the supremum over $\rho\in\S(\hilin)$,
we get
\begin{align*}
\chi_{\alpha}^*(\map)&\ge\log\dim \hilout-H_{\alpha}^{\min}(\map).
\end{align*}
\end{proof}

Lemma \ref{lemma:KW cap} and \eqref{WWY upper bound} yield immediately the following
\begin{cor}\label{cor:KW cap2}
For any channel $\map$ in the KW-class, we have
\begin{align*}
\chi_{\alpha}\nw(\map^{\otimes n})\le n\chi_{\alpha}\nw(\map),\ds\ds\ds n\in\bN,\ds\alpha\in[1/2,+\infty).
\end{align*}
\end{cor}

Finally, Corollary \ref{cor:KW cap2} and Theorem \ref{thm:qchannel sc bounds} yield
\begin{theorem}
For any quantum channel $\map$ in the KW-class, and any rate $R>0$,
\begin{align*}
sc(R,\map)
&=
\sup_{\alpha>1}\frac{\alpha-1}{\alpha}\left\{R-\chi_{\a}^*(\map)\right\}\\
&=
\sup_{\alpha>1}\frac{\alpha-1}{\alpha}\left\{R-\log\dim\hilout+H_{\alpha}^{\min}(\map)\right\}.
\end{align*}
\end{theorem}

\appendix

\section{Universal symmetric states}
\label{sec:universal}

In this section we review the construction of a universal symmetric state that was given in \cite{universalcq} (see Lemma \ref{lemma:universal}).

For every $n\in\bN$, let $\mu_{\hil,n}$ be the $n$-th tensor power representation of the identical representation of
$\SU(\hil)$ on $\hil$, i.e., $\mu_{\hil,n}: \SU(\H)\ni A \longmapsto A^{\otimes n}$, and let
$\L_{\mu_{\hil,n}}(\hil^{\otimes n}):=\{A\in\L(\hil^{\otimes n}):\,\mu_{\hil,n}(U)A=A\mu_{\hil,n}(U)\ds\forall U\in\SU(\hil)\}$
be the commutant algebra of the representation.
According to the Schur-Weyl duality (see, e.g., \cite[Chapter 9]{GoodW}),
$\symm(\hil^{\otimes n})$ and $\L_{\mu_{\hil,n}}(\hil^{\otimes n})$
are each other's commutants, i.e.,
\begin{align*}
\symm(\hil^{\otimes n})=\L_{\mu_{\hil,n}}(\hil^{\otimes n})'=\mu_{\hil,n}(\SU(\hil))'',\ds\ds\ds
\L_{\mu_{\hil,n}}(\hil^{\otimes n})=\symm(\hil^{\otimes n})'=\{\pi_{\hil}|\,\pi\in\Sym_n\}''.
\end{align*}
(Note that the double commutant is equal to the algebra generated by the given set.)
Moreover, $\hil^{\otimes n}$ decomposes as
\begin{align*}
\H^{\otimes n}\simeq\bigoplus_{\lambda\in Y_{n,d}} U_{\lambda}\otimes V_{\lambda},
\end{align*}
where $U_{\lambda}$ and $V_{\lambda}$ carry irreducible representations of $\Sym_n$ and $\SU(\hil)$, respectively,
and we have
\begin{align*}
\symm(\hil^{\otimes n})=\bigoplus_{\lambda\in Y_{n,d}}\End(U_{\lambda})\otimes I_{V_{\lambda}},\ds\ds\ds
\L_{\mu_{\hil,n}}(\hil^{\otimes n})=\bigoplus_{\lambda\in Y_{n,d}}I_{U_{\lambda}}\otimes\End(V_{\lambda}).
\end{align*}
Here, $Y_{n,d}$ is the set of
the Young diagrams up to the depth $d:=\dim\hil$, defined as
\begin{align*}
Y_{n,d}=\Set{ \lambda=(n_1,n_2,\dots,n_d) | n_1\ge n_2\ge \dots\ge n_d\ge 0,\;\sum_{i=1}^n n_i=n }.
\end{align*}

In particular, every permutation invariant state $\rho_n\in\symm(\hil^{\otimes n})$ can be written according to the above decomposition as
\begin{align*}
\rho_n=\bigoplus_{\lambda\in Y_{n,d}} p_{\lambda} \cdot \rho_{\lambda}\otimes \frac{I_{V_{\lambda}}}{\dim V_{\lambda}},
\end{align*}
where $\rho_{\lambda}$ is a density operator acting on $U_{\lambda}$ and $\{p_{\lambda}\}$ is a probability function on $Y_{n,d}$.
Using inequalities $\rho_{\lambda}\le I_{U_{\lambda}}$ and $p_{\lambda}\le 1$, we have
\begin{align*}
\rho_n
\le \bigoplus_{\lambda\in Y_{n,d}} \frac{1}{\dim V_{\lambda}} \cdot I_{U_{\lambda}}\otimes I_{V_{\lambda}}
\le \max_{\lambda}\{\dim U_{\lambda}\}\cdot|Y_{n,d}|\cdot\sigma_{u,n},
\end{align*}
where
\begin{align*}
\sigma_{u,n} := \bigoplus_{\lambda\in Y_{n,d}}\frac{1}{|Y_{n,d}|}\cdot
\frac{I_{U_{\lambda}}}{\dim U_{\lambda}}\otimes\frac{I_{V_{\lambda}}}{\dim V_{\lambda}}.
\end{align*}
It is known that
\begin{align*}
\max_{\lambda}\{\dim U_{\lambda}\} \le (n+1)^{\frac{d(d-1)}{2}},
\quad
|Y_{n,d}|\le (n+1)^{d-1},
\end{align*}
and hence $\sigma_{u,n}$ satisfies the criteria in Lemma \ref{lemma:universal}.

\section{The auxiliary function}
\label{sec:aux}

In this section we prove that various versions of the so-called auxiliary function are concave. This was not needed
in the main body of the paper, but since the proof follows very naturally from some simple considerations in the paper, and the 
problem itself is interesting for information theory, we decided to include a brief discussion here. For more on the background, 
we refer to the recent paper \cite{Cheng-Hsieh}.

Recall the definitions of the $Q_{\alpha}$ quantities for classical probability distributions, and the 
$Q_{\alpha}\old,\,Q_{\alpha}\nw,\,Q_{\alpha}\bog$ quantities for positive semidefinite operators, given in Section \ref{sec: Renyi def}. For the rest, let $Q_{\alpha}\x$ be any function on pairs of positive semidefinite operators, such that it reduces to the classical $Q_{\alpha}$ for commuting operators, i.e., 
for any non-zero non-negative functions $p,q\in\bR_+^{\X}\setminus\{0\}$ on some finite set $\X$, any orthonormal system
$\{\ket{x}\}_{x\in\X}$ in some Hilbert space, and any $\alpha\in(0,+\infty)\setminus\{1\}$,
\begin{align*}
Q_{\alpha}\x\bz\sum_{x\in\X}p(x)\pr{x}\Big\|\sum_{x\in\X}q(x)\pr{x}\jz=Q_{\alpha}(p\|q),
\end{align*}
where the latter is the classical $Q_{\alpha}$ quantity of $p$ and $q$.
This is satisfied by $Q_{\alpha}\x$ for $\xx=\oldd$, $\xx=\neww$ and $\xx=\bogg$.
For any $\rho,\sigma\in\B(\hil)_+$, and any $\alpha\in(0,+\infty)\setminus\{1\}$, 
define
\begin{align}
\psi_{\alpha}\x(\rho\|\sigma):=\log Q_{\alpha}\x(\rho\|\sigma),\ds\ds\ds
D_{\alpha}\x(\rho\|\sigma):=\frac{1}{\alpha-1}\log Q_{\alpha}\x(\rho\|\sigma)-\frac{1}{\alpha-1}\log\Tr\rho
\end{align}
as in \eqref{def:psi} and \eqref{Renyi div def}, and let $\psi_{1}\x(\rho\|\sigma):=\log\Tr\rho$.
%
For a map $W:\,\X\to\B(\hil)_+$, where $\hil$ is a finite-dimensional Hilbert space, 
and for any finitely supported probability distribution $P\in\P_f(\X)$, let
\begin{align*}
\chi_{1,\alpha}\x(W,P)&:=\inf_{\sigma\in\S(\hil)}
D_{\alpha}\x\bz\sum_x P(x)\pr{x}\otimes W_x\Big\|\sum_x P(x)\pr{x}\otimes \sigma\jz\\
\chi_{2,\alpha}\x(W,P)&:=\inf_{\sigma\in\S(\hil)}
\sum_x P(x)D_{\alpha}\x\bz W_x\|\sigma\jz.
\end{align*} 
When $W$ is a classical-quantum channel (i.e., all $W(x)$ have unit trace), the above quantities are exactly the generalized 
Holevo quantities given in \eqref{chi def} and \eqref{chi def2}, at least for the three $(t)$ values considered there.
We define the corresponding \ki{auxiliary functions} as
\begin{align*}
E_{0,i}\x(s,W,P):=-\frac{\alpha-1}{\alpha}\chi_{i,\alpha}\x(W,p)\Bigg\vert_{\alpha=\frac{1}{1+s}},
\end{align*}
$i=1,2$. Note that with $\alpha=\frac{1}{1+s}$, we have
\begin{align}
E_{0,1}\x(s,W,P)&=
-\sup_{\sigma\in\S(\hil)}\frac{1}{\alpha}\psi_{\alpha}\x(\ext{W}(P)\|P\otimes \sigma)
=
-\sup_{\sigma\in\S(\hil)}(1+s)\psi_{\frac{1}{1+s}}\x(\ext{W}(P)\|P\otimes \sigma)\label{auxiliary2}\\
E_{0,2}\x(s,W,p)&=
-\sup_{\sigma\in\S(\hil)}\sum_{x\in\X}P(x)\frac{1}{\alpha}\psi_{\alpha}\x(W(x)\|\sigma)
=
-\sup_{\sigma\in\S(\hil)}\sum_{x\in\X}P(x)(1+s)\psi_{\frac{1}{1+s}}\x\bz W(x)\|\sigma\jz\label{auxiliary3}
\end{align}
for $\alpha\in(0,1)$, or equivalently, $s>0$, and the same formulas hold with minima instead of the maxima
for $\alpha>1$, or equivalently, $s\in(-1,0)$.

\begin{prop}\label{prop:aux conc}
Assume that $Q\x$ is such that for any $\rho,\sigma\in\B(\hil)_+$, the map
$\alpha\mapsto \psi_{\alpha}\x(\rho\|\sigma)$ is convex on $(0,1)$. Then for any 
map $W:\,\X\to\B(\hil)$, and any $P\in\P_f(\X)$, the maps
$s\mapsto E_{0,i}\x(s,W,P)$ are concave on $(0,+\infty)$ for $i=1,2$.
\end{prop}
\begin{proof}
Applying Lemma \ref{lemma:convex transformation} to $f(\alpha):=\psi\x_{\alpha}(.,.)$ with the affine function 
$\vfi(s):=1+s$, we get that $s\mapsto (1+s)\psi\x_{\frac{1}{1+s}}(.,.)$ is convex, and 
by Lemma \ref{lemma:inf sup convexity},
taking the supremum over 
$\sigma$ in \eqref{auxiliary2}, \eqref{auxiliary3}, preserves this convexity.
\end{proof}

\begin{cor}\label{cor:aux conc}
By Lemma \ref{lemma: conv mon}, $Q_{\alpha}\x$ with $\xx=\oldd,\,\xx=\neww,\,\xx=\bogg$ satisfy the convexity 
assumption of Proposition \ref{prop:aux conc}, and hence the corresponding auxiliary functions are concave.
\end{cor}

The case $\xx=\oldd$ is special in the sense that there is an explicit expression for $E_{0,1}\x(s,W,P)$, due to the Sibson identity \cite{Sibson,KW}. Indeed, one can easily see that 
\begin{align*}
&D_{\alpha}\bz\sum_x P(x)\pr{x}\otimes W_x\Big\|\sum_x P(x)\pr{x}\otimes \sigma\jz\\
&=
\frac{1}{\alpha-1}\log\sum_x P(x)\Tr W_x^{\alpha}\sigma^{1-\alpha}
=
\frac{1}{\alpha-1}\log\Tr \omega(\alpha,P)^{\alpha}\sigma^{1-\alpha}
+
\frac{\alpha}{\alpha-1}\log\Tr\bz\sum_x P(x)W_x^{\alpha}\jz^{1/\alpha}
\end{align*}
where $\omega(\alpha,P):=\bz\sum_x P(x)W_x^{\alpha}\jz^{1/\alpha}/\Tr\bz\sum_x P(x)W_x^{\alpha}\jz^{1/\alpha}$.
By the strict positivity of $D_{\alpha}$ on states, we get 
\begin{align}\label{Sibson}
\min_{\sigma\in\S(\hil)}
D_{\alpha}\bz\sum_x P(x)\pr{x}\otimes W_x\Big\|\sum_x P(x)\pr{x}\otimes \sigma\jz
=
\frac{\alpha}{\alpha-1}\log\Tr\bz\sum_x P(x)W_x^{\alpha}\jz^{1/\alpha},
\end{align}
or equivalently, with $\alpha=\frac{1}{1+s}$,
\begin{align*}
E_{0,1}\old(s,W,P)
=
-\log\Tr\bz\sum_x P(x)W_x^{\alpha}\jz^{1/\alpha}
=
-\log\Tr\bz\sum_x P(x)W_x^{\frac{1}{1+s}}\jz^{1+s}.
\end{align*}

By Corollary \ref{cor:aux conc}, we have the following:
\begin{cor}
For any map $W:\,\X\to\B(\hil)_+$, and any finitely supported probability distribution $P\in\P_f(\X)$, the map
$s\mapsto -\log\Tr\bz\sum_x P(x)W_x^{\frac{1}{1+s}}\jz^{1+s}$ is concave.
\end{cor}

\begin{rem}
The same result was proved very recently in \cite{Cheng-Hsieh} by completely different methods, using the properties of certain operator means. Our proof is considerably simpler, and Proposition \ref{prop:aux conc} and Corollary \ref{cor:aux conc}
also give extensions to auxiliary functions defined from other R\'enyi divergences, and for the value $i=2$.
\end{rem}

Unfortunately, the proof method of Proposition \ref{prop:aux conc} does not work for $\alpha>1$, equivalently, for 
$s\in(-1,0)$; this is due to the fact that in this case we have infima in \eqref{auxiliary2} and \eqref{auxiliary3} instead of suprema, and in general, the infimum of convex functions need not be convex. Note, however, that the proof of 
Lemma \ref{lemma:P optimization} yields the following:
\begin{prop}
For any map $W:\,\X\to\B(\hil)_+$, and any finitely supported probability distribution $P\in\P_f(\X)$, the maps
$s\mapsto E_{0,i}\bog\old(s,W,P)$ are concave on $(-1,0)$ for $i=1,2$.
\end{prop}

%
%
%
%
%

\section*{Acknowledgments}

The authors are grateful to Marco Tomamichel, Mark M.~Wilde and Andreas Winter for comments on the paper; in particular, for pointing out that our results for classical-quantum channels have applications for the strong converse exponent of certain classes of quantum channels.
The authors are also grateful to an anonymous referee for various suggestions that helped to improve the presentation of the paper.
This work was partially supported by the MEXT Grant-in-Aid
(A) No.~20686026 ``Project on Multi-user Quantum Network'' (TO),
and by
the European Research Council Advanced Grant ``IRQUAT'', the
Spanish MINECO  Project No. FIS2013-40627-P, the Generalitat de Catalunya CIRIT Project No. 2014 SGR 966, 
the Hungarian Research Grant OTKA-NKFI K104206, and
by the Technische Universit\"at M\"unchen -- Institute for Advanced Study, funded by the German Excellence Initiative and the European Union Seventh Framework Programme under grant agreement no. 291763 (MM). Part of this work was done while MM was with the 
F\'{\i}sica Te\`{o}rica: Informaci\'{o} i Fenomens Qu\`{a}ntics,
Universitat Aut\`{o}noma de Barcelona, 
and later with the Zentrum Mathematik, M5, and the Institute for Advanced Study, Technische Universit\"at M\"unchen.

\bibliography{bibliography}

\end{document}